\documentclass[showpacs,twocolumn,pra,longbibliography,superscriptaddress,notitlepage]{revtex4-2}
\usepackage{qcircuit}
\usepackage[dvips]{graphicx}
\usepackage{amsmath,amssymb,amsthm,mathrsfs,amsfonts,dsfont}
\usepackage{mathtools}
\usepackage{tensor}
\usepackage{caption,subcaption}
\usepackage{epsfig}
\usepackage{braket}
\usepackage{bm}
\usepackage{bbm}
\usepackage{enumerate}
\usepackage{color}
\usepackage[normalem]{ulem}
\usepackage{comment}
\usepackage[colorlinks = true]{hyperref}

\usepackage{makecell} 

\hypersetup{colorlinks=true, linkcolor=blue, citecolor=blue, urlcolor=blue}

\usepackage{enumitem}

\usepackage{algorithm}  
\usepackage{algpseudocode}  

\graphicspath{{./figure/}}

\allowdisplaybreaks

\newtheorem{proposition}{Proposition}

\newtheorem{definition}{Definition}

\newcommand{\mc}{\mathcal}

\newcommand{\mbb}{\mathbb}
\newcommand{\mr}{\mathrm}

\newcommand{\pma}[1]{\begin{pmatrix} #1 \end{pmatrix}}
\newcommand{\psm}[1]{\begin{psmallmatrix} #1 \end{psmallmatrix}}

\newcommand{\tr}{\mathrm{Tr}}

\newcommand{\blue}[1]{{\color{blue} #1}}

\newcommand{\chicago}{Pritzker School of Molecular Engineering, The University of Chicago, Illinois 60637, USA}
\let\oldacl\addcontentsline
\renewcommand{\addcontentsline}[3]{}

\usepackage{varioref}
\labelformat{equation}{(#1)}

\labelformat{section}{#1}

\labelformat{figure}{#1}

\labelformat{proposition}{#1}

\labelformat{lemma}{#1}

\labelformat{theorem}{#1}

\labelformat{observation}{#1}

\labelformat{definition}{#1}

\begin{document}

\title{Error-structure-tailored early fault-tolerant quantum computing}

\begin{abstract}
Fault tolerance is widely regarded as indispensable for achieving scalable and reliable quantum information processing. 
However, the spacetime overhead required for fault-tolerant quantum computation remains prohibitively large. 
A critical challenge arises in many quantum algorithms with Clifford + $\varphi$ compiling, where logical rotation gates $R_{Z_L}(\varphi)$ serve as essential components. 
The Eastin-Knill theorem fundamentally limits these operations, preventing their transversal implementation in quantum error correction codes and necessitating resource-intensive workarounds through $T$-gate compilation combined with magic state distillation and injection. 
In this work, we consider error-structure-tailored fault tolerance, where fault-tolerance conditions are analyzed by combining perturbative analysis of realistic dissipative noise processes with the structural properties of stabilizer codes.
Based on this framework, we design a 1-fault-tolerant continuous-angle rotation gate $R_{Z_L}(\varphi)$ in stabilizer codes, implemented via dispersive-coupling Hamiltonians.
Our approach could circumvent the need for $T$-gate compilation and distillation, offering a hardware-efficient solution that maintains simplicity, minimizes physical footprint, and requires only nearest-neighbor interactions. 
Integrating with recent small-angle-state preparation techniques, we can suppress the gate error to $91|\varphi| \cdot p^2$ for small rotation angles $|\varphi|<0.2$ rad (where $p$ denotes the physical error rate). 
For current achievable hardware parameters ($p=1\cdot 10^{-3}$), this enables reliable execution of over $10^7$ small-angle rotations when $|\varphi|\approx 10^{-3}$, meeting the requirements of many near-term quantum applications.
We estimate that, compared to the 15-to-1 magic state distillation and magic state cultivation approaches, our method reduces spacetime resource costs by factors of $1337.5\times$ and $43.6\times$, respectively, for a Heisenberg Hamiltonian simulation task under realistic hardware assumptions.
\end{abstract}
\date{\today}
            

\author{Pei Zeng}  
\email{qubitpei@gmail.com}
\affiliation{\chicago}

\author{Guo Zheng}
\affiliation{\chicago}

\author{Qian Xu}
\affiliation{\chicago}
\affiliation{Institute for Quantum Information and Matter, Caltech, Pasadena, California, USA}
\affiliation{Walter Burke Institute for Theoretical Physics, Caltech, Pasadena, California, USA}

\author{Liang Jiang}
\email{liang.jiang@uchicago.edu}
\affiliation{\chicago}

\maketitle

\section{Introduction} \label{sec:intro}

The goal of building quantum computer is to provide unique and useful applications compared to classical counterparts.
While recent experimental advances~\cite{hugginsUnbiasingFermionicQuantum2022,morvanPhaseTransitionsRandom2024,liu2025robustquantumcomputationaladvantage} have demonstrated that certain computational tasks can be executed on noisy intermediate-size quantum (NISQ) devices without quantum error correction (QEC), theoretical studies~\cite{aharonov2023polynomial,schuster2024polynomial} suggest that the scalability of the noisy devices face fundamental limitation. 
In contrast, fully fault-tolerant (FT) quantum computers, which are expected to provide unique advantages in various applications such as Hamiltonian simulation and factoring, require the gate resources which are prohibitively high~\cite{lee2021evenmore,regev2025efficient}, making them unlikely to be available in the near future.
As a transition stage from NISQ to fully FT quantum computers, early FT devices with limited error suppression are attracting increasing attention~\cite{karabarwa2024early,lin2022heisenberg}. 
Well-designed error correction techniques and quantum algorithms may enable such early fault-tolerant devices to surpass classical simulators and become both practical and impactful in the near future.

In many quantum algorithms designed for the early-stage FT quantum computers~\cite{lin2022heisenberg,dong2022ground,zeng2025simple,ding2023simultaneous,yu2025lindbladian}, logical-$Z$ rotation gates $R_{Z_L}(\varphi):= e^{i\varphi Z_L}$ play a crucial role, as they enable the efficient incorporation of classical analog information into quantum circuits. 
When considering the fully fault-tolerant implementation of $R_{Z_L}(\varphi)$, however, the rotation gate becomes notoriously resource-intensive. 
Due to Eastin-Knill's theorem~\cite{eastin2009restrictions}, the continuous group of the rotation gates $R_{Z_L}(\varphi)$ cannot be implemented transversally on quantum error correction codes.
As a result, a common practice to implement $R_{Z_L}(\varphi)$ is to first compile it into a sequence of tens of $T$ gates and then implement them by magic state distillation~\cite{Litinski2019magic} and injection. 
We remark that, even with state-of-the-art ancillary-based compilation techniques~\cite{bocharov2015efficient}, achieving a compilation infidelity of $10^{-8}$ requires nearly 40 $T$ gates. Meanwhile, the spacetime resource cost of magic state distillation scales with tens to hundreds of $d^3$, where $d$ is the code distance~\cite{litinski2019gameofsurfacecodes}.

In a recent study on the early FT architectures, Akahoshi et al.~\cite{akahoshi2024partially} proposed an interesting partial FT design. In this approach, Clifford gates are implemented in a fully FT manner, while rotation gates are performed non-fault-tolerantly, introducing a first-order error of $\frac{2}{15}p$, where $p$ is the circuit noise level. When the circuit noise $p$ is suppressed to a low value of $10^{-4}$ where higher-order error becomes negligible, this design enables the realization of quantum circuits surpassing the capabilities of classical simulators. However, the presence of first-order errors imposes a significant limitation on the scalability of this partial FT architecture.

Here, we consider error-structure-tailored fault tolerance, where the FT is defined by carefully examining the noisy dissipative process during the gate implementation of the stabilizer codes. Similar ideas were formerly explored in the gate design of bosonic quantum error correction codes~\cite{Rosenblum2018fault,ma2020path,Xu2024FaultTolerant}. 
Here, we show how to leverage the error structure in qubits and the error detectability of the stabilizer codes \emph{at the same time} to implement the continuous group of $Z$-rotation gates $R_{Z_L}(\varphi)$. 
As a concrete example, we design 1-fault-tolerant $Z$-rotation gates on the $[[4,1,1,2]]$ subsystem code with 4 physical qubits, 1 logical qubit, 1 gauge qubit and distance of 2, based on the dispersive coupling between two data qubits. As an alternative, the 1-fault-tolerant $Z$-rotation gate can also be implemented by a dispersive coupling to a 3-level $g$-$f$ qubit~\cite{ma2020path,kubica2023erasure} or a dual-rail qubit~\cite{teohDualrailEncodingSuperconducting2023b,kubica2023erasure}. 
We then utilize the 1-FT error-structure-tailored fault-tolerant rotation gate $R_{Z_L}(\varphi)$ to design two schemes: the expansion scheme and the projection scheme, to fault-tolerantly prepare the analog rotation state $\ket{r_\varphi}_L := R_{Z_L}(\varphi) \ket{+}_L$ on a generic stabilizer code in a resource-efficient way. Together with the rotation-state injection method~\cite{litinski2019gameofsurfacecodes}, this allows us to perform analog rotation gate $R_{Z_L}$ without the $T$-gate compiling, thus remarkably save the spacetime cost for the early fault-tolerant quantum computers. The expansion scheme is suitable for arbitrary rotation angle $\varphi$, which can prepare the rotation state $\ket{r_\varphi}_L$ with the error of about $70\cdot p^2$ with a high successful probability over $70\%$ where $p$ is the physical error rate of each gate; the projection scheme is suitable for the preparation of $\ket{r_\varphi}_L$ with small $\varphi$, with the trace distance of about $15|\varphi|\cdot p^2$. 
Specifically, considering the current hardware error rate of $1\cdot 10^{-3}$ and the angle size $|\varphi|=1\cdot 10^{-3}$, the successful probability of the projection scheme is $18.7\%$ and $5.4\%$ on a surface code with distance $d=12$ and $d=18$, respectively.

\begin{figure}[htbp]
    \centering
    \includegraphics[width=0.48\textwidth]{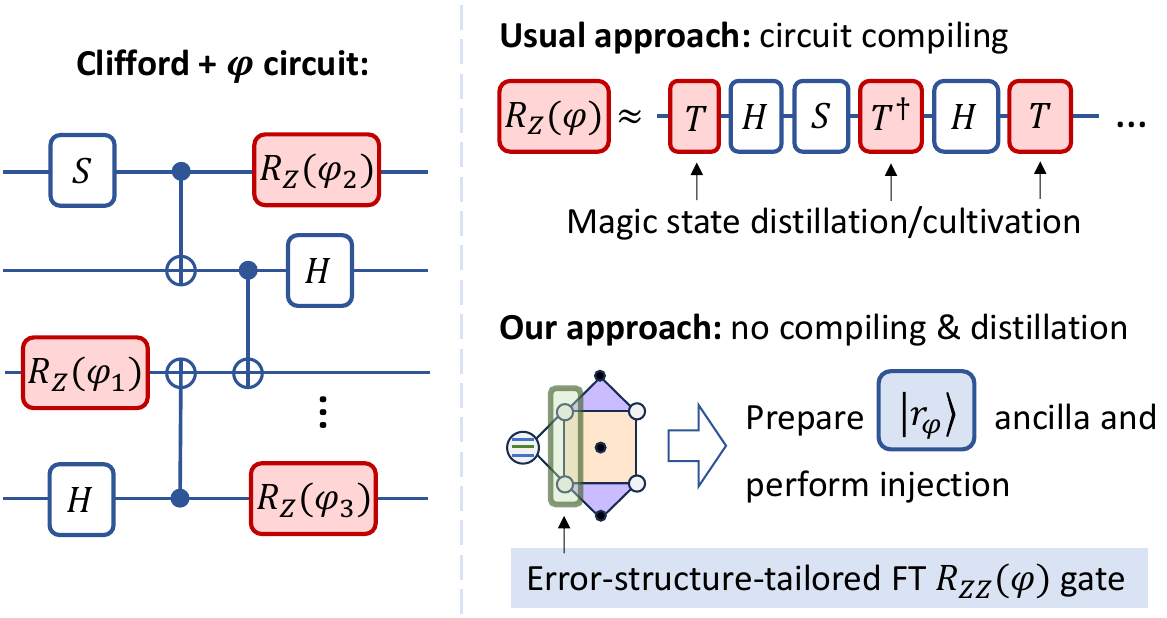}
    \caption{Illustration of the main idea. For the Clifford+$\varphi$ circuits, unlike the usual approach to compile the $R_Z(\varphi)$ gates to many $T$ gates and then perform magic state distillation or cultivation, we directly prepare the $\ket{r_\varphi}:=R_{Z}(\varphi)\ket{+}$ ancilla based on the error-structure-tailored FT analysis and perform injection.}
    \label{fig:idea}
\end{figure}

By carefully taking the probabilistic rotation state injection and coherent error cancellation into consideration, we have shown that the rotation gate based on the projection scheme can be performed with a dephasing error below $91 |\varphi|\cdot p^2$.
Combined with the canonical error mitigation techniques~\cite{temme2017error,endo2018practical}, this allows us to perform about $1.1 \cdot 10^7$ small rotation gate with the current achievable physical error rate $p=1\cdot 10^{-3}$ and the rotation angle $|\varphi|\approx 10^{-3}$. We remark that, this rotation angle value is common in many early FT algorithms especially the ones related to Trotter-based Hamiltonian simulation~\cite{childs2021theory} and quantum phase estimation~\cite{lin2022heisenberg}. More concretely, we have shown that for a $N$-site Heisenberg model with normalized disordered magnetic field~\cite{childs2021theory}, our rotation gate design enables a simulation time of $T_{\mr{max}}\simeq 31.4$ with $N=100$ and $p=1\cdot 10^{-3}$, which is much longer than what can be demonstrated on the NISQ devices.
To demonstrate the advantage of our Clifford+$\varphi$ approach, we have shown a reduction in the spacetime resource requirements by factors of 1337.5$\times$ and 43.6$\times$, respectively, to the 15-to-1 magic-state-distillation approach~\cite{Litinski2019magic} and the recent magic-state-cultivation approach~\cite{gidney2024cultivation}, for the $50$-spin Trotter-based Heisenberg Hamiltonian simulation tasks.

\section{Error-structure-tailored fault tolerance} \label{sec:FTdef}

We first review the canonical definition of the FT gadgets~\cite{aliferis2005quantum,gottesman2009intro}. 
The FT gadgets are associated to a QEC code, which is usually a $[[n,k,d]]$ stabilizer code with $n$ physical qubits and $k$ logical qubits. In this work, we focus on the case with $k=1$. A stabilizer code is defined by a set of stabilizer generators $\{S_i\}_{i=0}^{n-k-1}$ which span an abliean subgroup $\mc{S}$ of the Pauli group $\mbb{P}^{n}$ without $-I$. The distance $d$ is defined to be the smallest weight of a non-identity logical operator, i.e., $d=\min\{ |l| \,|\, l\in \mc{C}(\mc{S})\backslash \mc{S}\}$ where $\mc{C}(\mc{S})$ is the centralizer of the stabilizer group $\mc{S}$. 

To begin with, we first introduce the concept of error filter and ideal decoder. Denote the ideal code space projector as
\begin{equation}
\Pi_0 = \prod_{i=0}^{n-k-1} \frac{1}{2} ( I + S_i ).
\end{equation}
Now, we extend the projector definition by considering the set of Pauli operators with weight less or equal to $r$, $\mbb{P}_r := \{P\in\mbb{P}^n | |P|\leq r \}$. Here, $r$ satisfy $r\leq t:= \lceil (d-1)/2 \rceil$. 
The $r$-filter with $r\leq t:= \lceil (d-1)/2 \rceil$ is an extended code space projector with at most $r$ errors,
\begin{equation} \label{eq:Pir}
\Pi_{(r)} = \sum_{P\in \mbb{P}_{r,\mc{S}} } P\, \Pi_0 P.
\end{equation}
Here, $\mbb{P}_{r,\mc{S}} = \mbb{P}_r/\sim$ is a quotient set of $\mbb{P}_r$ by the equivalence relationship of $P_1 \sim P_2: P_1 P_2 \in \mc{C}(\mc{S})$. We choose $P$ in $\mbb{P}_{r,\mc{S}}$ to avoid unwanted double counting since $P_1 \Pi_0 P_1$ and $P_2 \Pi_0 P_2$ will correspond to the same space if $P_1 P_2 \in \mc{C}(\mc{S})$.

\begin{figure}[htbp]
    \centering
    \includegraphics[width=0.4\textwidth]{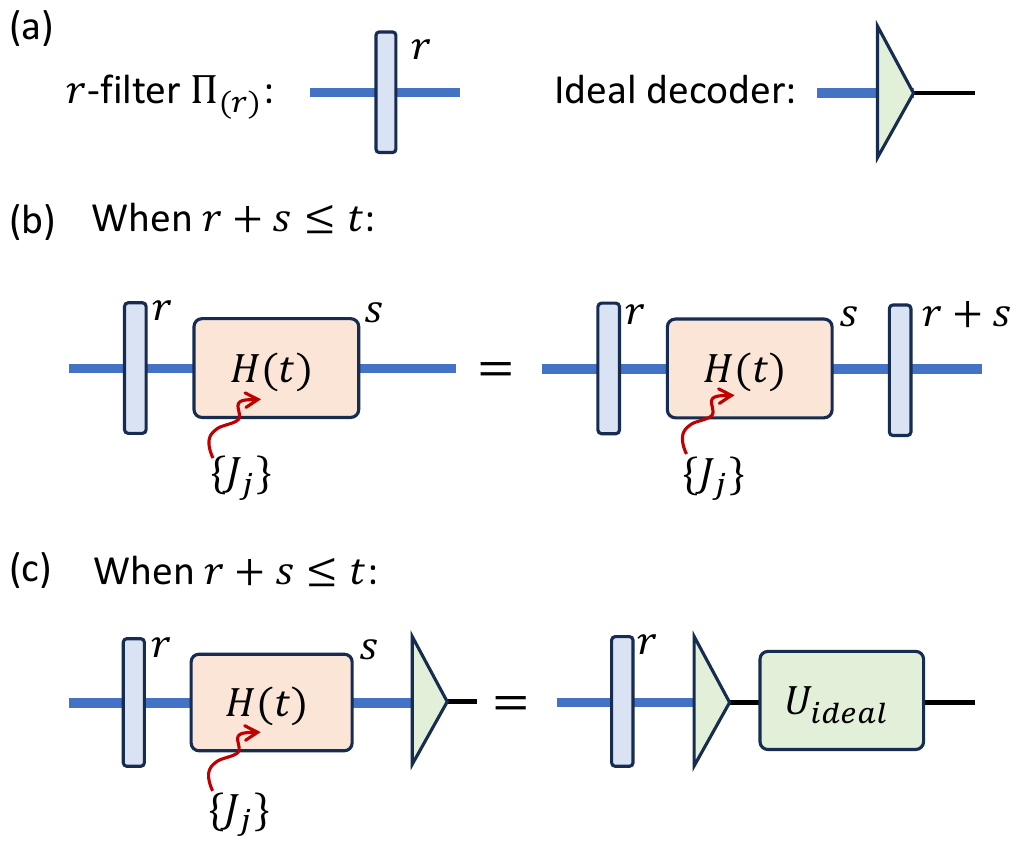}
    \caption{(a) $r$-filter is a projector $\Pi_{(r)}$ defined in \autoref{eq:Pir} illustrated as a blue rectangle. Ideal decoder is illustrated as a green triangle followed by decoded logical block (black line). (b,c) The error-structure-tailored FT gadget requirements. We consider the noisy gate implementation as a dissipative process described by Lindblad master equation in \autoref{eq:Lindblad}. We consider the Dyson expansion of the dynamics, and truncate to the $s$ order to consider the fault tolerance.}
    \label{fig:FTdef}
\end{figure}

The ideal decoder is a noiseless gadget where we first perform a syndrome measurement by the projectors,
\begin{equation}
\Pi_{\vec{s}} = \prod_{i=0}^{n-k-1} \frac{1}{2} \left( I + (-1)^{s_i} S_i \right),
\end{equation}
where $\vec{s} = [s_0, ..., s_{n-k-1}]$ is the $(n-k)$-bit syndrome measurement outcome. Based on $\vec{s}$, we then perform an error recovery Pauli operator $P(\vec{s})\in\mbb{P}^n$ on the system. The Knill-Laflemme condition ensure that, if a weight-$r$ occurs on an ideal code state with $r\leq t$, then we can recovery the code state after the ideal decoding.

Without loss of generality, we focus on the definiton of the FT single-code-block gate gadget. The FT state preparation, measurement, QEC and multi-code-block gate can be defined in a similar manner~\cite{gottesman2009intro}.

\begin{definition}[$t$-FT gate~\cite{gottesman2009intro}] \label{def:FTgate_canonical}
The single-code-block gate is $t$-FT if it satisfies,
\begin{enumerate}
\item If the input state passes the $r$-filter and $s$ faults occur during the gate implementation, then the output state can pass the $s+r$ filter when $s+r\leq t$.
\item If the input state passes the $r$-filter and $s$ faults occur during the gate implementation, then the output state after the ideal decoder is equivalent to first implement an ideal decoder after the $r$-filter then apply an ideal gate. 
\end{enumerate}
\end{definition}
Here, the first condition ensure that the error propagation is limited and the second condition guarantee that the function of the noisy gate is close to the idea gate.

The canonical FT definition in Def.~\ref{def:FTgate_canonical} counts the discrete faults during the gate implementation. A common practice to perform the FT analysis is to consider a local stochastic Pauli error model~\cite{gottesman2014fault}, where a random Pauli error may occur after the implementation of each physical gate or idling operation. While convenient for the theoretical analysis, this model and the FT definition in Def.~\ref{def:FTgate_canonical} fail to capture the practical noise structure and the properties of the dynamics. 

In practice, the gate is performed by a Hamiltonian control sequence. The (Markovian) noise can be modeled by the dissipative interaction to the environment. The overall process is described by a Lindblad master equation,
\begin{equation} \label{eq:Lindblad}
\frac{d \rho}{dt} = \mc{L}(t) \rho = -i [ H(t), \rho ] + \left( \sum_j D[\sqrt{\gamma_j} J_j]\right) \rho,
\end{equation}
where $H(t)$ is the system Hamiltonian and $D[O]=O\bullet O^\dagger - \frac{1}{2}\{ O^\dagger O, \bullet \}$ is a Lindblad superoperator associated with a jump operator $O$, and $\gamma_j$ is the dissipation rate for the normalized error $J_j$.

Suppose the spectral norm of the system Hamiltonian $\|H\|$ is much larger than the dissipation rates $\{\gamma_j\}$. In this case, we can expand the Lindblad propagator $\mc{G}(t,0) = e^{\mc{L}t}$ by the Dyson series with respect to $\mc{D}:= \sum_{j} D[\sqrt{\gamma_j}J_j]$,
\begin{equation} \label{eq:DysonExpansion}
\rho(t) = \mc{G}(t,0)\rho(0) = \sum_{q=0}^{\infty} \mc{G}_{q}(t,0) \rho(0),
\end{equation}
where $\mc{G}_0(t,0) = \mc{U}(t,0) = U(t,0)\bullet U^\dagger(t,0)$ with $U(t,0):= \mc{T} \exp[ -i\int_0^t H(t') dt' ] $ is the noiseless evolution and 
\begin{equation}
\begin{aligned}
\mc{G}_q(t) = \left[ \int_{t_h = 0}^t d t_h \right] \mc{T}\Big( & \mc{U}(t,t_q) \mc{D}... \\
 & \mc{D}\mc{U}(t_2,t_1) \mc{D} \mc{U}(t_1,0) \Big),
\end{aligned}
\end{equation}
is the $q$-th order noise expansion. $q\geq 1$ indicates the expansion order related to the dissipator $\mc{D}$. Here, $\mc{G}_q(t)$ describes the $q$-th order noisy dynamics of the whole procedure. We can define the error-structure-tailored $t$-FT gate based on the noise expansion.

\begin{definition}[Error-structure-tailored $t$-FT gate] \label{def:FTgate_error_tailored}
Consider a noisy gate described by a Lindbladian master equation in \autoref{eq:Lindblad}. Consider the Dyson series expansion of its propagator $\mc{G}(t,0)$ with respect to the noise superoperator $\mc{D}$ in \autoref{eq:DysonExpansion}. 
The single-code-block gate is $t$-FT if it satisfies,
\begin{enumerate}
\item If the input state passes the $r$-filter, then the output state with the truncated dynamics $\mc{G}^{[s]}$ can pass the $s+r$ filter when $s+r\leq t$.
\item If the input state passes the $r$-filter, then the output state after the ideal decoder with the truncated dynamics $\mc{G}^{[s]}$ is equivalent to first implement an ideal decoder after the $r$-filter then apply an ideal gate. 
\end{enumerate}
\end{definition}

We can also define FT gadgets for the quantum error detection (QED) code in Def.~\ref{def:FTgate_canonical} and \ref{def:FTgate_error_tailored} by modifying the definition of the $r$-filter and the ideal decoder. The $r$-filter with $r\leq d-1$ for the QED code can be defined similarly as \autoref{eq:Pir}. The ideal decoder for QED code is simply an ideal projection to the code space $\Pi_0$: if the projection fails, we then perform post-selection. In Appendix~\ref{sec:AppFTQED}, we show that when focusing on the post-selected quantum state which passes the QED, we can generalize the FT analysis for QEC code in Ref.~\cite{gottesman2009intro}: the noisy circuit, under the post-selection condition, can be used to simulate the ideal circuit with an error of $O(p^{d})$.

\section{1-Fault-tolerant logical rotation gate on [[4,1,1,2]] code} \label{sec:RZZgate4112}

As an illustrative example, we examine the implementation of $R_{Z_L}$ gate on a $[[4,1,1,2]]$ subsystem QED code. 
As shown in \autoref{fig:4112code}, we align the $4$ data qubits $D_0, D_1, D_2, D_3$ for the $[[4,1,1,2]]$ code as a 2-by-2 grid $\psm{D_0, D_1 \\ D_2, D_3}$. The stabilizers of the code are
\begin{equation} \label{eq:4112SZSX}
S_Z = \pma{ Z, & Z \\ Z, & Z }, \quad  S_X = \pma{ X, &X \\ X, & X }.
\end{equation}
The gauge qubit is defined by the following gauge operators,
\begin{equation} \label{eq:4112gZgX}
g_Z = \pma{ Z, & Z \\ I, & I}, \quad g_X = \pma{ X, & I \\ X, & I}.
\end{equation}
The logical operators are
\begin{equation} \label{eq:4112LZLX}
L_Z = \pma{ Z, & I \\ Z, & I}, \quad L_X = \pma{ X, & X \\ I, & I}.
\end{equation}

The introduction of the redundant gauge qubit is beneficial, as it allows us to perform weight-2 checks, such as $\psm{X,X \\ I, I}$, $\psm{I, I \\ X, X}$, $\psm{Z, I \\ Z, I}$, and $\psm{I, Z \\ I, Z}$, instead of directly performing the weight-4 checks on $S_X$ and $S_Z$. This only impacts the state of the gauge qubit, reducing the overhead associated with second-order error rates. When the gauge qubit is fixed to $g_Z$ (resp. $g_X$), we can treat the code as the smallest $[[4,1,2]]$ rotated surface code with the stabilizer generators of $S_Z, S_X$ and $g_Z$ (resp. $g_X$).

\begin{figure}[htbp]
    \centering
    \includegraphics[width=0.4\textwidth]{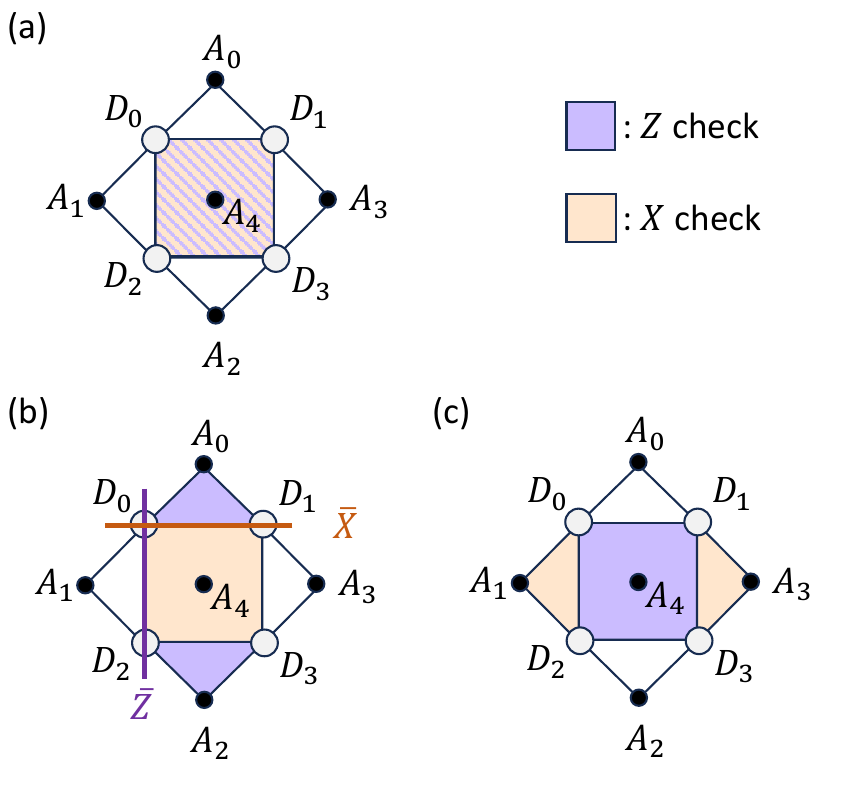}
    \caption{(a) The $[[4,2,2]]$ code. The stabilizers are weight-4 $X$- and $Z$-operators on the data qubits. Treating one logical qubit as the gauge qubit, we obtain the $[[4,1,1,2]]$ subsystem code. (b,c) When the gauge qubit is fixed on the $Z$-basis (resp., $X$-basis), the $Z$-check (resp., $X$-check) can be written as two weight-2 gauge operators (shaded triangles), the code becomes the $[[4,1,2]]$ stabilizer code. 
    }
    \label{fig:4112code}
\end{figure}

The $Z$-rotation gate on the code is $R_{Z_L}(\varphi)=\exp(i\varphi Z_0 Z_2)$. In the canonical circuit QED implementation, if we directly decompose it to elementary CNOT/CZ gates and single qubit rotation gate, the gate will be non fault-tolerant and lead to undetectable first-order error~\cite{akahoshi2024partially}. Now, consider the $R_{Z_L}(\varphi)=\exp(i\varphi Z_0 Z_2)$ gate to be implemented with a $ZZ$-type dispersive-coupled Hamiltonian $H_{ZZ} = -\chi \cdot Z_0 Z_2$ with an evolution time $T = \varphi/\chi$ on two qubits each with the computational basis $\ket{g}$ and $\ket{e}$. 
The qubit noise is usually characterized by the jump operators for the relaxation $J_{e\to g}^{(i)} = \ket{g}\bra{e}$, excitation $J_{g\to e}^{(i)} = \ket{e}\bra{g}$ and dephasing $J_{\phi}^{(i)} = \ket{e}\bra{e}$ on $i$-th qubit. Since the excitation is usually much smaller than the relaxation and the analysis for the excitation error is similar to the relaxation error, we will ignore the excitation error in the current analysis.

\begin{proposition} \label{prop:4112FTdispersive}
The $Z$-rotation gate $R_{Z_L}(\varphi)$ on the [[4,1,1,2]] QED code implemented by the dispersive-coupled Hamiltonian $H_{ZZ}$ with the Markovian noises characterized by the jump operators of relaxation $\{J_{e\to g}^{(i)}\}_i$ and dephasing $\{J_{\phi}^{(i)}\}_i$ satisfy the 1-FT gate requirement in Def.~\ref{def:FTgate_error_tailored}.
\end{proposition}
The proof of \autoref{prop:4112FTdispersive} can be found in Appendix~\ref{ssc:AppPropZZdisp}. We remark that, Proposition~\ref{prop:4112FTdispersive} still holds if we consider the excitation noise $J^{(i)}_{g\to e}$.

In some scenarios, we may want to avoid the direct coupling of the data qubits. In this case, we also consider an alternative $1$-FT logical-rotation gate design on the $[[4,1,1,2]]$-code based on the dispersive coupling of the data qubits to a $3$-level ancilla qubit, shown in \autoref{fig:RZZgate_ancilla}. 
Here, the $3$-level $g$-$f$ qubit is introduced to avoid the harmful error propagation in the $R_{ZZ}(\varphi)$ gate. 
In a usual $3$-level system on a transmon, the relaxation jump operators are $J_{f\to e}=\ket{e}\bra{f}$ and $J_{e\to g}=\ket{g}\bra{e}$ and there is no direct relaxation from $f$ to $g$ level.
This error structure is widely used in the fault-tolerant gate constrution for the bosonic codes~\cite{ma2020path,Xu2024FaultTolerant}, which also demonstrated in experiments~\cite{Rosenblum2018fault,reinholdErrorcorrectedGatesEncoded2020a,puttermanHardwareefficientQuantumError2025}.
In our case, the ancilla qubit is first prepared on the ground state $\ket{g}$, then coherently drive to $\ket{g}+\ket{f}$ by a $Y$-drive. 
After that, the ancilla qubit is coupled to the two data qubits, $D_0$ and $D_2$, by a dispersive Hamiltonian $H_{fe}=\chi \ket{f}_A\bra{f}\otimes \ket{e}_{D_0(D_2)}\bra{e}$. The dispersive Hamiltonian is used to realize two CZ gates between $A$ and the two data qubits. After that, we apply a $R_X(\varphi)$ gate on the $g$-$f$ qubit. We then apply the dispersive coupling again and rotate the ancillary qubit back to $\ket{g}$. 
In the ideal case when no error happens, the measurement outcome $o_{zz}$ on the ancillary qubit will be $g$, and a $R_{ZZ}(\varphi)$ gate is successfully implemented on the data qubits. 

\begin{figure}[htbp]
    \centering
    \includegraphics[width=0.4\textwidth]{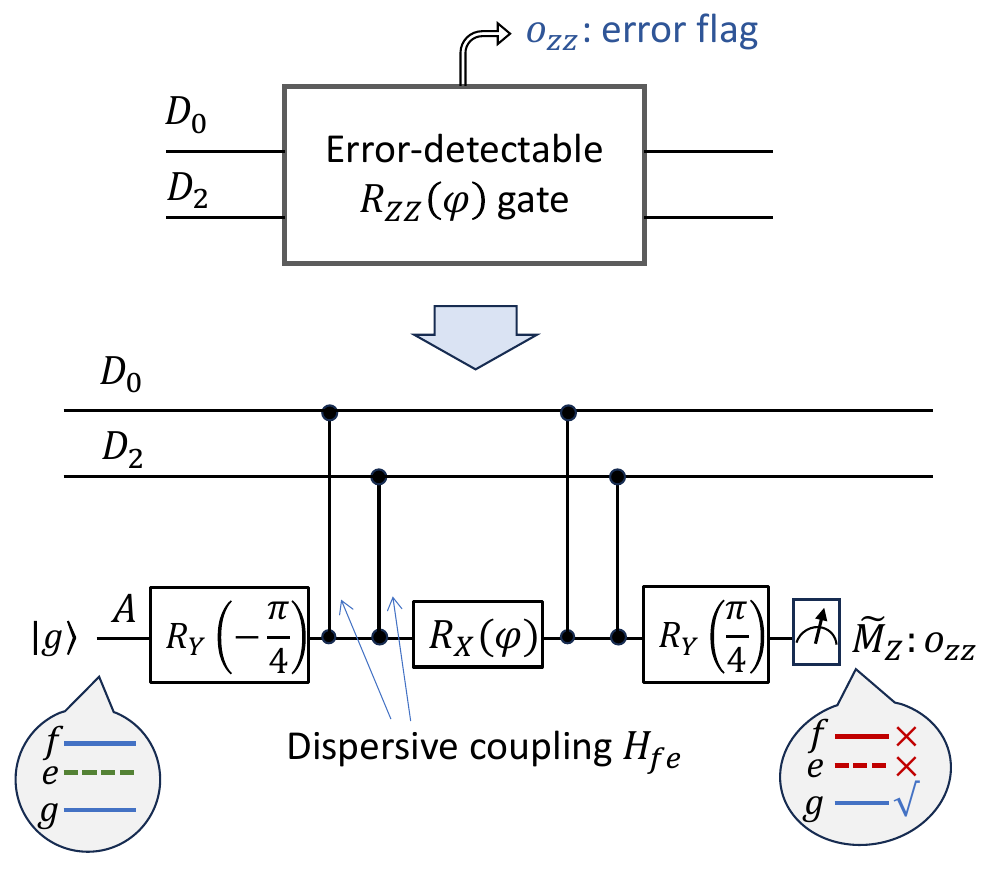}
    \caption{Logical rotation gate $R_{Z_L}(\varphi)$ on the $[[4,1,1,2]]$ code by dispersively coupling the two data qubits $D_0$ and $D_2$ to a $3$-level g-f qubit. We assume the relaxation and dephasing error occuring on all the qubits during the whole procedure.}
    \label{fig:RZZgate_ancilla}
\end{figure}

\begin{proposition} \label{prop:4112FTancilla}
The $Z$-rotation gate $R_{Z_L}(\varphi)$ on the [[4,1,1,2]] QED code implemented by the circuit in \autoref{fig:RZZgate_ancilla} with the Markovian noises characterized by the jump operators of relaxation $\{J_{f\to e}^{(i)}, J_{e\to g}^{(i)}\}_i$ and dephasing $\{J_{\phi}^{(i)}\}_i$ satisfy the 1-FT gate requirement in Def.~\ref{def:FTgate_error_tailored}.
\end{proposition}
The proof of \autoref{prop:4112FTancilla} can be found in Appendix~\ref{ssc:AppPropZZancilla}. Proposition~\ref{prop:4112FTancilla} also holds if we consider the excitation noises on qubits and the $3$-level $g$-$f$ qubit.
Alternatively, one can also design a similar ancilla-based $R_{Z_L}(\varphi)$ gate using a dual-rail transmon~\cite{kubica2023erasure} or cavity~\cite{teohDualrailEncodingSuperconducting2023b}. 
This alternative design is discussed in Appendix~\ref{sec:App_dual_rail}. Here, we mainly focus on the design based on the $3$-level $g$-$f$ qubit.

To test the performance of the logical rotation gate, we consider to run a logical $\ket{r_\varphi}_L:= R_{Z_L}(\varphi)\ket{+}_L$ state preparation circuit on the $[[4,1,1,2]]$ QED code. In the circuit, we first prepare the logical $\ket{+}_L$ state, then perform the logical rotation gate following the above methods. Afterwards, we perform a noisy QED circuit. The $[[4,1,1,2]]$ $\ket{+}_L$ state preparation circuit and the QED circuit is in Appendix~\ref{sec:4112prepQEDexp}.
If the state passes the noisy QED, we then check the trance distance of the logical state $\rho$ to $\ket{r_\varphi}_L$. The results are shown in \autoref{fig:Num4112coherent}.

In the whole state preparation circuit, we consider the fully coherent noise simulation. For the $R_{Z_L}(\varphi)$ gate, we perform the Lindbladian simulation with the ancilla relaxation and dephasing error. We consider practical Hamiltonian parameters based on the $g$-$f$ qubit studies~\cite{kubica2023erasure}: the dispersive coupling $\chi = 2\pi\cdot 5$ MHz, and the Rabi driving strength for the $g$-$f$ qubit is $\Omega_{gf} = 2\pi\cdot 20$ MHz. During the implementation of the $R_{Z_L}(\varphi)$ gate, we consider a noisy Lindbladian simulation with the dissipation rates $\gamma_{f\to e} = \gamma_{e\to g} = \gamma_{\phi}$ which correspond to the $f\to e$, $e\to g$ relaxation rate and the dephasing rate. We simply denote them as $\gamma$.
To simplify the simulation, we ignore the error structure in the circuit elements other than the $R_{Z_L}(\varphi)$ gate, including the state preparation, measurements, idling, Hadamard gates and $\mr{CNOT}$ gates, and consider the worst-case circuit-level depolarizing noise model with a depolarizing rate $p$. 
The typical parameters for current superconducting quantum platform have $p = 10^{-3}$ and $\gamma =2\times 10^4$ Hz, with a ratio of $p/\gamma = 50$ ns characterizing the typical time scale of a superconducting quantum operation. 
As we are getting better coherence for our superconducting qubits, we expect that both $p$ and $\gamma$ will become even smaller, while the ratio remain mostly unchanges. Hence, we fixed the ratio of $p/\gamma = 50$ ns to check the error scaling with respect to all the circuit noise parameters.

\begin{figure}[htbp]
    \centering
    \includegraphics[width=0.5\textwidth]{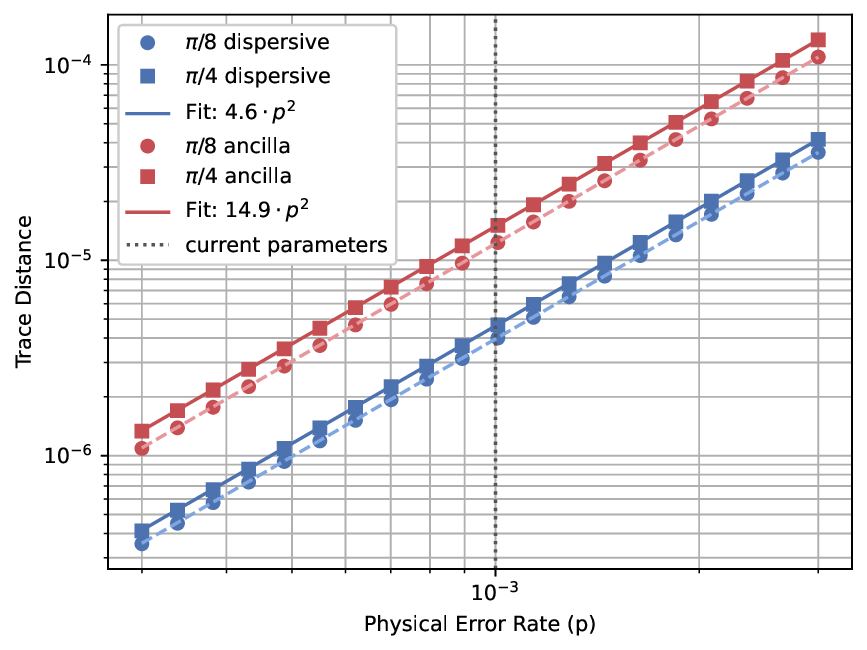}
    \caption{Performance of the $1$-FT rotation gate on the $[[4,1,1,2]]$ code. We simulate the noisy $R_{Z_L}$ gate by preparing the logical rotation state $\ket{r_\varphi}_L$ by acting $R_{Z_L}$ on the $\ket{+}_L$ state. The performance is characterized by the trace distance to the ideal state. We consider both direct dispersive coupling (blue curves) and the ancillary-based approach (red curves) in \autoref{fig:RZZgate_ancilla}. 
    }
    \label{fig:Num4112coherent}
\end{figure}

From \autoref{fig:Num4112coherent} we can see that the logical rotation gates based on both the direct dispersive coupling Hamiltonian and the ancillary-based approach in \autoref{fig:RZZgate_ancilla} are $1$-FT. The error increases slightly when the rotation angle $\varphi$ becomes larger.
The ancilla-based logical rotation gate is noisier than the direct dispersive coupling gate due to a longer implementation time and more mechanisms for the second-order errors.

\section{1-Fault-tolerant logical rotation state preparation on surface codes: expansion scheme} \label{sec:expansion}

The $1$-FT gate on the $[[4,1,1,2]]$ code can be used as a subroutine to prepare a $1$-FT logical rotation state $\ket{r_\varphi}_L := R_{Z_L}(\varphi) \ket{+}_L$ on a general stabilizer code by code expansion or code deformation. Here, we take the surface code as an illustrative example, shown in \autoref{fig:expansion}. We first prepare the logical rotation state $\ket{r_\varphi}_L$ on the $[[4,1,1,2]]$ code, then expand it to the surface code. After the code expansion, we first use a small surface code to perform further QED, then expand it to a larger surface code to store the prepared state $\ket{r_\varphi}_L$ and perform QEC.
The details of the code expansion can be found in Appendix~\ref{sec:4112prepQEDexp}.

\begin{figure*}[htbp]
    \centering
    \includegraphics[width=0.85\textwidth]{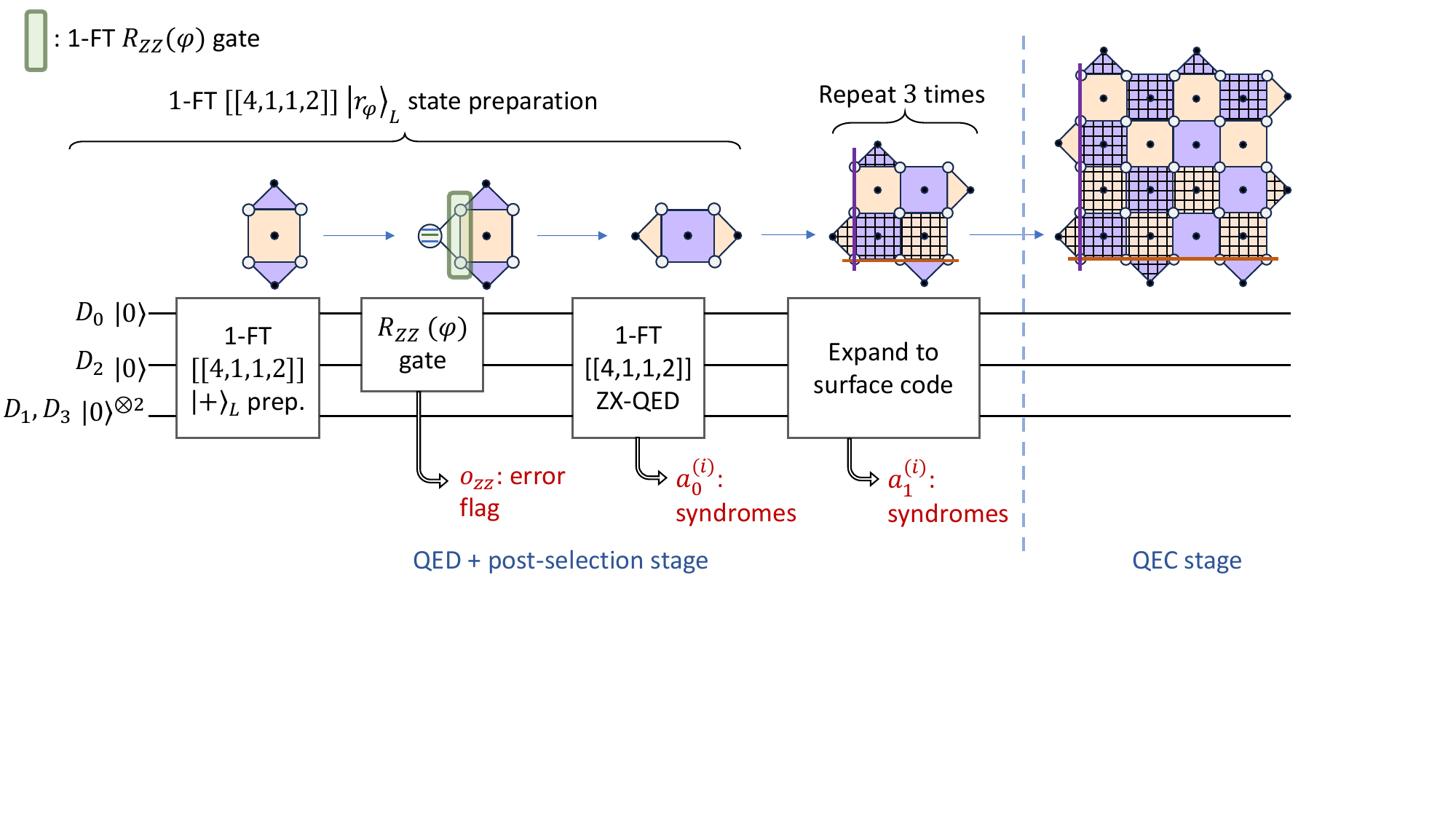}
    \caption{Expansion scheme to prepare the surface code rotation state $\ket{r_
    \varphi}_L$ 1-fault-tolerantly. There are QED and QEC stages. The stabilizers with grid indicates the ones with fixed syndrome values, serving as the detectors for QED or QEC.
    }
    \label{fig:expansion}
\end{figure*}

During the surface code expansion stage, there is a trade-off between the final logical error rate and the successful probability to pass the post-selection in the QED stage in \autoref{fig:expansion}. We choose to expand the $[[4,1,1,2]]$ code first to the $d=3$ surface code and measure the syndromes $3$ times to suppress the measurement error. If unwanted measurement outcomes occur, we then discard this round of state preparation; otherwise we further expand the code to a larger surface code and perform quantum error correction. 

To test the performance of the $\ket{r_\varphi}_L$ state preparation by the expansion scheme in \autoref{fig:expansion}, we perform circuit-level numerical simulations to estimate the logical error rate after the $d=3$ surface code expansion in \autoref{fig:expansion}, with results shown in \autoref{fig:expansion_num}. Here, we focus on the estimation of the error during the QED stage in \autoref{fig:expansion} since the error in the QEC stage will be of higher order.
For simulation simplicity, we ignore error structure in all circuit elements except the $R_{ZZ}(\varphi)$ gate while assuming worst-case canonical depolarizing noise; for the $R_{ZZ}(\varphi)$ gates, we first perform Lindbladian noisy gate simulation in QuTiP using parameters similar to \autoref{fig:Num4112coherent} then extract the Pauli error model for Stim circuit-level simulation. We fix $\varphi=\pi/4$ (yielding a Clifford/logical $S$ gate) for tractable Stim simulation, with logical error rates for $|\varphi|<\pi/4$ being similar to $\varphi=\pi/4$. For $|\varphi|>\pi/4$, one can first perform Clifford $S$ gates and implement the remaining rotation gates with $|\varphi|<\pi/4$. \autoref{fig:expansion_num}(a) demonstrates successful preparation probabilities $>80\%$ while maintaining small logical error rates.
\autoref{fig:expansion_num}(b) shows that the logical error rate of the prepared state is about $70\cdot p^2$ and $84\cdot p^2$ when we use the ancilla-free dispersive $R_{ZZ}(\varphi)$ gate in \autoref{prop:4112FTdispersive} and the ancilla-based dispersive $R_{ZZ}(\varphi)$ gate in \autoref{prop:4112FTancilla}, respectively. Under the current hardware condition where the depolarizing error rate is $1\cdot 10^{-3}$, one can prepare a rotation state $\ket{r_\varphi}_L$ with arbitrary $\varphi$ with an infidelity of $7\cdot 10^{-5}$ to $8.4\cdot 10^{-5}$.


\begin{figure}[htbp]
    \centering
    \begin{subfigure}[b]{\columnwidth}
        \centering
        \includegraphics[width=\linewidth]{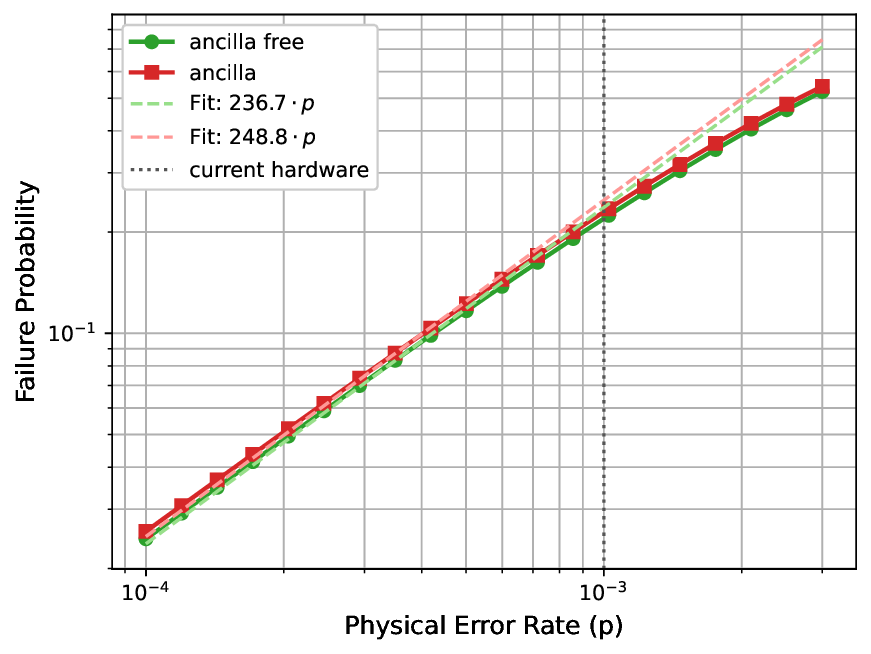}
        \caption{}
        \label{fig:fail_prob_expansion}
    \end{subfigure}
    
    \vspace{\floatsep} 
    
    \begin{subfigure}[b]{\columnwidth}
        \centering
        \includegraphics[width=\linewidth]{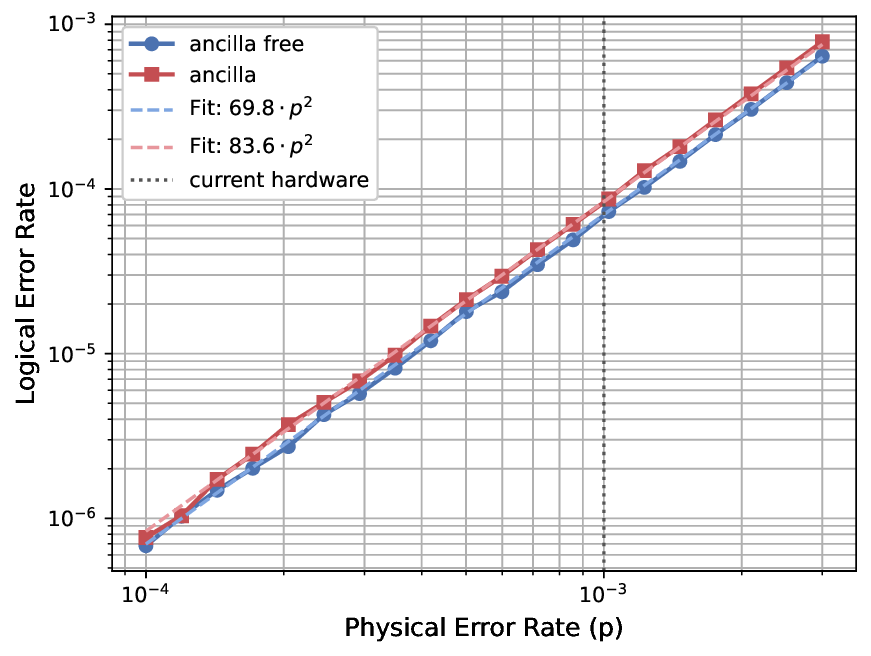}
        \caption{}
        \label{fig:error_prob_expansion}
    \end{subfigure}
    
    \caption{Performance characteristics of the Expansion scheme. (a) The overall failure probability. (b) The average infidelity of the Pauli channel to prepare $\ket{r_\varphi}_L$.}
    \label{fig:expansion_num}
\end{figure}

We emphasize that the above resource state preparation protocol is probabilistic and cannot directly implementing a deterministic rotation gate on the surface codes or other stabilizer codes. To realize rotation gates using these resource states $\ket{r_\varphi}_L$, we employ the rotation state injection circuits~\cite{litinski2019gameofsurfacecodes,akahoshi2024partially} (illustrated in Fig.~\ref{fig:RotationInjection}). These circuits are analogous to standard $T$-state injection circuits using CNOT gates. However, unlike $T$-state injection where an $S$ gate correction for measurement outcome $o_x=1$ ensures deterministic $T$-gate implementation, the arbitrary-angle case presents a key difference: when $o_x=1$, the implemented gate becomes $R_{Z_L}(-\varphi)$ instead of the desired $R_{Z_L}(\varphi)$, and this inversion cannot be straightforwardly corrected by a Clifford gate to obtain the target rotation.

To address this problem, we can introduce a repeat-until-success (RUS) procedure to adaptively implement the target rotation gate~\cite{litinski2019gameofsurfacecodes,akahoshi2024partially}. Suppose we have implemented unwanted $(-\varphi)$-rotation when $o_x=1$, we then prepare the ancilliary state $\ket{r_{2\varphi}}_L$ and perform the injection circuit to implement $(2\varphi)$-rotation; if we get $o_x=1$ again, we then prepare the ancilliary state $\ket{r_{4\varphi}}_L$ and so on. Since the probability of getting $o_x=1$ each time is $\frac{1}{2}$, the expected number of trials to implement the taret rotation gate $R_{Z_L}(\varphi)$ is $\mbb{E}(n_{\mr{trial}}) = 1\times \frac{1}{2} + 2 \times \frac{1}{4} + ... = 2$. In the worst case, we can set a maximum value for $n_{\mr{trial}}$ which may induce a small failure probability of $\delta=1/2^{n_{\mr{trial}}}$ for each rotation gate. In the later resource estimation discussion in \autoref{sec:Error_Count}, we will also take the repetition time of the rotation state injection into consideration and show that the overhead introduced by the repetition is mild.

\begin{figure}[htbp]
    \centering
    \includegraphics[width=0.42\textwidth]{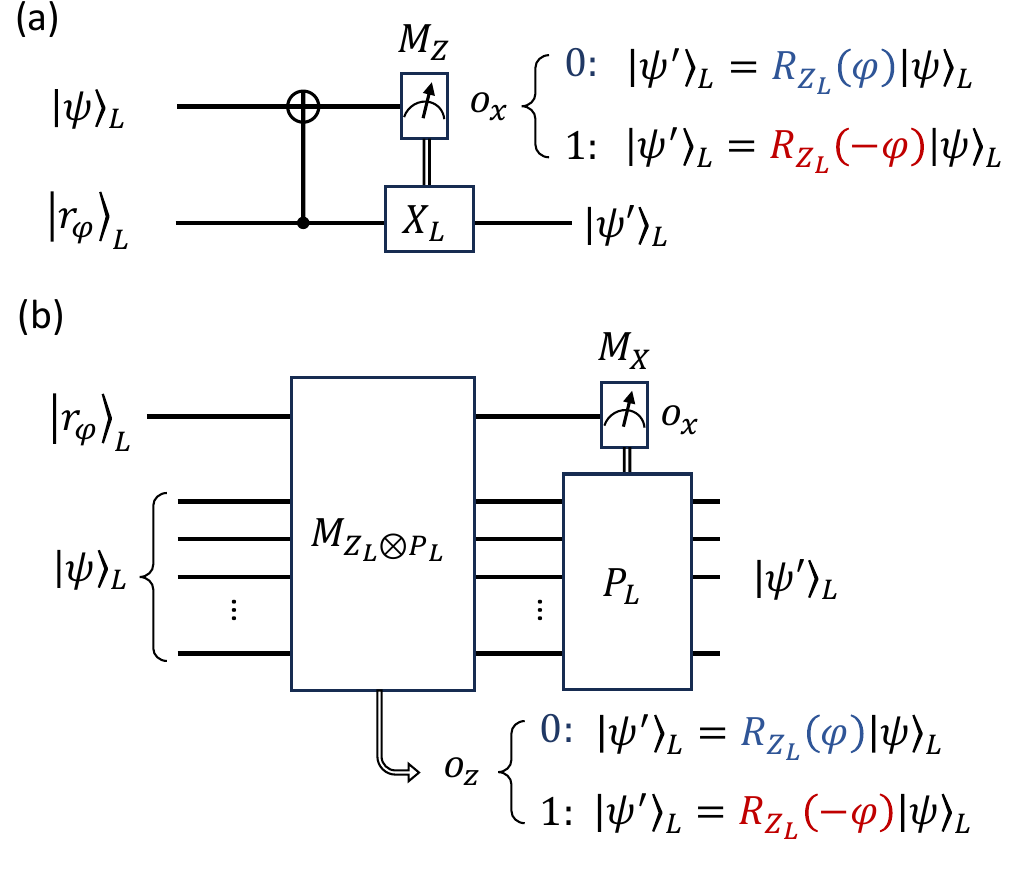}
    \caption{ Probabilistic implementation of the target rotation gate based on the rotation state~\cite{litinski2019gameofsurfacecodes,akahoshi2024partially}. (a) $R_{Z_L}(\varphi)$ gate based on the 1-bit teleportation protocol. When the measurement outcome $o_x$ is $1$, the resulting gate on the input state is $R_{Z_L}(-\varphi)$, in which case we need to introduce a repeat-until-success (RUS) procedure to compensate this gate. (b) Multi-qubit Pauli rotation gate $R_{P_L}(\varphi)$ based on the measurement-based circuit. When $o_z$ is $1$, the resulting gate is $R_{P_L}(-\varphi)$, which need to be compensated based on a RUS procedure.
    }
    \label{fig:RotationInjection}
\end{figure}

\section{1-Fault-tolerant Logical rotation state preparation on surface codes: projection scheme} \label{sec:projection}

The expansion scheme above can be used to prepare a rotation state $\ket{r_\varphi}_L$ with a general rotation angle $\varphi$. 
In many early fault-tolerant quantum algorithms, especially the ones related to the Trotterization of Hamiltonians, the value of $\varphi$ is usually small. In Ref.~\cite{choi2023fault}, Choi~et~al. introduced an interesting idea to suppress the state preparation error of $\ket{r_\varphi}_L$ based on the value $\varphi$ and prepare the rotation state $\ket{r_\varphi}_L$ with a trace distance of $O(|\varphi|\cdot p)$. Here, we briefly introduce their idea in \autoref{ssc:BayesianEM}, and then further improve it to $O(|\varphi|\cdot p^2)$ based on the error-structure-tailored $ZZ$-rotation gate construction in \autoref{sec:RZZgate4112}.

\subsection{Bayesian logical error suppression: review} \label{ssc:BayesianEM}

We review the protocol from Ref.~\cite{choi2023fault}, taking the preparation of $\ket{r_\varphi}_L$ in the $d=3$ rotated surface code as an example. In the protocol, we first prepare logical $\ket{+}_L$ state fault-tolerantly on the rotated surface code; then, for a given minimal set of qubits supporting logical-$Z$ operator, we implement single-qubit $Z$-rotation on these qubits. In the noiseless case, the resulting state becomes,
\begin{equation} \label{eq:psi_id}
\begin{aligned}
\ket{\psi_{\mr{id}}} &= [(\cos^3\theta III + i^3 \sin^3\theta ZZZ)  \\ 
&\quad + i \cos^2\theta\sin\theta (ZII+IZI+IIZ) \\
&\quad +  i^2 \cos\theta\sin\theta^2 (IZZ+ZIZ+ZZI) ] \ket{+}_L \\
&= \sqrt{ P_{\mr{pass|id}} } \ket{\psi_{\mr{id,pass}}} + \sqrt{ P_{\mr{fail|id}} } \ket{\psi_{\mr{id,fail}}}.
\end{aligned}
\end{equation}
Here, $\ket{\psi_{\mr{id,pass}}}\propto (\cos^3\theta III + i^3 \sin^3\theta ZZZ)\ket{+}_L$ represents the post-selected state when projected to the code space, occurring with probability $P_{\mr{pass|id}}:= \cos^6\theta + \sin^6\theta$. The failed component $\ket{\psi_{\mr{id,fail}}}$ occurs with probability $P_{\mr{fail|id}}:=1 - P_{\mr{pass|id}}$. The post-selected state simplifies to:
\begin{equation}
\begin{aligned}
\ket{\psi_{\mr{id,pass}}} &= \frac{1}{\sqrt{P_{\mr{pass|id}}}} \left( \cos^3\theta I_L + i^3 \sin^3\theta Z_L \right) \ket{+}_L \\
&= (\cos\varphi I_L + i \sin\varphi Z_L)\ket{+}_L.
\end{aligned}
\end{equation}
where $\varphi:= - \sin^{-1}(\frac{1}{\sqrt{P_{\mr{pass|id}}}} \sin^3\theta)\approx -\theta^3 + O(\theta^5)$. This allows determining the physical rotation angle $\theta$ needed to implement a target logical rotation $\varphi$, projecting to the desired state $\ket{r_\varphi}_L$.

In the above discussion, we have assumed an ideal noiseless implementation. In practice, when additional noise occurs, it may lead to erroneous implementations of rotation gates. Without loss of generality, we consider the effect of Pauli errors. While many low-weight Pauli errors (e.g., Pauli $XII$ errors) occurring before or after the single-qubit rotations are detectable through subsequent surface code syndrome measurements, some specific Pauli errors may remain undetected by the syndrome measurement.
For example, if the Pauli $ZII$ noise occurs, the state we prepare becomes,
\begin{equation} \label{eq:psi_ud1}
\begin{aligned}
&\ket{\psi_{\mr{ud(1)}}} \\
&= \sqrt{P_{\mr{pass|ud(1)}}} \ket{\psi_{\mr{ud(1),pass}}} + \sqrt{P_{\mr{fail|ud(1)}}} \ket{\psi_{\mr{ud(1),fail}}}.
\end{aligned}
\end{equation}
We use the subscript ``$\mr{ud(1)}$'' to denote the weight-$1$ undetectable errors, including the Pauli $ZII$, $IZI$, and $IIZ$ errors.
Applying QED on $\ket{\psi_{\mr{ud(1)}}}$, the post-selected state in the code space is $\ket{\psi_{\mr{ud(1),pass}}}\propto (\cos\theta I_L + i\sin\theta Z_L)\ket{+}_L$, with the probability $P_{\mr{pass|ud(1)}}:= \cos^2\theta \sin^2\theta\approx \theta^2$. Note that $\ket{\psi_{\mr{ud(1),pass}}} = \ket{r_{\varphi_1}}$ with $\varphi_1=\theta$. If weight-$2$ undetectable errors like the Pauli $ZZI$ error occurs, the resulting post-selected state will be different. However, these weight-$2$ will occur with a higher-order probability $P_{\mr{ud(2)}} = O(p^2)$ and pass post-selection with probability $P_{\mr{pass|ud(2)}}\approx \theta^4$, which is negligible in the leading-order error analysis. 

In practice, conditioning on passing the quantum error detection, the overall post-selected resulting state is
\begin{equation} \label{eq:rho_pass}
\rho_{\mr{pass}} = P_{\mr{id|pass}}\, \rho_{\mr{id,pass}} + P_{\mr{ud(1)|pass}}\, \rho_{\mr{ud(1),pass}} + O(|\varphi|^2 p^2),
\end{equation}
where $\rho_{\mr{id,pass}} = \ket{r_\varphi}\bra{r_\varphi}$ and $\rho_{\mr{ud(1),pass}} = \ket{r_{\varphi_1}}\bra{r_{\varphi_1}}$. $P_{\mr{ud(1)|pass}}$ is the Bayesian probability of undetectable erorr given by
\begin{equation} \label{eq:Pud1_pass}
\begin{aligned}
P_{\mr{ud(1)|pass}} = \frac{ P_{\mr{ud(1),pass}} }{ P_{\mr{pass}} } \approx \frac{ P_{\mr{ud(1)}} P_{\mr{pass|ud(1)}} }{P_{\mr{id}} P_{\mr{pass|id}}}.
\end{aligned}
\end{equation}
In the approximation, we only keep the leading order term of the denominator with respect to $p$. Consider a small angle limit of $|\theta|\to 0$, we have $P_{\mr{id}} P_{\mr{pass|id}} \approx 1$ in \autoref{eq:psi_id} and $P_{\mr{pass|ud(1)}}\approx \theta^2$ in \autoref{eq:psi_ud1}, hence $P_{\mr{ud(1)|pass}} \approx P_{\mr{ud(1)}}\cdot \theta^2$.

We can calculate the trace distance~\cite{nielsen2010quantum} between the post-selected state $\rho_{\mr{pass}}$ and the ideal state $\rho_{\mr{id,pass}}$,
\begin{equation} \label{eq:Dtr_rho_pass}
\begin{aligned}
&\quad D_{\mr{tr}}(\rho_{\mr{pass}}, \rho_{\mr{id,pass}}) = \frac{1}{2}\left| \rho_{\mr{pass}} - \rho_{\mr{id,pass}} \right| \\
&\approx \frac{1}{2}P_{\mr{ud(1)|pass}} \left|\, \ket{r_\varphi}\bra{r_\varphi} - \ket{r_{\varphi_1}}\bra{r_{\varphi_1}} \,\right| \\
&\approx P_{\mr{ud(1)}} \frac{ P_{\mr{pass|ud(1)}} }{ P_{\mr{id}} } \sin\Delta_\varphi.
\end{aligned}
\end{equation}
In the approximation of the second and third line, we only keep the leading order terms with respect to $p$. Here, $\Delta_\varphi:= |\varphi - \varphi_1 |\approx \theta^3 +\theta$. In the samll angle limit of $|\theta|\to 0$, we have $D_{\mr{tr}}(\rho_{\mr{pass}}, \rho_{\mr{id,pass}}) \approx P_{\mr{ud(1)}} \theta^2 \cdot \theta \approx P_{\mr{ud(1)}}|\varphi|$ based on \autoref{eq:Pud1_pass}. 

In the above protocol, the undetectable noise occurs with the probability of $P_{\mr{ud(1)}} = O(p)$. Therefore, the error of the rotation state preparation is $O(|\varphi|p)$, suppressed by a factor of $|\varphi|$ due to the Bayesian property of the undetectable error in \autoref{eq:Pud1_pass}. Note that this protocol is still non-fault-tolerant, since it does not improve the final error scaling related to the physical error $p$. In \autoref{ssc:MultiRot}, we improve the error scaling to $O(|\varphi| p^2)$. 

We remark that the choice of trace distance as the error identifier is important. The infidelity~\cite{nielsen2010quantum} between $\rho_{\mr{pass}}$ and $\rho_{\mr{id,pass}}$ scales as $O(p\cdot |\varphi|^{4/3})$~\cite{toshio2024practical}, which is smaller than the trace distance. This discrepancy arises from coherent errors introduced when rotating from the erroneous state $\ket{r_{\varphi_1}}$ to the target state $\ket{r_\varphi}_L$. Later in \autoref{ssc:Prob_Coher_EC}, we will discuss how to convert the coherent error to the Pauli error.

To understand why the trace distance characterize the error of the rotation state preparation, we consider the usage of the state $\rho_{\mr{pass}}$ in \autoref{eq:rho_pass} to implement the logical rotation gate $R_{Z_L}(\varphi)$ based on the circuit in \autoref{fig:RotationInjection}(a). Suppose the measurement outcome $o_x = 0$, the channel we implement on the input state $\ket{\psi}_L$ is the following stochastic rotation channel,
\begin{equation} \label{eq:N_varphi_d}
\begin{aligned}
&\mc{N}_{\varphi}(\cdot) = \mc{E}_{\varphi}\circ \mc{R}_{\varphi} (\cdot) \\
=& P_{\mr{id|pass}} \mc{R}_{\varphi} (\cdot) + P_{\mr{ud(1)|pass}} \mc{R}_{\varphi_1}(\cdot) + O(|\varphi|^2 p^2).
\end{aligned}
\end{equation}
Here, $\mc{E}_{\theta}$ is the coherent error channel after the ideal rotation unitary channel $\mc{R}_{\varphi}(\cdot):= R_{Z_L}(\varphi) \cdot R_{Z_L}(\varphi)^\dagger$. As a noisy quantum gate implementation, we care about its \emph{worst-case} performance, which is characterized by the diamond norm distance of the over-rotation channel $\mc{E}_{\theta}$ to the identity channel,
\begin{equation} \label{eq:rotation_diamond}
\epsilon_\diamond(\mc{E}_{\theta}) = P_{\mr{ud(1)|pass}}\, |\Delta_\varphi|\, \sqrt{1 + \Delta_\varphi^2} \approx |\varphi|P_{\mr{ud(1)}}.
\end{equation}
This provides a worst-case upper bound for the trace distance between the ideal state $\mc{R}_{\varphi} (\rho)$ and the real state $\mc{N}_{\varphi}(\rho)$ for arbitrary input $\rho$. On the other hand, the average channel infidelity - which relates to the infidelity of the prepared state - can only provide an average-case error estimation. Consequently, we primarily focus on trace-distance estimation for the $\ket{r_\varphi}_L$ state preparation.

In the above discussion, we have assumed ideal syndrome detection. In practice, we can suppress measurement errors by performing multiple syndrome measurement rounds and post-selecting those rounds where all syndromes are consistent. This approach suppresses syndrome measurement errors to higher order, ensuring that the leading-order error scaling in preparing $\ket{r_\varphi}_L$ remains $|\varphi|P_{\text{ud}(1)}| = O(|\varphi|p)$.

\subsection{Multi-rotation scheme with error-detectable ZZ-rotation gates} \label{ssc:MultiRot}

The scheme in Ref.~\cite{choi2023fault} can suppress the state preparation error of $\ket{r_\varphi}_L$ by a factor of $|\varphi|$. However, it is not fault-tolerant as it fails to improve the error scaling with respect to $p$. Here, we propose a $1$-fault-tolerant protocol that prepares $\ket{r_\varphi}_L$ with an error of $O(|\varphi| \cdot p^2)$, leveraging the error-detected $ZZ$-rotation gate presented in \autoref{sec:RZZgate4112}.

\begin{figure*}[htbp]
    \centering
    \includegraphics[width=0.88\textwidth]{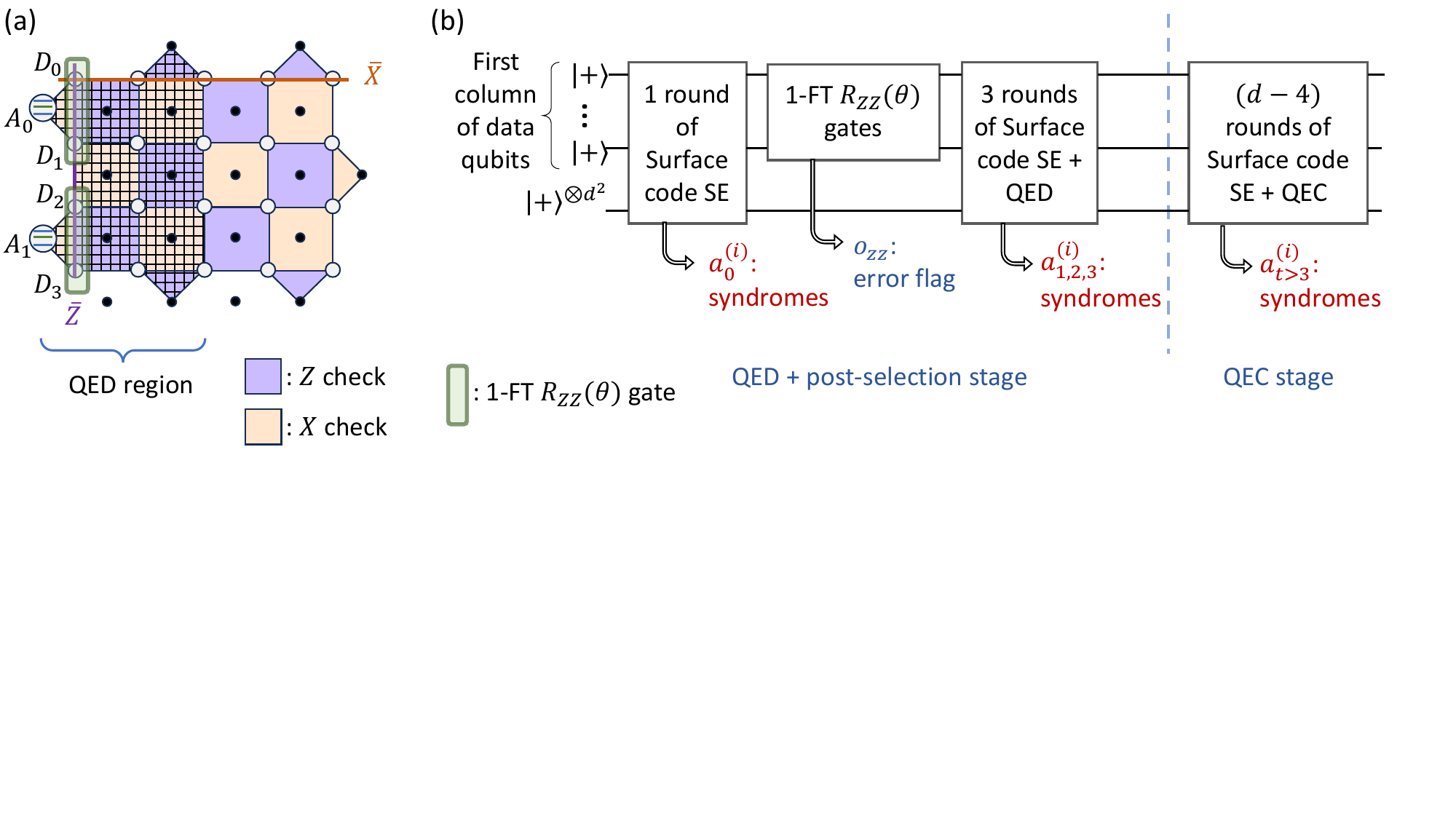}
    \caption{Projection scheme to prepare the rotation state $\ket{r_\varphi}_L$ with a small angle $\varphi$. (a) Example on the $4\times 5$ rotated surface code. The grid region (i.e., the 3 leftmost columns of syndromes) indicates the syndromes used for QED during the projection stage after the 1-FT $R_{ZZ}(\theta)$ gates while the others are used for QEC. (b) Full circuit to prepare the rotation state. Here, SE indicates the syndrome extraction. We first perform 3 rounds of QED post-selection in the QED region, then perform $(d-4)$ rounds of SE and perform QEC using all the $d$ rounds of SE.}
    \label{fig:ProjectionScheme}
\end{figure*}

Without loss of generality, we focus our discussion on preparing the rotation state $\ket{r_\varphi}_L$ in the $d \times (d+1)$ rotated surface code, where the $Z$-distance $d=2k$ is chosen to be even. This even-$d$ assumption simplifies subsequent analysis. We note that preparation of $\ket{r_\varphi}_L$ in the canonical $(d+1) \times (d+1)$ rotated surface code can be achieved by first preparing the state in the $d \times (d+1)$ code then performing code expansion.
Alternative protocols for the direct preparation of $\ket{r_\varphi}_L$ in $(d+1) \times (d+1)$ rotated surface codes are presented in Appendix~\ref{sec:App_extend_rotated}.

Consider preparing a logical $\ket{+}_L$ state on the $d \times (d+1)$ code. For the $d=4$ case (illustrated in \autoref{fig:ProjectionScheme}(a)), we implement a $R_{Z_L}(\varphi)$ gate by applying pairwise $R_{ZZ}(\theta)$ gates on the support of a weight-$d$ logical-$Z$ operator (e.g., qubits $D_0$-$D_3$ in \autoref{fig:ProjectionScheme}(a)), instead of the weight-$1$ $Z$-rotations from \autoref{ssc:BayesianEM}. The noiseless output state becomes:
\begin{equation} \label{eq:psi_id_2}
\begin{aligned}
\ket{\psi_{\mr{id}}} &= [(\cos^2\theta IIII + i^2 \sin^2\theta ZZZZ) \\
&\; + i \cos\theta\sin\theta (IIZZ+ZZII)] \ket{+}_L \\
&= \sqrt{ P_{\mr{pass|id}} } \ket{\psi_{\mr{id,pass}}} + \sqrt{ P_{\mr{fail|id}} } \ket{\psi_{\mr{id,fail}}},
\end{aligned}
\end{equation}
where $P_{\mr{pass|id}}=\cos^4\theta + \sin^4\theta$, $P_{\mr{fail|id}} = 1 - P_{\mr{pass|id}}$ and $\ket{\psi_{\mr{id,pass}}}\propto (\cos^2\theta I_L - \sin^2\theta Z_L)\ket{+}_L \propto R_{X_L}(\pi/4) \ket{r_{\varphi}}_L$ with $\varphi:= -\sin^{-1}\left(\frac{1}{\sqrt{P_{\mr{pass|id}}}} \sin^2\theta\right)\approx \theta^2$.

Suppose we perform syndrome measurement and post-select states with correct syndrome outcomes. This enables preparation of a state approximating the rotation state $\ket{r_\varphi}_L$ up to a Clifford gate $R_{X_L}(\pi/4)$. The complete rotation state preparation protocol is illustrated in \autoref{fig:ProjectionScheme}(b). To suppress syndrome measurement errors, we first execute 3 rounds of syndrome measurements and then post-select only when all syndromes in the designated grid region (\autoref{fig:ProjectionScheme}(b)) are correct.
We restrict post-selection to the three leftmost syndrome columns since they exclusively induce low-order errors. This reduces the QED syndrome count from $O(d^2)$ to $O(d)$, significantly improving success probability.
For the states passing the post-selection, we then utilize all available syndromes (including those outside the grid region) for subsequent QEC procedures.
This ensures comprehensive error correction while maintaining the efficiency benefits of selective post-selection.

More generally, suppose $d=2k$, we denote a chosen weight-$d$ support of the logical $Z_L$ operator to be $\mr{Supp}(Z_L)$ and index the pair of data qubits on $\mr{Supp}(Z_L)$ to be $i=0,1,...,k-1$. In the $i$th pair of data qubits $D_{2i}$ and $D_{2i+1}$, we apply the $ZZ$-rotation gate $R_{ZZ,i}(\theta)$ on them. We can expand the overall effect of $k$ $R_{ZZ}(\theta)$ gates as
\begin{equation} \label{eq:kR_ZZ_expand}
\begin{aligned}
\prod_{i=1}^k R_{ZZ,i}(\theta) &= \sum_{b=0}^{2^k} u_{|b|} (Z_{(2)})^b,
\end{aligned}
\end{equation}
where $b$ is a $k$-bit string, $|b|$ denotes the weight of $b$, $u_w:= i^w \sin^w \theta \cos^{k-w}\theta$ and
\begin{equation}
(Z_{(2)})^b := \prod_{i: b_i=1} Z_{2i} Z_{2i+1}.
\end{equation}
Applying the rotation gate to $\ket{+}_L$, the noiseless resulting state is
\begin{equation} \label{eq:psi_id_m2}
\begin{aligned}
\ket{\psi_{\mr{id}}} &= \prod_{i=1}^k R_{ZZ,i}(\theta) \ket{+}_L = \sum_{b=0}^{2^k} u_{|b|} (Z_{(2)})^b \ket{+}_L \\
&= \sqrt{ P_{\mr{pass|id}} } \ket{\psi_{\mr{id,pass}}} + \sqrt{ P_{\mr{fail|id}} } \ket{\psi_{\mr{id,fail}}},
\end{aligned}
\end{equation}
where
\begin{equation} \label{eq:Ppass_id_psipass}
\begin{aligned}
P_{\mr{pass|id}} &:= \cos^{2k}\theta + \sin^{2k} \theta, \\
\ket{\psi_{\mr{id,pass}}} &\propto (u_{00...0} I_L + u_{11...1} Z_L ) \ket{+}_L.
\end{aligned}
\end{equation}
Therefore, the ideal state after the post-selection is~\cite{toshio2024practical}
\begin{equation} \label{eq:ideal_state_k}
\begin{aligned}
&\ket{\psi_{\mr{id,pass}}} \propto (\cos^k\theta I_L + i^k\sin^k\theta Z_L) \ket{+}_L \\
&=
\begin{cases}
e^{-i\frac{\pi}{4}} R_{X_L}(\frac{\pi}{4})\ket{r_{(-1)^j\varphi}}_L, &  k =2j, \\
\ket{r_{(-1)^j\varphi}}_L, &  k =2j+1.
\end{cases}
\end{aligned}
\end{equation}
Here, the target angle $\varphi$ is related to $\theta$ by $\varphi:= \sin^{-1}(\frac{1}{\sqrt{P_{\mr{pass|id}}}} \sin^k \theta) \approx \theta^k + O(\theta^{k+2})$. 

From \autoref{eq:ideal_state_k}, we can see that the form of the ideal rotation state $\ket{r_\varphi}_L$ is only related to the number of physical rotation gates $k$. Indeed, we can generalize the above discussion to the case where we first group the $d$ data qubits on $\mr{Supp}(Z_L)$ to $k$ groups each containing $m$ qubits, then perform the weight-$m$ $Z$-rotation $R_{Z_k}$ on each group of qubits. After the projection, the ideal state can be related to the target rotation state $\ket{r_\varphi}_L$ by \autoref{eq:ideal_state_k}. This is the $(m,k,d)$-multi-rotation protocol introduced by Toshio~et~al.~\cite{toshio2024practical} as a generalization of Choi's protocol. However, the multi-rotation protocol itself does not guarantee the leading-order error suppression. Specifically, if we adopt the na\"ive non-transversal $R_{ZZ}(\theta)$ gate construction~\cite{toshio2024practical}, the undetectable $ZZ$ error will occur with the probability of $O(p)$, which leads to a target state preparation with the error of $O(|\varphi|p)$.

We now achieve the error scaling of $O(|\varphi|p^2)$ by introducing the error-structure-tailored $R_{ZZ}(\theta)$ gate in \autoref{sec:RZZgate4112} in the multi-rotation protocol. Here, we focus on the case with $m=2$, i.e., $d=2k$. Later in \autoref{fig:SuccProb}, we will discuss the expansion of $R_{ZZ}(\theta)$ gate to $R_{ZZZ}(\theta)$ gate with $m=3$.

\begin{proposition} \label{prop:ProjectionFT}
Consider to prepare the rotation state $\ket{r_\varphi}_L$ on a $d\times (d+1)$ surface code where $d=2k$ by the projection scheme in \autoref{fig:ProjectionScheme}(b), where the 1-FT $R_{ZZ}(\theta)$ gates are implemented by the dispersive coupling Hamiltonian in Proposition~\ref{prop:4112FTdispersive} or the ancilla-based circuit in \autoref{fig:RZZgate_ancilla} and Proposition~\ref{prop:4112FTancilla}. Then, the trace distance of the noisy output state to the ideal state is $O(|\varphi|\cdot p^2)$ conditioning on the successful projection to the code space.
\end{proposition}
The proof of Proposition~\ref{prop:ProjectionFT} can be found in Appendix~\ref{ssc:AppProofProjectionFT}. The intuition behind the proof is that among all the errors with weight less or equal than $2$, only Pauli $Z_{2i}Z_{2i+1}$-type errors or the errors with support on the Pauli $Z_{2i}Z_{2i+1}$ space will introduce undetectable logical errors. Thanks to the error-structure-tailored FT design in the $R_{ZZ}(\theta)$ gates (Section~\ref{sec:RZZgate4112}), we can prove that first-order errors do not have support on the Pauli $Z_{2i}Z_{2i+1}$ operator space. Consequently, undetectable errors scale as $O(p^2)$. Furthermore, Bayesian error analysis similar to Section~\ref{ssc:BayesianEM} indicates that when conditioning on passing post-selection, the trace distance of the resulting state to the ideal $\ket{r_\varphi}_L$ state scales as $O(|\varphi|\cdot p^2)$. Specifically, we have shown that the trace distance of the prepared state $\rho$ to $\ket{r_\varphi}_L$ can be estimated by
\begin{equation} \label{eq:Dtr_rho_rvarphi}
D_{\mr{tr}}(\rho,\ket{r_\varphi}_L) = P_{\mr{ud(2)}}|\varphi| + O(|\varphi|^2 p^4),
\end{equation}
where $P_{\mr{ud(2)}}$ is the probability of the weight-2 undetectable error occurs,
\begin{equation} \label{eq:P_mrud2_main}
\begin{aligned}
P_{\mr{ud(2)}} &= \sum_{i=1}^k \Pr(Z_{2i} Z_{2i+1}, I_{\backslash{2i,2i+1}}) \\
&\leq \sum_{i=1}^k \Pr(Z_{2i} Z_{2i+1}).
\end{aligned}
\end{equation}
Here, $\Pr(Z_{2i} Z_{2i+1}, I_{\backslash{2i,2i+1}})$ denotes the probability of a Pauli $Z$ error occurring on qubits $2i$ and $2i+1$ supporting $\bar{Z}$ with no detectable errors on other qubits within the post-selection region defined in \autoref{fig:ProjectionScheme}(a), during the full QED procedure in \autoref{fig:ProjectionScheme}(b). Meanwhile, $\Pr(Z_{2i} Z_{2i+1})$ refers to the probability of the Pauli error $Z_{2i} Z_{2i+1}$, which is of order $O(p^2)$.
Thus, $P_{\mathrm{ud(2)}} = O(k \cdot p^2)$. However, as $k$ increases, the value of $\Pr(Z_{2i} Z_{2i+1}, I_{\backslash{2i,2i+1}})$ becomes significantly smaller than $\Pr(Z_{2i} Z_{2i+1})$. Later we will see that the value of $P_{\mathrm{ud(2)}}/p^2$ increases slower than a linear function of $k$ when we increase $k$. 

The implementation of the projection scheme for a general surface code with $d=2k$ is shown in \autoref{fig:ProjectionScheme}(b), where we: (i) Prepare the logical $\ket{+}_L$ state; (ii) Perform $k$ pairwise 1-FT $R_{ZZ}(\theta)$ gates on the logical $Z$ edge; (iii) Perform 3 rounds of syndrome projection and use the three leftmost syndrome columns for QED; (iv) If the state passes post-selection: we perform additional syndrome measurements, execute QEC and store it for later usage.

To demonstrate the performance of the projection scheme, we perform a circuit-level noise simulation on a $6\times 7$ rotated surface code and set the target angle to be $\varphi = 1\cdot 10^{-3}$ $\mr{rad}$. The results are shown in \autoref{fig:projection_num}. 
Similar to the expansion protocol simulation, we ignore the error structure in all circuit elements except the $R_{ZZ}(\theta)$ rotation gates while assuming worst-case canonical depolaring noise; for the $R_{ZZ}(\varphi)$ gates, we first perform Lindbladian simulation in QuTiP using parameters similar to \autoref{fig:Num4112coherent} then extract the Pauli error model for Stim circuit-level simulation. 
Here, since $\varphi$ is now small, which corresponds to non-Clifford implementation of $R_{ZZ}(\theta)$ gate, we cannot directly simulate the non-Clifford circuits. We consider similar numerical techniques in Ref.~\cite{toshio2024practical} to faithfully estimate the trace distance of the output noisy state to the ideal rotation state $\ket{r_\varphi}_L$. The simulation details can be found in Appendix~\ref{ssc:AppProjNum}. 

\begin{figure}[htbp]
    \includegraphics[width=\linewidth]{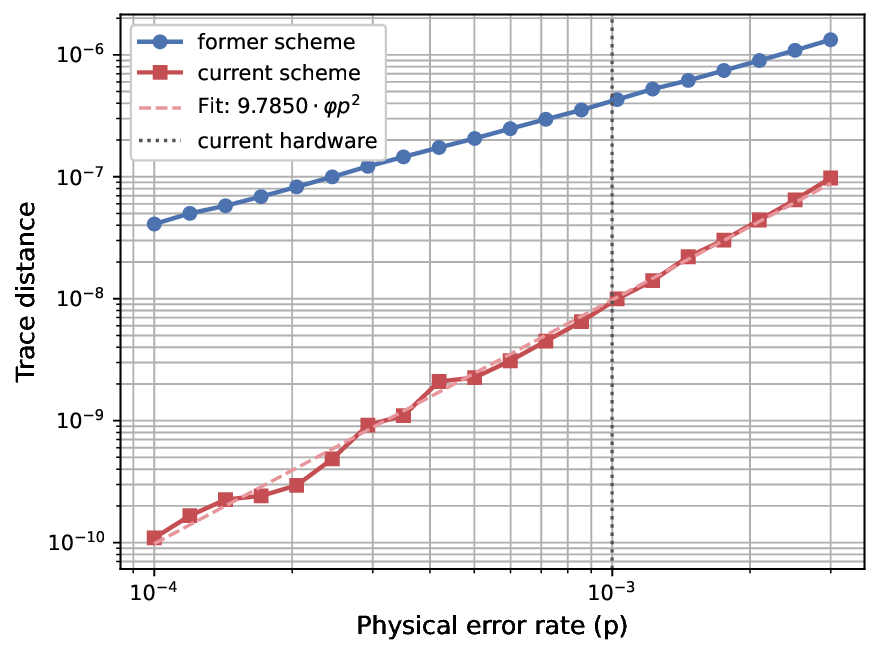}
    \caption{Trace distance of the projection scheme to prepare a logical rotation state $\ket{r_\varphi}_L$ on a $6\times 7$ rotation surface code with the target angle $\varphi=1\cdot 10^{-3}$. The numerical simulation is performed with a maximal sample number of $1\cdot 10^8$ with the methods similar to Ref.~\cite{toshio2024practical}, which is detailed in Appendix~\ref{ssc:AppProjNum}.} 
    \label{fig:projection_num}
\end{figure}

In \autoref{fig:projection_num}, the former scheme (the blue curve) implement the canonical $R_{ZZ}(\theta)$ gates with the depolarizing noise model, whose state preparation error is $O(|\varphi|\cdot p)$. Thanks to the 1-FT property of our $R_{ZZ}(\theta)$ gates, we can prepare $\ket{r_\varphi}_L$ with the trace distance of $O(|\varphi|\cdot p^2)$, illustrated by the red curve. The overhead $c(k)$ is about $9.8$ when $k=3$. Furthermore, under the current hardware condition with $p=1\cdot 10^{-3}$, one can prepare $\ket{r_\varphi}_L$ with an trace distance error smaller than \blue{$1\cdot 10^{-8}$}. 
We remark that the trace distance estimation shown in \autoref{fig:projection_num} only incorporates errors introduced during the QED and post-selection stage in \autoref{fig:ProjectionScheme}(b). This is reasonable when the code distance gets larger: recalling that the surface code distance is $d=2k$, as $k$ increases, errors introduced during the QEC stage can be suppressed exponentially.

To ensure that the projection scheme can be performed when the surface code distance gets larger, we check the the dependence of the trace distance of the prepared state to $\ket{r_\varphi}_L$ to the number of rotation gates $k$. From \autoref{eq:Dtr_rho_rvarphi} and \autoref{eq:P_mrud2_main} we have
\begin{equation} \label{eq:Dtr_Pud2}
D_{\mr{tr}}(\rho,\ket{r_\varphi}_L) \approx P_{\mr{ud(2)}}|\varphi| = c(k) p^2 \cdot |\varphi|.
\end{equation}
Therefore, we only need to check how the coefficient $c(k):=P_{\mr{ud(2)}}/p^2 \approx D_{\mr{tr}}(\rho,\ket{r_\varphi}_L)/(p^2 \cdot |\varphi|) $ increases with $k$. 
\autoref{fig:error_scale} illustrates the relationship of $c(k)$ to $k$. We can see that while $c(k)$ increases slower then a linear function of $k$, $c(k) \leq 3.3\cdot k$ serves as a good upper bound \blue{when $k\geq 3$}. 
The reason we consider the rotation gate number $k \geq 3$ is twofold: (i) For $k = 0$ or $1$, the protocol becomes trivial. (ii) For $k = 2$, the error rate of the projection scheme is significantly higher because the erroneous rotation angle $\varphi_1$ is fixed at $\pi/4$, leading to substantial deviations in the output.

\begin{figure}[htbp]
    \centering
    \includegraphics[width=0.4\textwidth]{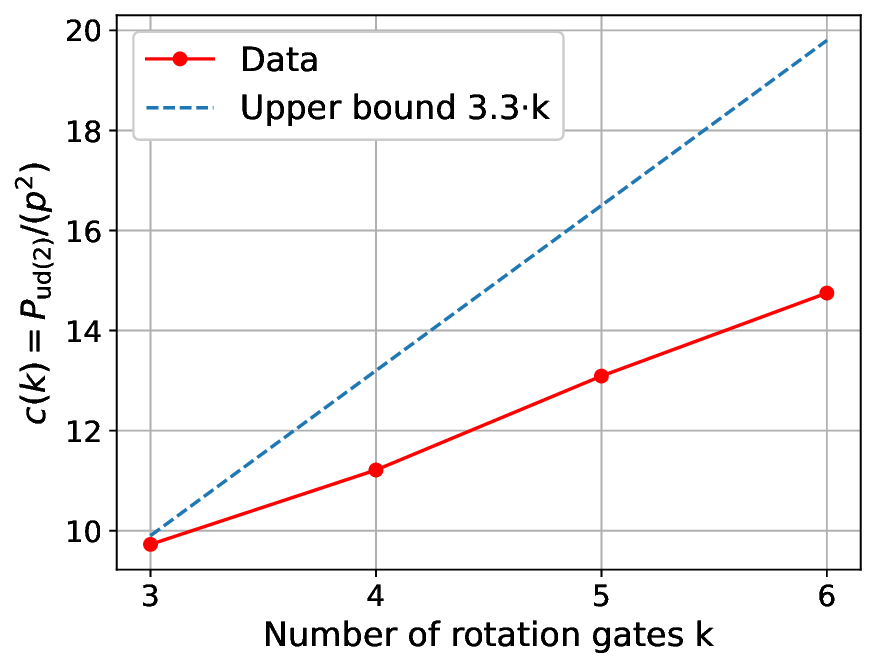}
    \caption{Dependence of $c(k)=P_{\mr{ud(2)}}/p^2$ with the number of rotation gates $k$. Recall that the size of the rotated surface code is $2k\times (2k+1)$.}
    \label{fig:error_scale}
\end{figure}

Another important issue is to ensure the successful probability of preparing the target state remains reasonably high as the code distance $d$ increases. When the number of rotation gates $k$ becomes larger, the success probability decreases due to two main reasons:
\begin{enumerate}
\item When the code distance $d$ increases, the number of detectable error patterns grows.
\item Even in the absence of errors, the probability of projecting the state to the code space $p_{\mathrm{pass|id}} \approx 1 - k\theta^2$ decreases as $k$ increases (see Appendix~\ref{ssc:AppProofProjectionFT} for details). Recall that $\theta \approx \varphi^{1/k}$ when $|\varphi| \ll 1$.
\end{enumerate}
To increase the success probability $p_{\mathrm{pass|id}}$, we follow the approach in Ref.~\cite{toshio2024practical} and expand our $ZZ$-rotation gate non-transversally to a $ZZZ$-rotation gate. See Appendix~\ref{sec:App_extend_rotated} for a detailed discussion.

In \autoref{fig:SuccProb}, we show the success probability of implementing the projection scheme with rotated surface codes of distances $d=12$ and $d=18$. We compare two methods: the $m=2$ method uses the original dispersive $ZZ$-rotation gate design from \autoref{sec:RZZgate4112}, while the $m=3$ method uses the $ZZZ$-rotation constructed from dispersive $ZZ$ rotations via non-transversal expansion.
The success probability for the $m=2$ method is approximately $12.0\%$ for $d=12$ and $1.8\%$ for $d=18$. With the extension, these probabilities increase to $18.7\%$ and $5.4\%$, respectively. Note that the QED stage in the projection scheme (\autoref{fig:ProjectionScheme}(b)) involves only 4 rounds of syndrome detection and physical $ZZ$-rotation gates. Consequently, we can prepare the $\ket{r_\varphi}$ rotation state for $d=12$ within an average of 22 rounds of syndrome extraction without additional footprint.

\begin{figure}[htbp]
    \includegraphics[width=\linewidth]{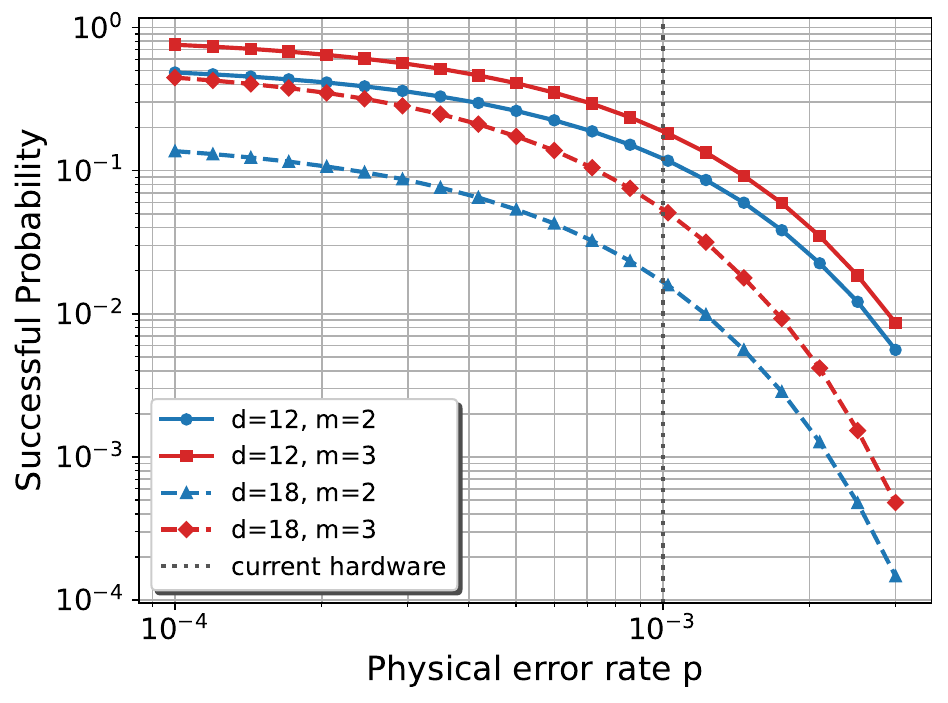
}
    \caption{Successful probability of the projection scheme with the surface code of $d=12$ and $d=18$ with respect to different physical error rate $p$. Here, we fix the rotation angle $\varphi=1\cdot 10^{-3}$. The $m=2$ method use the dispersive $ZZ$-rotation gate design from \autoref{sec:RZZgate4112}, while the $m=3$ method use the extended $ZZZ$-rotation gate.} 
    \label{fig:SuccProb}
\end{figure}

\section{Error mitigation and noise estimation} \label{sec:Error_Count}

In \autoref{sec:projection} we have shown that the projection scheme with the error-structure-tailored FT design and the Bayesian error suppression~\cite{choi2023fault,toshio2024practical} can be used to prepare small-angle rotation state $\ket{r_\varphi}_L$ with a low error rate of $O(|\varphi|\cdot p^2)$, which is below \blue{$1\cdot 10^{-8}$} under the current hardware condition. In this section, we estimate the number of quantum rotation gates we can implement with the prepared rotation state ancillae under proper quantum error mitigation techniques. 

Suppose we have prepared the small-angle ancilla state $\ket{r_\varphi}_L$ following the projection scheme in Sec.~\ref{sec:projection}. If we inject the small-rotation state into a target quantum state by the circuits in \autoref{fig:expansion_num} and succeed in the first trial (recall that the rotation state injection always succeeds with a probability of $1/2$), we will implement the following noisy rotation channel,
\begin{equation} \label{eq:N_varphi_d_QEM}
\begin{aligned}
\mc{N}_{\varphi}(\cdot) = (1- P_L) \mc{R}_{\varphi}(\cdot) + P_L \mc{R}_{\varphi_1}(\cdot) + O(|\varphi|^2 p^4),
\end{aligned}
\end{equation}
where $\mc{R}_{\varphi}(\cdot)$ indicates the ideal rotation channel with the target small angle $\varphi$ while $\mc{R}_{\varphi_1}(\cdot)$ is the wrong rotation with the angle $\varphi_1$. 
When $\varphi$ is small, we have $\varphi \approx \theta^k$ and $\varphi_1 \approx -\theta^{k-2}$. The wrong rotation probability $P_L$ is the conditional probability when the undetected error passes the syndrome detection 
\begin{equation} \label{eq:P_L}
\begin{aligned}
P_L &= P_{\mr{ud(2)|pass}} = \frac{ P_{\mr{ud(2)}} P_{\mr{pass|ud(2)}} }{ P_{\mr{pass}} }\\
&\approx P_{\mr{ud(2)}} \theta^2 = O(p^2 \theta^2).
\end{aligned}
\end{equation}
The derivation of \autoref{eq:N_varphi_d_QEM} is in the proof of Proposition~\ref{prop:ProjectionFT}.

We can decompose the rotation channel in \autoref{eq:N_varphi_d_QEM} as $\mc{N}_{\varphi}(\cdot) = \mc{E}_{\varphi} \circ \mc{R}_{\varphi}(\cdot)$ where
\begin{equation} \label{eq:E_varphi_d_QEM}
\mc{E}_{\varphi}(\cdot) = (1- P_L) \mc{I}(\cdot) + P_L \mc{R}_{\Delta\varphi}(\cdot) + O(|\varphi|^2 p^4),
\end{equation}
the noise induced by the rotation channel, which is a coherent stochastic rotation channel. Here, $\mc{I}(\cdot)$ is the identity channel and $\Delta\varphi := \varphi_1 - \varphi = O(\theta^{k-2})$.
The size of the noise $\mc{E}_{\varphi}(\cdot)$ is characterized by the diamond-norm distance $\epsilon_\diamond(\mc{E}_{\varphi}):=\frac{1}{2}\|\mc{E}_{\varphi}-\mc{I}\|_\diamond$. Similar to \autoref{eq:rotation_diamond}, we have
\begin{equation}
\epsilon_\diamond(\mc{E}_{\varphi}) = P_{\mr{ud(2)|pass}} |\Delta_\varphi| \sqrt{1+\Delta_\varphi^2} \approx |\varphi| P_L.
\end{equation}

\autoref{eq:E_varphi_d_QEM} characterizes the noise introduced by the imperfect resource state. To provide a complete noise estimation, the following two issues must be addressed:
\begin{enumerate}
\item \autoref{eq:N_varphi_d_QEM} assumes the injection succeeds on the first attempt. In practice, the RUS injection protocol from \autoref{sec:expansion} and \autoref{fig:RotationInjection} must be employed. A full noise analysis of this iterative procedure is required.
\item The noise channel in \autoref{eq:E_varphi_d_QEM} is a coherent stochastic rotation channel. For practical error mitigation and streamlined resource estimation, this should be converted to a standard Pauli noise channel compatible with Pauli error cancellation techniques.
\end{enumerate}
We now analyze these two problems. 

\subsection{Noise channel after the repeat-until-success procedure} \label{ssc:RUSnoise}

Consider the RUS protocol for rotation-gate injection as described in \autoref{sec:expansion} and illustrated in \autoref{fig:RotationInjection}. For trial numbers $K = 1, 2, 3, \dots$, we implement the injection with rotation angles that scale exponentially with $K$: $\varphi^{(K)} = 2^{K-1}\varphi$, where $\varphi^{(1)} = \varphi$, $\varphi^{(2)} = 2\varphi$, $\varphi^{(3)} = 4\varphi$, and so on, until the operation succeeds.
When the injection succeeds on the $K$th trial, the resulting rotation channel is
\begin{equation}
\begin{aligned}
\mc{N}^{K}_\varphi(\cdot) :&= \mc{N}_{2^{K-1}\varphi}\circ \mc{N}_{-2^{K-2}\varphi} \circ ... \circ \mc{N}_{-2\varphi} \circ \mc{N}_{-\varphi}(\cdot)\\
&= \mc{E}^K_\varphi \circ \mc{R}_\varphi(\cdot),
\end{aligned}
\end{equation}
where $\mc{E}^{K}_\varphi(\cdot)$ is an effective error channel that describes the accumulated error,
\begin{equation}
\mc{E}^K_\varphi(\cdot) = \mc{E}_{2^{K-1}\varphi}\circ \mc{E}_{-2^{K-2}\varphi} \circ ... \circ \mc{E}_{-2\varphi} \circ \mc{E}_{-\varphi}(\cdot).
\end{equation}

We remark that, when the rotation angle $\varphi^{j}$ where $j=1,2,...K$ goes beyond the region of $[-\frac{\pi}{8}, \frac{\pi}{8})$, we can always reduce the value of this angle into this region by applying logical $S$ gate. We define the wrapping function $\Lambda(x) := (x+ \frac{\pi}{8} )\text{ mod } \frac{\pi}{4} - \frac{\pi}{8}$ which ensures that $|\Lambda(x)|\leq \frac{\pi}{8}$.
Taking the wrapping into account, we can define the average channel $\mc{N}_\varphi$ over any possible $K$,
\begin{equation} \label{eq:tildeN}
\tilde{\mc{N}}_\varphi(\cdot) := \sum_{K=1}^\infty \left(\frac{1}{2} \right)^K \mc{N}_\varphi^K(\cdot) = \tilde{\mc{E}}_\varphi \circ \mc{R}_\varphi(\cdot),
\end{equation}
where $\tilde{\mc{E}}_\varphi$ denotes the effective error channel introduced by the whole RUS process,
\begin{equation}
\tilde{\mc{E}}_\varphi(\cdot) := \sum_{K=1}^\infty \left(\frac{1}{2}\right)^K \mc{E}_\varphi^K(\cdot).
\end{equation}
Following the analysis in Appendix~E of Ref.~\cite{toshio2024practical}, we have
\begin{equation} \label{eq:epsilon_tildeE}
\epsilon_\diamond(\tilde{\mc{E}}_\varphi) = O(|\varphi| P_L).
\end{equation}
In what follows, we discuss how to mitigate the error $\tilde{\mc{E}}_\varphi$.


\subsection{Probabilistic coherent error cancellation} \label{ssc:Prob_Coher_EC}

To mitigate the error with the form of \autoref{eq:E_varphi_d_QEM}, we consider to introduce the probabilistic coherent error cancellation techniques introduced by Toshio~et.~al.~\cite{toshio2024practical}. Consider the noise rotation channel of $\mc{N}_{\varphi}(\cdot)$ in \autoref{eq:N_varphi_d_QEM}. After the implementation of $\mc{N}_{\varphi}(\cdot)$, we append the following compensation channel,
\begin{equation} \label{eq:C_varphi_d}
\begin{aligned}
\mc{C}_{\varphi}(\cdot) &= (1-P_L)\, \mc{I}(\cdot) + P_L \,\tilde{\mc{N}}_{-\Delta\varphi}(\cdot) \\
&\approx (1-P_L) \, \mc{I}(\cdot) + P_L \, \mc{R}_{-\Delta\varphi}(\cdot) + O(|\varphi| p^4).
\end{aligned}
\end{equation}
Here, $\tilde{\mc{N}}_{-\Delta\varphi}(\cdot)$ is the averaged noisy small rotation channel in \autoref{eq:tildeN} with the target angle $-\Delta\varphi=\varphi - \varphi_1$. In the second line, we have used the fact that $P_L=P_{\mr{ud(2)|pass}}=O(p^2\theta^2)$ from \autoref{eq:P_L} and $\tilde{\mc{E}}_{-\Delta\varphi}=O(|\Delta\varphi|P_L)=O(|\theta^{k-2}| P_L)$ from \autoref{eq:epsilon_tildeE}.

Appending $C_{\varphi}$ after the noise channel $E_{\varphi}$ we have,
\begin{equation} \label{eq:Ec_varphi_rho}
\begin{aligned}
&\quad \mc{E}^c_{\varphi}(\rho) = \mc{C}_{\varphi}\circ \mc{E}_{\varphi} (\rho) \\
& \approx (1-P_L)^2 \rho + P_L^2 \mc{R}_{-\Delta\varphi}\circ \mc{R}_{\Delta\varphi} (\rho)  \\
& \quad + P_L(1-P_L) ( \mc{R}_{-\Delta\varphi} + \mc{R}_{\Delta\varphi} ) + O(|\varphi| p^4) \\
& \approx ( 1 - 2 P_L \sin^2 \Delta\varphi ) \rho + 2 P_L \sin^2 \Delta\varphi\,Z\rho Z + O(|\varphi| p^4),
\end{aligned}
\end{equation}
which becomes an incoherent dephasing error and hence easy to be mitigated by probabilistic error cancellation methods~\cite{temme2017error,endo2018practical}. Then, the worst-case error of the noise channel $\mc{E}^c_{\varphi}(\rho)$ given by the diamond distance is
\begin{equation} \label{eq:epsilon_E}
\epsilon_\diamond(\mc{E}^c_{\varphi}) = 2 P_L \sin^2(\Delta\varphi) \approx 2 P_{\mr{ud(2)}} \theta^{2k-2},
\end{equation}

Now, we consider to embed the above coherent error cancellation techniques to the RUS process. Suppose the rotation-state injection succeed in the $K$th trial, we append the compensation channel after each rotation gate,
\begin{equation}
\begin{aligned}
\mc{N}_\varphi^{c,K}(\cdot) :&= \mc{N}_{2^{K-1}\varphi}^c \circ \mc{N}_{-2^{K-2}\varphi}^c \circ ... \circ \mc{N}_{-2\varphi}^c \circ \mc{N}_{-\varphi}^c \\
&=: \mc{E}_\varphi^{c,K} \circ \mc{R}_\varphi(\cdot), \\
\mc{E}_\varphi^{c,K}(\cdot) :&= \mc{E}_{2^{K-1}\varphi}^c \circ \mc{E}_{-2^{K-2}\varphi}^c \circ ... \circ \mc{E}_{-2\varphi}^c \circ \mc{E}_{-\varphi}^c.
\end{aligned}
\end{equation}
The averaged rotation and noise channel of the RUS process is then given by
\begin{equation}
\begin{aligned}
\tilde{\mc{N}}_\varphi^{c}(\cdot) &= \sum_{K=1}^\infty \left(\frac{1}{2}\right)^{K} \mc{N}_\varphi^{c,K}(\cdot), \\
\tilde{\mc{E}}_\varphi^{c}(\cdot) &= \sum_{K=1}^\infty \left(\frac{1}{2}\right)^{K} \mc{E}_\varphi^{c,K}(\cdot).
\end{aligned}
\end{equation}
Ignoring the higher-order terms, we can write out the explicit form of the averaged noise channel,
\begin{equation} \label{eq:tildeE_varphi_rho_avg}
\tilde{\mc{E}}_\varphi^{c}(\rho) = (1- \tilde{P}_L) \rho + \tilde{P}_L Z\rho Z.
\end{equation}
With the averaged logical error rate
\begin{equation} \label{eq:tildeP_L_varphi}
\tilde{P}_L(\varphi) = \sum_{K=1}^\infty \left( \frac{1}{2}\right)^K \sum_{M=1}^K \epsilon_\diamond(\mc{E}_{2^{M-1}\varphi}^c).
\end{equation}
From \autoref{eq:epsilon_E} we can see that, $\epsilon_\diamond(\mc{E}^c_{2^{n-1}\varphi})$ scales as $O(|\theta|^{2k-2} P_{\mr{ud(2)}} ) = O(|\varphi|^{2-2/k} P_{\mr{ud(2)}} )$ when the parameter $n$ is small. However, when the trial number $K$ satisfies $2^{K-1}\varphi \geq 1$, the corresponding accumulated error $\sum_{n=1}^K \epsilon_\diamond(\mc{E}^c_{2^{n-1}\varphi})$ cannot be suppressed by the value of $\varphi$ and reaches $O(P_{\mr{ud(2)}})$. Since the above case appears in the RUS process with the probability of $2^{-K}\approx O(|\varphi|)$, the accumulated error scales with $O(|\varphi| P_{\mr{ud(2)}})$. 

We can then express the averaged error of the RUS process with coherent error mitigation as
\begin{equation} \label{eq:tildeP_Lphi_alpha}
\tilde{P}_L(\varphi) = \alpha_{\mr{RUS}} |\varphi| p^2.
\end{equation}
The prefactor $\alpha_{\mr{RUS}}$ indicates the cost overhead introduced by the RUS and coherent error mitigation procedure, which is determined by \autoref{eq:epsilon_E} and \autoref{eq:tildeP_L_varphi}. When the value of $P_{\mathrm{ud(2)}}$ is given, we can numerically estimate $\alpha_{\mr{RUS}}$. Recall from \autoref{eq:Dtr_Pud2} that $P_{\mathrm{ud(2)}} = c(k)\cdot p^2$. Here, we choose to set the overhead $c(k)=14.7$ based on the value of $k=6$ in \autoref{fig:error_scale}. This moderate $k$ value is suitable for the projection scheme with a reasonable successful probability. From \autoref{fig:SuccProb} we can see that, with the extended protocol of $m=3$ and $k=6$, we can run a surface code of $d=18$ with a successful probability of $5.4\%$.

\autoref{fig:alpha_RUS} indicates the dependence of the value $\alpha_{\mr{RUS}}$ with respect to different target rotation angle $\varphi$. The complicated curve shape is due to the angle-dependent wrapping procedure introduced in \autoref{ssc:RUSnoise}~\cite{toshio2024practical}. We can see that the maximum value of the error prefactor $\alpha_{\mr{RUS}}$ is close to $91$. We will use this value for the later resource estimation.

\begin{figure}[htbp]
    \centering
    \includegraphics[width=0.45\textwidth]{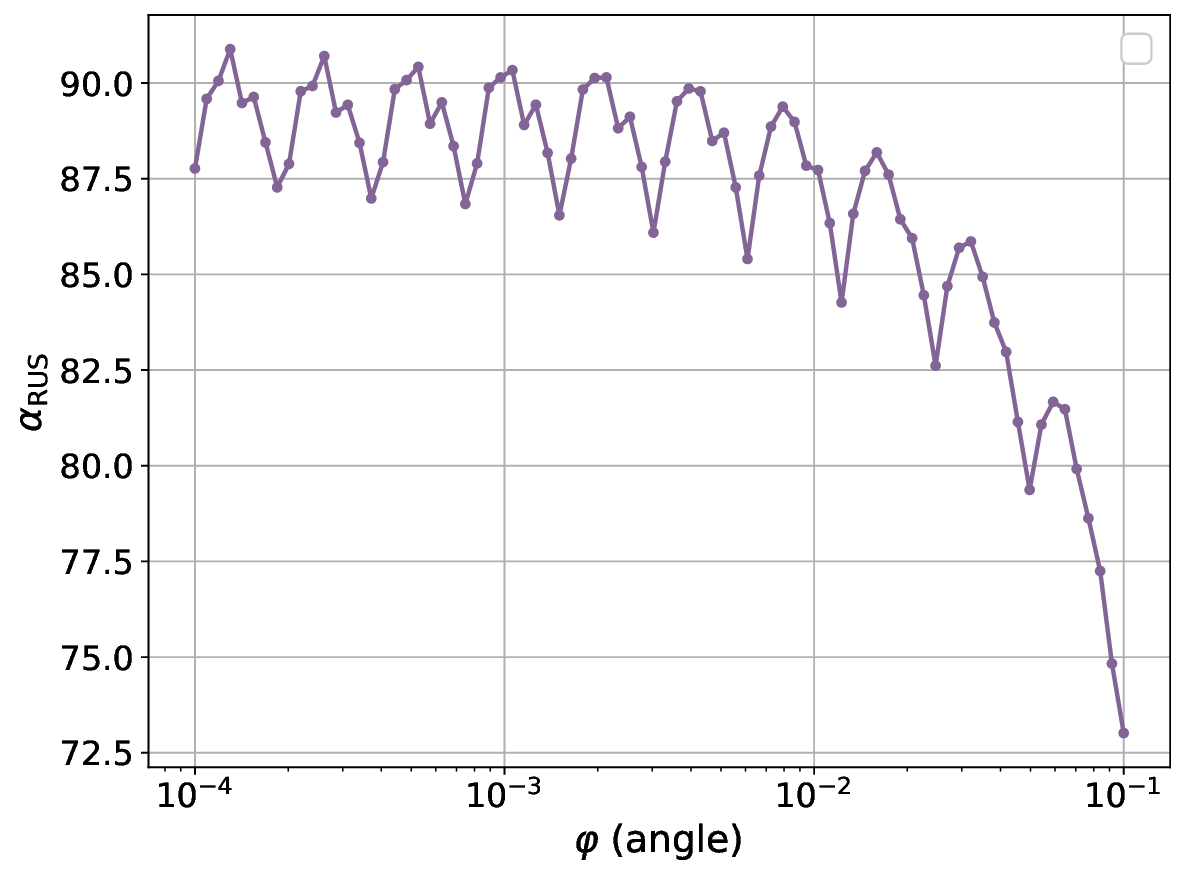}
    \caption{Dependence of the error prefactor $\alpha_{\mr{RUS}}$ of the avaraged gate error $\tilde{P}_L(\varphi)$ in \autoref{eq:tildeP_Lphi_alpha} with respect to the target rotation angle $\varphi$.}
    \label{fig:alpha_RUS}
\end{figure}

\subsection{Controll error cancellation}

In the above discussion, we implicitly assume that the error originates from the incoherent noise procedure of the implemented gates, especially the $R_{ZZ}(\theta)$ gates. However, in the actual devices, the control errors due to calibration errors, crosstalks and so on may leads to unwanted effects on the prepared $\ket{r_\varphi}$ states. Usually, the control errors can be suppressed using dynamical decoupling techniques~\cite{vandersypen2005NMR}. 
Here, we discuss the control error cancellation methods in Sec.~V of Toshio~et~al.~\cite{toshio2024practical} to cancel the remaining control error by a randomized transversal rotation.

Suppose the actual implemented rotation gates in the projection scheme in \autoref{fig:ProjectionScheme} become
\begin{equation}
\prod_{i=1}^k R_{ZZ,i}(\theta) \leftarrow \prod_{i=1}^k R_{ZZ,i}(\theta + \phi_i),
\end{equation}
where $\phi_i$ indicates the control error which satisfies $\phi_i\ll \theta$. 
We now discuss two different types of control errors.
\begin{enumerate}
\item The angle value of $\phi_i$ is random, centered around 0. This type of control error usually occurs in the experiment when the calibration is done with limited accuracy or the strength of the driving pulses are unstable with fluctuations. 
In this case, the effect of $\phi_i$ to the target rotaion will be suppressed to the second order $\phi_i^2$.

To see that, we consider a worst-case scenario where $\phi_i$ will randomly filp from two values $\pm\tilde{\phi}_i$ where $\tilde{\phi}_i$ denotes the largest value in the distribution of $\phi_i$. We have
\begin{equation}
\prod_{i=1}^k R_{ZZ,i}(\theta) \leftarrow \prod_{i=1}^k R_{ZZ,i}(\theta + (-1)^{m_i} \tilde{\phi}_i).
\end{equation}
Based on the calculation in Sec.~V of Ref.~\cite{toshio2024practical}, we have that the target rotation value $\tilde{\phi}$ with be close to the target value of $\varphi$ with a relative error of 
\begin{equation}
\tilde{\varphi} = \varphi\cdot(1 + \sum_i \tilde{\phi}_i^2/2).
\end{equation}

\item The angle value of $\phi_i$ is fixed and independent of the choice of $\theta$. This type of control error usually occurs on the gates before and after the real rotation gate. For example, the control errors on the dispersize coupled $CZ$ gates in the ancilla-based $R_{ZZ}(\theta)$ gate in \autoref{fig:RZZgate_ancilla}. 
In this case, we can introduce an active randomized multirotation design~\cite{toshio2024practical} to suppress the effect of $\phi$ to the second order. We actively randomize the implementation of the rotation gates by  
\begin{equation}
\prod_{i=1}^k R_{ZZ,i}(\theta) \leftarrow \prod_{i=1}^k R_{ZZ,i}((-1)^{n_i}\theta + \phi_i),
\end{equation}
where $n_i(=0,1)$ is a parameter indicating whether to flip the direction of the rotation gate or not. After this randomization operation, the target rotation angle $\tilde{\varphi}$ will be close to the target value of $\varphi$ with a relative error of 
\begin{equation}
\tilde{\varphi} = \varphi\cdot(1 + \sum_i \phi_i^2/2).
\end{equation}

\end{enumerate}
Recall from \autoref{eq:tildeP_Lphi_alpha} that the incoherent error is $\tilde{P}_L(\varphi)=\alpha_{\mr{RUS}}|\varphi| p^2$. Therefore, as long as the value of $\phi_i^2$ is much smaller than $p^2$, the physical noises $\tilde{P}_L(\varphi)$ will be the dominant error source and we will ignore the effect of the small control error in the later resource estimation.

\subsection{Resource estimation}

In \autoref{ssc:Prob_Coher_EC} we have shown that, by appling the coherent error cancellation during the whole RUS procedure, we can implement the rotation gate $R_{Z_L}(\varphi)$ with a Pauli noise channel whose diamond norm distance to the identity channel is given by \autoref{eq:tildeP_Lphi_alpha},
$$ \tilde{P}_L(\varphi) = \alpha_{\mr{RUS}}|\varphi| p^2, $$
where the prefactor $\alpha_{\mr{RUS}}$ is numerically shown to be smaller than $91$. Since now the noise channel is a Pauli channel, we can then implement the canonical probabilistic error cancellation techniques~\cite{temme2017error,endo2018practical} to mitigate these residue errors.

Recall from \autoref{eq:tildeE_varphi_rho_avg} that the Pauli noise channel after the coherent error cancellation is a dephasing channel $\tilde{\mc{E}}^c_\varphi$ with the dephasing rate $\tilde{P}_L$. Pauli error cancellation is to implement the inverse map,
\begin{equation}
\tilde{\mc{E}}^{c,-1}_\varphi(\rho) =\gamma ( (1 - \tilde{P}_L) \rho - \tilde{P}_L Z\rho Z ).
\end{equation}
Here $\gamma = 1/(1-2 \tilde{P}_L)$. We have $\tilde{\mc{E}}^{c,-1}_\varphi \circ \tilde{\mc{E}}^{c}_\varphi(\rho) = \rho$.
The inverse map $\tilde{\mc{E}}^{c,-1}_\varphi(\rho)$ itself is not a CPTP map since it contains a term with negative probability, hence cannot be directly implemented by quantum circuits. However, if we consider the estimation of the expectation value $\braket{O}=\tr(\mc{C}(\rho) O)$ of any observable $O$ at the end of the quantum circuit $\mc{C}_{\mc{N}}$ that contains the rotation gate, which is the usual scenario for many early FT algorithms, we can implement the inverse map effectively by 
\begin{equation} \label{eq:braketO}
\braket{O} \approx \gamma \left( (1-\tilde{P}_L) \braket{O}_{\tilde{\mc{E}}^{c}_\varphi} + \tilde{P}_L ( -\braket{O}_{\mc{Z}\circ\tilde{\mc{E}}^{c}_\varphi} )  \right),
\end{equation}
where $\mc{Z}(\rho) = Z\rho Z$ is the Pauli $Z$ channel. As a result, we can achieve the unbiased estimation of $\braket{O}$ by applying Pauli $Z$ operator with a probability of $\tilde{P}_L$. This comes with an extra $\gamma^2$ sampling cost, as the variance of the estimator in \autoref{eq:braketO} is amplified by $\gamma^2$.

Now, suppose we are going to perform $N_g$ rotation gates $\{R_{Z_L}(\varphi_1), R_{Z_L}(\varphi_2), ..., R_{Z_L}(\varphi_{N_g})\}$ in the circuit $\mc{C}$, each attached with an independent error mitigation channel with the form in \autoref{eq:braketO}. The variance of the final unbiased estimator is
\begin{equation} \label{eq:gamma_tot}
\begin{aligned}
\gamma_{\mr{tot}}^2 &= \prod_{i=1}^{N_g} \gamma_i^2 = \prod_{i=1}^{N_g} \left( \frac{1}{1- 2 \tilde{P}_L(\varphi_i)} \right)^2 \\
&\approx e^{ 4 \alpha_{\mr{RUS}} \varphi_{\mr{tot}} p^2 }  =: e^{4 P_{\mr{tot}}},
\end{aligned}
\end{equation}
where we assume the RUS and coherent error cancellation cost $\alpha_{\mr{RUS}}$ is approximately independent of $\varphi$. The total error rate $P_{\mr{tot}}$ and total rotation angle are defined by
\begin{equation}
P_{\mr{tot}} :=\sum_{i=1}^{N_g} \tilde{P}_L(\varphi_i)= \alpha_{\mr{RUS}}\varphi_{\mr{tot}}\cdot p^2, \quad \varphi_{\mr{tot}} := \sum_{i=1}^{N_g} |\varphi_i|.
\end{equation}

From \autoref{eq:gamma_tot} we can see that the sampling cost increases exponentially with the total error rate $P_{\mr{tot}}$. To keep the overall algorithm efficient, we need to limit $P_{\mr{tot}}$ to be a small constant,
\begin{equation}
P_{\mr{tot}} = \alpha_{\mr{RUS}}\varphi_{\mr{tot}}\cdot p^2 \lesssim 1.
\end{equation}
In this case, the error mitigation will introduce an extra sampling cost of $e^4\approx 55$ comparing to the ideal quantum circuit implementation. We then set a limitation to the overall rotation angle,
\begin{equation} \label{eq:varphi_tot_ineq}
\varphi_{\mr{tot}} \lesssim \frac{1}{\alpha_{\mr{RUS}}}\cdot \frac{1}{p^2}.
\end{equation}
Now, consider the current hardware condition with the physical error rate $p=1\cdot 10^{-3}$ and $\alpha_{\mr{RUS}}\leq 91$. We have 
\begin{equation} \label{eq:varphi_tot}
\varphi_{\mr{tot}} \leq 1.10\cdot 10^{4}\, \mr{rad}.
\end{equation}
In many early FT algorithms, especially the ones related to quantum simulations and quantum phase estimation (QPE), the rotation angles $\{\varphi_i\}$ are usually small. For example, in the QPE for the Hubbard model, to encure the energy accuracy of $\epsilon=0.01$, each rotation angle in the QPE circuit are set to be roughly $|\varphi|\approx 10^{-3}$ rad~\cite{toshio2024practical,kivlichan2020improvedfault}. Based on \autoref{eq:varphi_tot}, we can then implement approximately $1.10 \cdot 10^7$ such rotation gates.

We can perform a more detailed numerical analysis in the Trotter-based Hamiltonian simulation. For a Hamiltonian $H$ with the following Pauli decomposition,
\begin{equation}
H = \sum_{l=1}^L \alpha_l P_l,
\end{equation}
we define the $1$-norm of its coefficients as $\lambda:= \sum_{l=1}^L |\alpha_l|$. Consider a second-order Trotter formula,
\begin{equation}
\begin{aligned}
e^{-i T H} &\approx \left( \prod_{l=1}^L e^{-i\frac{x}{2} a_l P_l } \prod_{l=l}^1 e^{-i\frac{x}{2} a_l P_l } \right)^r \\
&= \left( \prod_{l=1}^L R_{P_l}(\varphi_l) \prod_{l=l}^1 R_{P_l}(\varphi_l) \right)^r ,
\end{aligned}
\end{equation}
where $r$ is the number of Trotter segments, $x=T/r$ is the unit evolution time for each segment, $\varphi_l:=-\alpha_l x/2$ is the rotation angle for the Pauli operator $P_l$ in each segment.

Based on \autoref{eq:varphi_tot_ineq}, we have that when the evolution time $T$ satisfies
\begin{equation} \label{eq:Tmax}
T \lesssim \frac{1}{\alpha_{RUS} \lambda p^2}.
\end{equation}
Consider an illustrative example of $N$-site one-dimensional Heisenberg model with a disordered magnetic field~\cite{childs2021theory},
\begin{equation} \label{eq:Heisenberg}
H = \sum_j (X_j X_{j+1} + Y_j Y_{j+1} + Z_j Z_{j+1}) + \sum_j h_j Z_j,
\end{equation}
with a periodic boundary condition and a random magnetic field $h_j\in [-h,h]$. This model is widely studied for its dynamical behaviours~\cite{Nandkishore2015manybody}. The averaged $1$-norm of the Hamiltonian is $\lambda=(3+h/2) N$, where $N$ is the site number.
Assuming $\alpha_{\mr{RUS}}\simeq 91$ and $h=1$, \autoref{eq:Tmax} provides the allowed simulation time,
\begin{equation} \label{eq:T_Heisenberg}
T\lesssim \frac{1}{318.5}\frac{1}{N p^2}.
\end{equation}
If we consider the current noise parameter $p=1\cdot 10^{-3}$ and the site number $N=100$, this allow us to perform a simulation up to $T_{\mr{max}} \simeq 31.4$, which is much longer than what can be demonstrated on the NISQ devices~ and is intractable on a classical computer without approximations~\cite{kimEvidenceUtilityQuantum2023}. We remark that, the error parameter of $p=1\cdot 10^{-3}$ is quite conservative based on the current experiments. If one consider a lower physical error rate of $p=1\cdot 10^{-4}$ like the ones in Ref.~\cite{akahoshi2024partially,toshio2024practical}, the allowed simulation time can be extended quadratically to $100 \cdot T_{\mr{max}}$.

Finally, we compare the overall spacetime resource cost of the non-Clifford gates in the Heisenberg-type Hamiltonian simulation in \autoref{eq:Heisenberg} based on the magic-state distillation method~\cite{Litinski2019magic}, magic-state-cultivation method~\cite{gidney2024cultivation} and our projection scheme in \autoref{fig:ResourceComp}. For a fair comparison, we consider to implement the $4$th-order Trotter formula with an accuracy requirement of $1\cdot 10^{-3}$, where a tight commutator bound is given in Childs~et~al.~\cite{childs2021theory}. We set the number of spins to be $N=50$ and the maximum evolution time $T=50$. The details of the resource estimation methods of different approaches can be found in Appendix~\ref{sec:AppCostEst}.

\begin{figure}[htbp]
    \centering
    \includegraphics[width=0.48\textwidth]{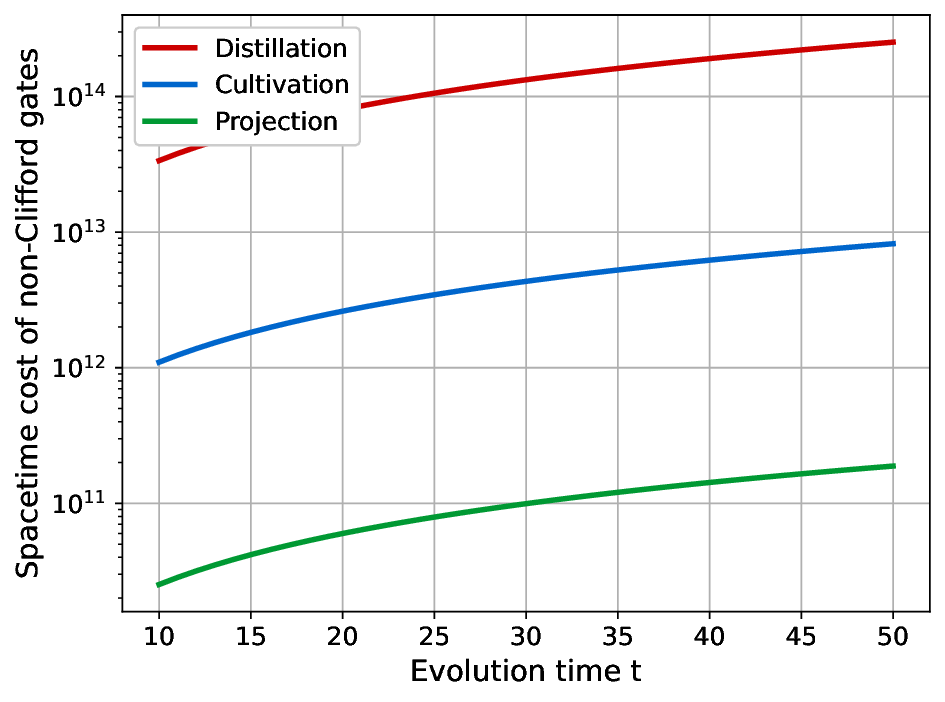}
    \caption{Comparison of the physical spacetime cost (i.e., the number of physical qubitcycles) of the magic state distillation~\cite{Litinski2019magic}, magic-state-cultivation method~\cite{gidney2024cultivation} and our projection scheme, in a $50$-spin Heisenberg Hamiltonian simulation task.}
    \label{fig:ResourceComp}
\end{figure}

From \autoref{fig:ResourceComp} we can see that the spacetime cost, characterized by the physical qubit-cycles (i.e., volume) of our method, is $1337.5$ times and $43.6$ times smaller than the state-of-the-art $15\text{-to-}1$ distillation protocol~\cite{Litinski2019magic} and the recently proposed magic-state-cultivation protocol~\cite{gidney2021stim}, respectively. The major reason is that our method not only enjoys a small footprint and low spacetime cost similar to surface code lattice surgery (comparable to cultivation approaches), but also eliminates the cumbersome $T$-gate compilation procedure.

\section{Summary and outlook}

In this work, we have shown how to generalize the framework of FT gadget analysis to leverage the specific error structures during the gate implementation. 
We provide a rigorous framework to define error-structure-tailored fault tolerance in the canonical stabilizer code circuits.
As concrete example, we have shown how to implement the 1-FT arbitrary rotation gates in [[4,1,1,2]]-code and surface codes without digital compiling and magic state distillation, which is formerly impossible due to the Eastin-Knill's theorem.
Furthermore, we combine our 1-FT error-structure-tailored gate design with the recently proposed small-angle projection scheme~\cite{choi2023fault,toshio2024practical} and show that we can implement the small angle rotation gate efficiently with an error of $\tilde{P}_L = \alpha_{\mr{RUS}}\cdot |\varphi| p^2$ where $\alpha_{\mr{RUS}}$ is almost a constant. By carefully taking the cost of probabilistic coherent error cancellation and repeat-until-success injection implementation into account, we have shown that the error overhead $\alpha_{\mr{RUS}} \lesssim 91$, which allow us to perform about $1.10\cdot 10^7$ small rotation gates $R_{Z_L}(\varphi)$ with $\varphi\approx 10^{-3}$ when the error mitigation techniques are introduced, enabling many promising early FT algorithms.

Recently, there are advances showning that $T$-ancillary state preparation can be performed in a high accuracy by the cultivation techniques without the usage of magic state distillation~\cite{gidney2024cultivation}. It is interesting to explore the possibility of similar techniques in our rotation state preparation $\ket{r_\varphi}_L$ procedure to further reduce the cost by the $T$ gate compiling. Meanwhile, we can also explore the usage of our error-structure-tailored FT design to further improve the $T$-state cultivation performance.

In our FT gate design, we leverage the error structure of the 2-level bare qubits or 3-level qutrits. To further enhance the performance, it will be beneficial to consider the hardware-level encoding to suppress physical errors to higher order~\cite{gottesman2001encoding,mirrahimi2014dynamically,brockQuantumErrorCorrection2025}, introduce error bias~\cite{aliferis2007accuracy,regladeQuantumControlCat2024,puttermanHardwareefficientQuantumError2025} or detect leakage errors~\cite{teohDualrailEncodingSuperconducting2023b,wuErasureConversionFaulttolerant2022}. We anticipate this will further suppress the error rate and boost the successful probability as well.

\begin{acknowledgments}
We acknowledge support from the ARO(W911NF-23-1-0077), ARO MURI (W911NF-21-1-0325), AFOSR MURI (FA9550-21-1-0209, FA9550-23-1-0338), DARPA (HR0011-24-9-0359, HR0011-24-9-0361), NSF (ERC-1941583, OMA-2137642, OSI-2326767, CCF-2312755, OSI-2426975), and the Packard Foundation (2020-71479). 
\end{acknowledgments}

\bibliographystyle{apsrev}

\onecolumngrid
\newpage

\clearpage

\begin{appendix}

\section{$[[4,1,1,2]]$ code: state preparation, quantum error detection and expansion} \label{sec:4112prepQEDexp}

A small-overhead $1$-FT logical $\ket{+}$ state preparation circuit is shown in \autoref{fig:4112QED}(a) by creating two pairs of Bell states without performing any measurements. Ideally, this process prepares the state $\ket{+}_L\ket{0}_g$, meaning the gauge qubit is initialized in the $\ket{0}$ state. Under the circuit-level depolarizing error model, the state preparation is $1$-FT.

The ``ZX''-type QED on the $[[4,1,1,2]]$ code is shown in \autoref{fig:4112QED}(b).
In the first half of the scheme, we measure the $Z$-type gauge operators $\psm{Z, I \\ Z, I}$ and $\psm{I, Z \\ I, Z}$ to determine the stabilizer value of $S_Z$. Additionally, if the gauge operator $g_Z$ is fixed (which is the case after our $\ket{+}_L$-state preparation), we can use the gauge operators themselves for error detection. In the second half of the procedure, we measure the $X$-type gauge operators $\psm{X, X \\ I, I}$ and $\psm{I, I \\ X, X}$ to determine the value of $S_X$. During this phase, the gauge qubit is projected onto the $X$-basis, and the value of $g_X$ becomes fully random.

\begin{figure}[htbp]
    \centering
    \includegraphics[width=0.45\textwidth]{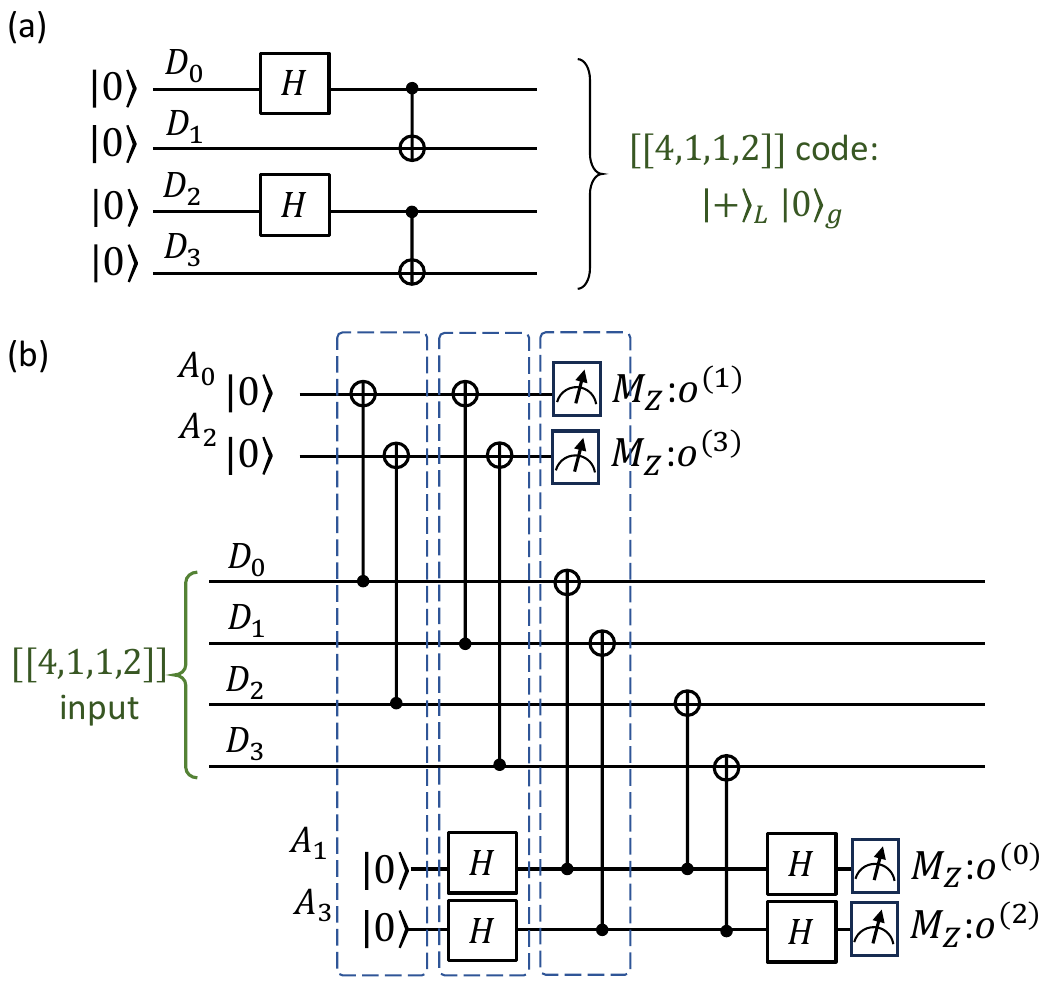}
    \caption{(a) Logical $\ket{+}_L$ state preparation on $[[4,1,1,2]]$ code. (b) ``ZX''-type QED on the $[[4,1,1,2]]$ code. The gauge qubit is fixed on $Z_g$ from the input. We first measure two weight-2 $Z$-type gauge operators, then measure two weight-2 $X$-type gauge operators. The latter step will project the gauge qubit onto the $X$-basis with a random value of $g_X$. The operations in the blue dashed box will be performed at the same time.}
    \label{fig:4112QED}
\end{figure}

Now, we discuss how to expand the codeword of the $[[4,1,1,2]]$ code to the surface code in a 1-fault-tolerant manner. As shown in \autoref{fig:4112expand}(a), suppose we begin with the $[[4,1,1,2]]$ code with the gauge qubit fixed to $g_X$ and aim to expand it to a $d=5$ rotated surface code. First, we embed the $[[4,1,1,2]]$ data qubit onto a corner of the 5-by-5 grid and initialize the remaining data qubits in either the $\ket{0}$ or $\ket{+}$ states. Next, we perform surface code syndrome measurements, as depicted in \autoref{fig:4112expand}(b), for two rounds. We discard the expansion if either of the following occurs: (1) some fixed syndrome values in the first round of syndrome measurements (marked with a mesh texture in \autoref{fig:4112expand}(b)) differ from their previous values, or (2) some syndrome values differ between the two rounds of surface code syndrome measurements.

\begin{figure}[htbp]
    \centering
    \includegraphics[width=0.5\textwidth]{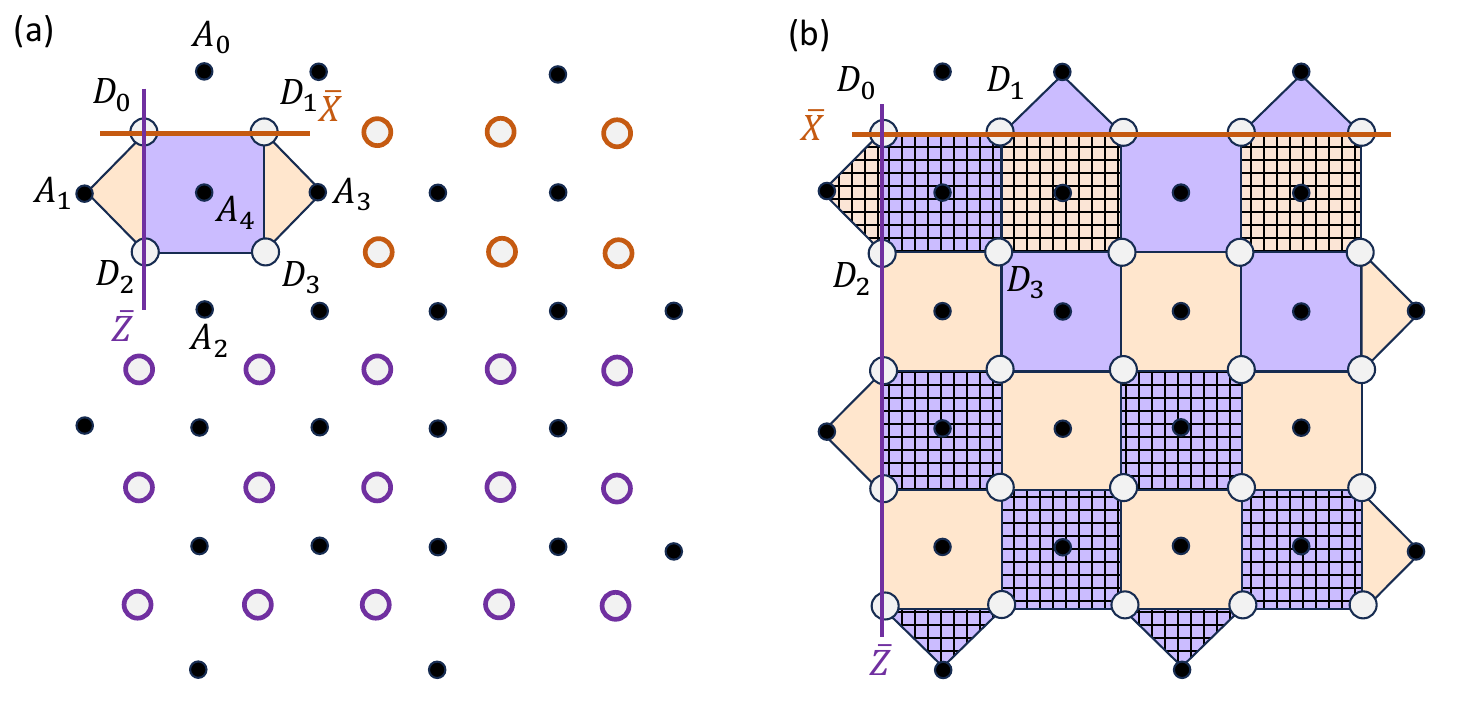}
    \caption{ $1$-FT expansion from the $[[4,1,1,2]]$ code to a $d=5$ rotated surface code: (a) Place the $[[4,1,1,2]]$ code, with the gauge fixed to $g_Z$, in the corner of a 5-by-5 grid. Initialize the orange qubit in the $\ket{+}$ state and the blue qubit in the $\ket{0}$ state. (b) Perform the surface code syndrome measurement three times. Syndromes with a mesh texture indicate values that are fixed before the syndrome measurement.}
    \label{fig:4112expand}
\end{figure}

\section{Fault-tolerant analysis with QED codes} \label{sec:AppFTQED}

In this section, we discuss how to generalize the FT analysis for QEC code in Ref.~\cite{aliferis2005quantum,gottesman2009intro} to the case of QED code. 
Here for simplicity, we mainly focus on the canonical FT model where local stochastic noise may occur after the implementation of each physical operation. The discussion of error-structure-tailored FT can be performed in a similar manner.
Without loss of generality, we focus on the FT discussion with the $[[n,k,d]]$ QED code with $k=1$. Similar to \autoref{eq:Pir}, we can define the $r$-filter with $r\leq d-1$ to be an extended code space projector with at most $r$ errors,
\begin{equation} \label{eq:Pir_QED}
\Pi_{(r)} = \sum_{P\in \mbb{P}_{r,\mc{S}} } P\, \Pi_0 P.
\end{equation}
Here, $\mbb{P}_{r,\mc{S}} = \mbb{P}_r/\sim$ is a quotient set of $\mbb{P}_r$ by the equivalence relationship of $P_1 \sim P_2: P_1 P_2 \in \mc{C}(\mc{S})$. 

The ideal decoder for the QED code is defined to be a noiseless codespace projection measurment with the projector $\{\Pi_0, I-\Pi_0\}$. When we fail to project the state to the code space, we discard this round.
We remark that, here we do not follow the approach in Ref.~\cite{aliferis2007accuracy} where the authors define the ideal QED decoder to be an ideal syndrome measurement appended by an arbitrary error recovery to the code space. Alternatively, we define the ideal decoder as a noiseless projection onto the code space, which is trace non-preserving. 
In what follows, we show that when focusing on the post-selected normalized quantum state which passes the post-selections, we can bound the error rate of the noisy outcome state to $O(p^d)$ where $p$ is the physical error rate.

We first define the requirements of FT-QED gadgets, including logicla state preparation, measurement, gate and QED, as follows.

\begin{definition}[$t$-FT logical measurement for QED codes] \label{def:FTmeas_QED}
The logical measurement is $t$-FT if it satisfy the following property:

If (i) the input state passes the $r$-filter, (ii) $s$ faults occur during the logical measurement, and (iii) $s + r \leq t$, then the normalized statistics of the measurement outcome are the same as first applying an ideal decoder, then performing a noiseless logical measurement.
\end{definition}

\begin{definition}[$t$-FT logical state preparation for QED codes] \label{def:FTprep_QED}
The logical state preparation is $t$-FT if it satisfies,
\begin{enumerate}
\item If $s$ faults occur during the state preparation, then the output state can pass the $s$-filter if $s\leq t$.
\item If $s$ faults occur during the gate implementation and $s\leq t$, then the conditional output state after applying an ideal decoder is the ideal logical state.
\end{enumerate}
\end{definition}

\begin{definition}[$t$-FT gate for QED codes] \label{def:FTgate_QED}
Consider a logical gate acting on $m$ code blocks. The gate is $t$-FT if it satisfies,
\begin{enumerate}
\item If (i) the input state on the $i$-th block passes the $r_i$-filter for $i=1,\dots,m$, (ii) $s$ faults occur during the gate implementation, and (iii) $s+\sum_{i=1}^m r_i\leq t$, then the output state can pass the $s+\sum_{i=1}^m r_i$ filter.
\item If (i) the input state on the $i$-th block passes the $r_i$-filter for $i=1,\dots,m$, (ii) $s$ faults occur during the gate implementation, and (iii) $s+\sum_{i=1}^m r_i\leq t$, then the conditional output state after applying an ideal decoder is the same as the conditional state where we first apply an ideal decoder on each code block after the $r_i$-filter and then apply an ideal gate.
\end{enumerate}
\end{definition}

\begin{definition}[$t$-FT quantum error detection for QED codes] \label{def:FTQED_QED}
The quantum error detection is $t$-FT if it satisfies,
\begin{enumerate}
\item If $s$ faults occur during the QED and $s\leq t$, then the output state can pass the $s$-filter.
\item If (i) the input state passes the $r$-filter, (ii) $s$ faults occur during the QED, and (iii) $s+\sum_{i=1}^m r_i\leq t$, then the conditional output state after applying an ideal decoder is the same as the conditional state after we directly apply an ideal decoder on the input state.
\end{enumerate}
\end{definition}

We remark that, due to the probabilistic nature of QED, we mainly focus on tracking the resulting state or measurement statistics after the QED postselection in the FT analysis. The goal of our FT analysis is to show that the noisy FT circuit can be used to simulate the ideal circuit with a small error rate of order $O(p^d)$ where $p$ is the error rate. We following the approach in \cite{gottesman2009intro}:

We consider a local stochastic error model: we gradually replace the noisy circuit elements to ideal noiseless elements.

\begin{proposition} Suppose we use $t$-FT gadgets defined for the QED codes, any logical circuit under the local stochastic error model that passes the postselection has a logical error rate of $O(p^{t+1})$, where $p$ is the physical error rate.
\end{proposition}

\begin{proof}

Without loss of generality, we consider an ideal logical circuit with the form shown in \autoref{fig:FT_idealcirc}(e), consisting of one state preparation gadget, two gate gadgets, and one logical measurement gadget.
In practice, we implement this circuit using noisy $t$-FT gadgets as shown in \autoref{fig:FT_idealcirc}(a).
Our goal is to demonstrate that when the circuit passes post-selection, the noisy implementation in \autoref{fig:FT_idealcirc}(a) can be reduced to the ideal circuit in \autoref{fig:FT_idealcirc}(e) with a logical error rate of $O(p^{t+1})$.
To establish this, we partition the circuit using dashed boxes (called extended rectangles or exRecs) as shown in \autoref{fig:FT_idealcirc}(a). These exRecs are defined as:
\begin{enumerate}
\item For quantum gates: $\text{exRec} = \text{ED} \circ \text{Ga} \circ \text{ED}$ (front ED gadget, gate gadget, and rear ED gadget)
\item For state preparation: $\text{exRec} = \text{SP} \circ \text{ED}$ (state preparation gadget and rear ED gadget)
\item For logical measurement: $\text{exRec} = \text{ED} \circ \text{LM}$ (front ED gadget and logical measurement gadget)
\end{enumerate}

\begin{figure}[htbp]
    \centering
    \includegraphics[width=0.4\textwidth]{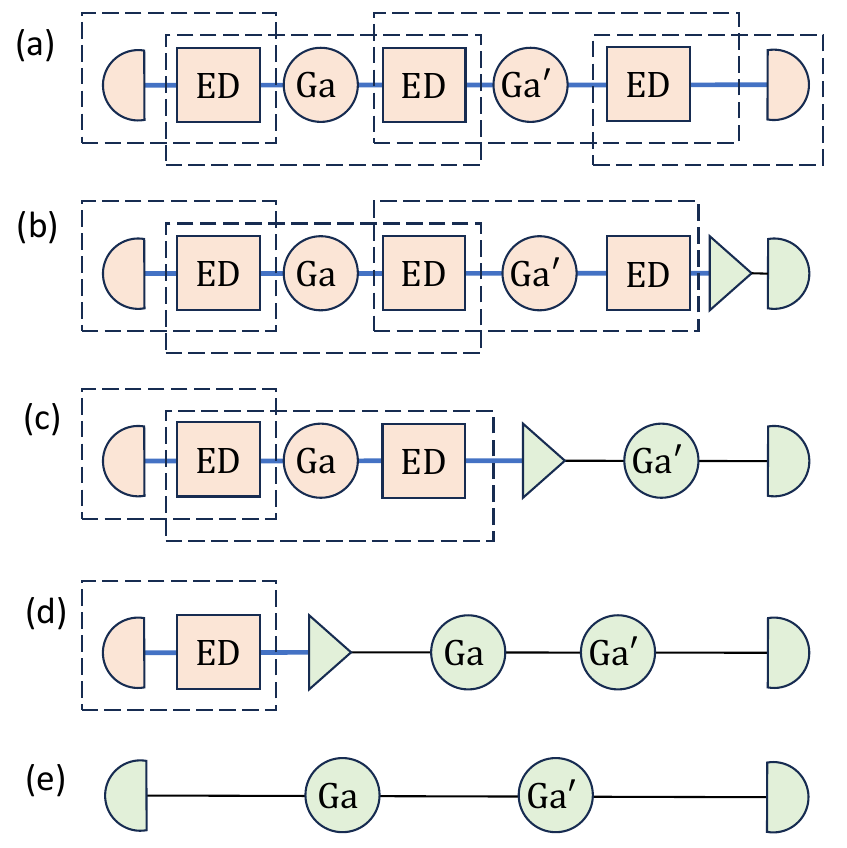}
    \caption{Reduction of a logical FT circuit to the ideal circuit. The left-half-circles, the right-half-circles, the triangles, the ``EC'' squares and the ``Ga'' circles indicate the state preparation gadgets, the measurement gadgets, the ideal decoder, the error-correction gadgets, and the logical gates, respectively. The yellow and green gadgets indicate noisy and ideal (noiseless) gadgets, respectively.
    }
    \label{fig:FT_idealcirc}
\end{figure}

We define an exRec as $t$-good if it contains at most $t$ faults. We now demonstrate that when all dashed boxes (exRecs) in \autoref{fig:FT_idealcirc} are $t$-good, the noisy circuit can be reduced to the ideal circuit following the sequence shown in \autoref{fig:FT_idealcirc}.
Without loss of generality, we examine a single gate reduction step from \autoref{fig:FT_idealcirc}(c) to (d). Specifically, we prove that for any input state, the post-selected output state after processing through the exRec and ideal decoder in \autoref{fig:FTproof_exRec}(a) equals the output state in \autoref{fig:FTproof_exRec}(e) when the extended rectangle contains $\leq t$ faults.

\begin{figure}[htbp]
    \centering
    \includegraphics[width=0.4\textwidth]{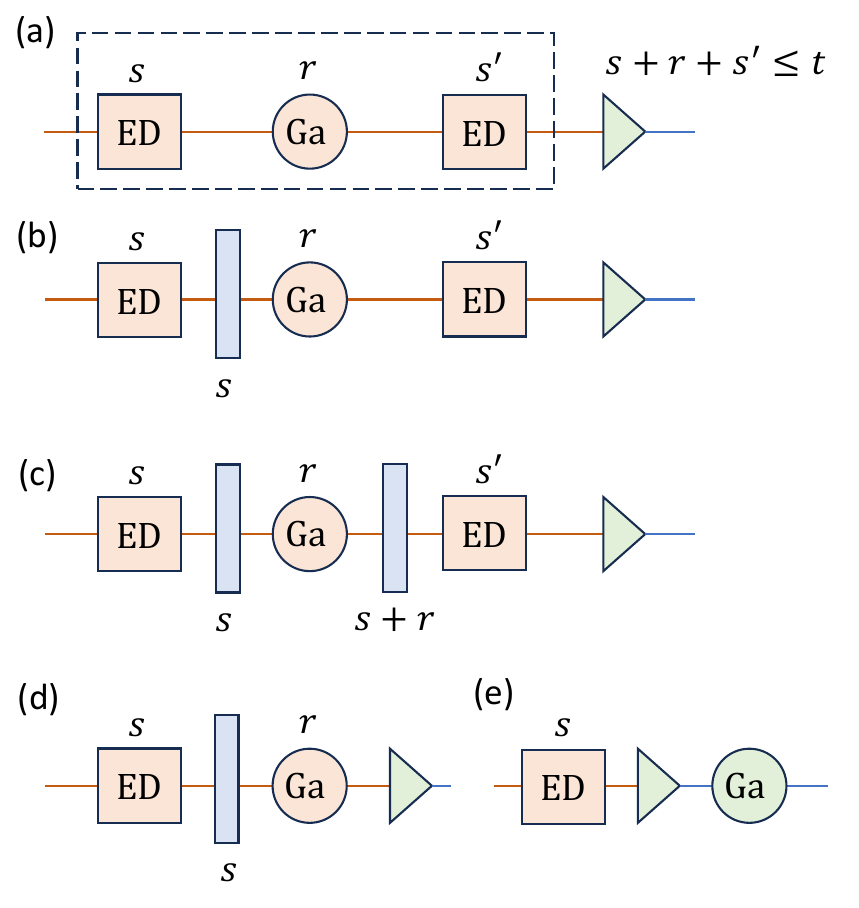}
    \caption{Reduction of the extended rectangular to an ideal gadget when the total number of faults is smaller than $t$. Recall that we focus on the reduction of the state which pass the postselection.
    }
    \label{fig:FTproof_exRec}
\end{figure}

Assume there are $s$, $r$, and $s'$ faults in the front ED gadget, gate gadget, and rear ED gadget respectively within the exRec shown in \autoref{fig:FTproof_exRec}(a).
\begin{enumerate}
\item By the first property of the ED gadget (Def.~\ref{def:FTQED_QED}), \autoref{fig:FTproof_exRec}(a) is equivalent to \autoref{fig:FTproof_exRec}(b).
\item By the first property of the gate gadget (Def.~\ref{def:FTgate_QED}), \autoref{fig:FTproof_exRec}(b) is equivalent to \autoref{fig:FTproof_exRec}(c).
\item Applying the second property of the ED gadget (Def.~\ref{def:FTQED_QED}), we reduce \autoref{fig:FTproof_exRec}(c) to \autoref{fig:FTproof_exRec}(d).
\item Applying the second property of the gate gadget (Def.~\ref{def:FTgate_QED}), we reduce \autoref{fig:FTproof_exRec}(d) to \autoref{fig:FTproof_exRec}(e).
\end{enumerate}
Note that in the last two reduction steps, the circuits may have different postselection probabilities. Here, "reduction" refers to equivalence of the postselected quantum states after quantum error detection.

Assuming all the faults occur independently, the probability of an $\mr{exRec}$ to have more than $t$ faults is $O(p^{t+1})$. Therefore, the reduction step from \autoref{fig:FT_idealcirc}(c) to (d) succeed with an error of $O(p^{t+1})$. We can apply similar techniques to complete the full reduction from \autoref{fig:FT_idealcirc}(a) to (e).
\end{proof}

We remark that, in the above discussion, we focus on the fault-tolerant analysis of the QED codes with canonical fault-tolerance definition based on the local stochasical error model. 
We can also generalize the discussion to the case when some of the gate gadgets are defined under the error-structure-tailored FT in Def.~\ref{def:FTgate_error_tailored}, where the faults are described by the Markovian noisy jump during the continuous Hamiltonian control.

Suppose we have an $\mr{exRec}$ for a quantum gate where the front and rear ED gadgets are defined based on the canonical FT model, while the gate gadget is an error-structure-tailored $t$-FT gate following the definition in Def.~\ref{def:FTgate_error_tailored}.
One can show that the reduction from \autoref{fig:FTproof_exRec}(a) to (e) still holds when there are $s$, $r$, and $s'$ faults in the front ED gadget, gate gadget, and rear ED gadget respectively within the exRec shown in \autoref{fig:FTproof_exRec}(a) and $s+r+s'\leq t$. Here, $r$ faults for the error-structure-tailored FT gate gadget refers to a truncated error dynamics up to the $r$-th order. Since the error of the $r$-th-order truncated dyson series is $O(p^{r+1})$, the reduction from \autoref{fig:FTproof_exRec}(a) to (e) is with the error of $O(p^{t+1})$.

\section{Fault-tolerance of the logical rotation gates} \label{sec:AppZZrot}

\subsection{Proof of \autoref{prop:4112FTdispersive} } \label{ssc:AppPropZZdisp}

\begin{proof}
For the simplicity of discussion, we fix the gauge qubit to $\ket{0}$. The $[[4,1,1,2]]$ subsystem code can then be regarded as a $[[4,1,2]]$ stabilizer code with the stabilizer of $S_X, S_Z$ in \autoref{eq:4112SZSX} and $g_Z$ in \autoref{eq:4112gZgX}. The code space projector is,
\begin{equation}
\Pi_0 = \frac{1}{8}(1+S_X)(1+S_Z)(1+g_Z).
\end{equation}
One can also fix the gauge qubit to $\ket{+}$ where the above discussion still holds.

Following the Dyson series expansion in \autoref{eq:DysonExpansion}, the dynamics during the gate implementation can be expanded by
\begin{equation}
\mc{G}(t,0) = \sum_{q=0}^\infty \mc{G}_q(t,0).
\end{equation}
In the FT analysis, we truncate the dynamics to $s$-th order to check its leading-order dynamics. We will denote the initial state $\rho_0 = \rho(0)$ and the $q$-th-order propagator $\mc{G}_q = \mc{G}_q(t,0)$.

First, we consider the scenario with $r=1, s=0$, i.e, the input state pass the $1$-filter, and we consider the noiseless propagation. We need to show that:
\begin{enumerate}
\item If $\Pi_{(1)} \rho_0 \Pi_{(1)} = \rho_0$, then $\Pi_{(1)} \mc{G}_0(\rho_0) \Pi_{(1)} = \mc{G}_0(\rho_0)$.
\item If $\Pi_{(1)} \rho_0 \Pi_{(1)} = \rho_0$, then $\Pi_0\, \mc{G}_0(\rho_0)\, \Pi_0 = \mc{R}_{Z_L}(\varphi) (\Pi_0 \rho_0 \Pi_0)$. 
\end{enumerate}

Recall that $\mc{G}_0 = \mc{R}_{Z_L}(\varphi)$. We have
\begin{equation} \label{eq:Pi1G0Pi1}
\begin{aligned}
&\Pi_{(1)} \mc{G}_0(\rho_0) \Pi_{(1)} \\ 
=& \Pi_{(1)}  e^{i\varphi Z_L} \rho_0 e^{-i\varphi Z_L} \Pi_{(1)} \\
=& \Pi_{(1)}  e^{i\varphi Z_L} \Pi_{(1)} \rho_0 \Pi_{(1)} e^{-i\varphi Z_L} \Pi_{(1)} \\
=& e^{i\varphi Z_L} \Pi_{(1)} \rho_0 \Pi_{(1)} e^{-i\varphi Z_L} \\
=& e^{i\varphi Z_L} \rho_0 e^{-i\varphi Z_L} = \mc{G}_0(\rho_0).
\end{aligned}
\end{equation}
In the second line, we use the assumption that $\Pi_{(1)} \rho_0 \Pi_{(1)} = \rho_0$. In the third line of \autoref{eq:Pi1G0Pi1}, we use the following property,
\begin{equation} \label{eq:P1eiZLPi1}
\begin{aligned}
&\Pi_{(1)}  e^{i\varphi Z_L} \Pi_{(1)}  \\
= & \sum_{P_1, P_2\in \mbb{P}_{1,S}} P_1 \Pi_0 P_1 e^{i\varphi Z_L} P_2 \Pi_0 P_2 \\
= & \sum_{P_1, P_2\in \mbb{P}_{1,S}, P_1P_2\in \mc{C}(\mc{S})} P_1 P_1 e^{i\varphi Z_L} P_2 \Pi_0 P_2 \\
= & \sum_{P_2\in \mbb{P}_{1,S}} e^{i\varphi Z_L} P_2 \Pi_0 P_2 = e^{i\varphi Z_L} \Pi_{(1)}.
\end{aligned}
\end{equation}
In the third line of \autoref{eq:P1eiZLPi1}, we use the property that $P_1 e^{i\varphi Z_L} P_2\in \mc{C}(\mc{S})$, hence $P_1 P_2 = \mc{C}(\mc{S})$. As a result, $\Pi_0 P_1 e^{i\varphi Z_L} P_2 \Pi_0 = P_1 e^{i\varphi Z_L} P_2 \Pi_0$.

On the other hand, it is easy to check that $\Pi_0 \mc{G}_0(\rho_0) \Pi_0 = \mc{R}_{Z_L}(\varphi) (\Pi_0 \rho_0 \Pi_0)$ holds. Therefore, the gate is fault-tolerant when $r=1, s=0$.

Second, we consider the scenario with $r=0, s=1$, i.e., the input state is in the code space of $\Pi_0$, and we consider the first-order errors. We need to show that:
\begin{enumerate}
\item If $\Pi_{0} \rho_0 \Pi_{0} = \rho_0$, then $\Pi_{(1)} \mc{G}^{[1]}(\rho_0) \Pi_{(1)} = \mc{G}^{[1]}(\rho_0)$.
\item If $\Pi_{0} \rho_0 \Pi_{0} = \rho_0$, then $\Pi_0\, \mc{G}^{[1]}(\rho_0)\, \Pi_0 \propto \mc{R}_{Z_L}(\varphi) (\Pi_0 \rho_0 \Pi_0)$. 
\end{enumerate}

Recall that $\mc{G}^{[1]} = \mc{G}_0 + \mc{G}_1$. The noiseless dynamics $\mc{G}_0$ naturally satisfy the FT requirements. For the first-order noise dynamics,
\begin{equation}
\mc{G}_1 = \mc{G}_1(t,0) = \int_0^t dt_1 \mc{T}\left( \mc{U}(t,t_1)\, \mc{D}\, \mc{U}(t_1, 0) \right), 
\end{equation}
where $\mc{U}(t,0)(\rho) = e^{-it H_{ZZ}} \rho e^{it H_{ZZ}}$, $\mc{D}(\rho) = \sum_j \gamma_j \left( J_j \rho J_j^\dagger - \frac{1}{2}( J_j^\dagger J_j \rho + \rho J_j^\dagger J_j ) \right) $. 

Without loss of generality, we focus on the dynamics associated with the jump operator $J_{e\to g} = \ket{g}\bra{e}\otimes I$ on a specific jump time $t_1\in[0,t]$, 
\begin{equation}
\begin{aligned}
\mc{U}(t,t_1) \mc{D}_{e\to g} \mc{U}(t_1, 0) (\rho_0).
\end{aligned}
\end{equation}
The dissipator $\mc{D}_{e\to g}$ can be decomposed as
\begin{equation}
\begin{aligned}
\mc{D}_{e\to g}(\bullet) &= \gamma_{e\to g} \left( J_{e\to g} \bullet J_{e\to g}^\dagger  - \frac{1}{2} \{ J_{e\to g}^\dagger J_{e\to g}, \bullet \}  \right) \\
&=: \gamma_{e\to g} ( \mc{J}_{e\to g} + \mc{B}_{e\to g} ),  
\end{aligned}
\end{equation}
where $\mc{J}_{e\to g}(\bullet) = J_{e\to g} \bullet J_{e\to g}^\dagger$ represents the jump dynamics, while $\mc{B}_{e\to g} = -\frac{1}{2} \{ J_{e\to g}^\dagger J_{e\to g}, \bullet \}$ is the effect of backaction. The backaction can be treated as a first-order modification onto the no-jump dynamics $\mc{G}_0$ to keep the overall dynamics trace-preserving. 

We first check the effect of the jump operator $\mc{J}_{e\to g}$. To this end, we notice that
\begin{equation} \label{eq:Jeg_commute}
\begin{aligned}
& U(t,t_1) \ket{g}\bra{e} U(t,t_1)^\dagger \\
= &\ket{g}\bra{e} \otimes \left(  \cos(2\chi(t-t_1)) I + i \sin(2\chi(t-t_1)) Z \right),
\end{aligned}
\end{equation}
recall that $\ket{g}\bra{e} = \frac{1}{2}(X+iY)$. Since the Pauli operators $X_0 I_2, X_0 Z_2, Y_0 I_2$ and $Y_0 Z_2$ all belong to $\mbb{P}_{1,\mc{S}}$ of $[[4,1,2]]$ code, we have
\begin{equation} \label{eq:Jump_FT}
\begin{aligned}
\Pi_1 \mc{U}(t,t_1) \mc{J}_{e\to g} \mc{U}(t_1,0) (\rho_0) \Pi_1 &= \mc{U}(t,t_1) \mc{J}_{e\to g} \mc{U}(t_1,0) (\rho_0), \\
\Pi_0 \mc{U}(t,t_1) \mc{J}_{e\to g} \mc{U}(t_1,0) (\rho_0) \Pi_0 &= 0,
\end{aligned}
\end{equation}
when the input state $\rho_0$ is in the code space. Therefore, the jump dynamics satisfy the FT condition.

Now, we check the effect of the backaction $\mc{B}_{e\to g}$. Note that $J_{e\to g}^\dagger J_{e\to g} = \ket{e}\bra{e}\otimes I$, which commutes with the noiseless dynamics $\mc{U}(t,t_1)$. As a result, the backation dynamics can be written as
\begin{equation} \label{eq:back_commute}
\begin{aligned}
&\mc{U}(t,t_1) \mc{B}_{e\to g} \mc{U}(t_1, 0) (\rho_0) \\
= &-\frac{1}{2} \Big( (\ket{e}\bra{e}\otimes I) R_{Z_L}(\phi) \rho_0 R_{Z_L}(\phi) \\
& \quad + R_{Z_L}(\phi) \rho_0 R_{Z_L}(\phi) (\ket{e}\bra{e}\otimes I) \Big).
\end{aligned}
\end{equation}
From \autoref{eq:back_commute}, we can check that,
\begin{equation} \label{eq:Back_FT}
\begin{aligned}
\Pi_1 \mc{U}(t,t_1) \mc{B}_{e\to g} \mc{U}(t_1,0) (\rho_0) \Pi_1 &= \mc{U}(t,t_1) \mc{B}_{e\to g} \mc{U}(t_1,0) (\rho_0), \\
\Pi_0 \mc{U}(t,t_1) \mc{B}_{e\to g} \mc{U}(t_1,0) (\rho_0) \Pi_0 &\propto \mc{R}_{Z_L}(\phi)(\rho_0),
\end{aligned}
\end{equation}
when the input state is in the code space. Combine \autoref{eq:Jump_FT} and \autoref{eq:Back_FT}, we have shown that the gate is FT when $r=0, s=1$.

\end{proof}

\subsection{Proof of \autoref{prop:4112FTancilla} } \label{ssc:AppPropZZancilla}

\begin{proof}
We prove the fault-tolerance of the ancillary-based $ZZ$-rotation gate in a similar manner to \autoref{prop:4112FTdispersive} by the Dyson series expansion of the whole dynamics in \autoref{eq:DysonExpansion}.

The case when $r=1, s=0$, i.e., the input state is noisy while the gate is noiseless is similar to the one in \autoref{prop:4112FTdispersive}.

When $r=0, s=1$, i.e., the input state is noiseless and the first-order error occurs, we can analyze the effect of the dissipator $D[\sqrt{\gamma_j} J_j]$ related to each jump operator $J_j$ during each step of the gate implementation. The effect of the dissipator can be decomposed to the jump operator and the backaction,
\begin{equation}
\begin{aligned}
D[\sqrt{\gamma_j} J_j](\bullet) &= \gamma_j \left( J_j \bullet J_j^\dagger  - \frac{1}{2} \{ J_j^\dagger J_j, \bullet \}  \right) \\
&=: \gamma_j ( \mc{J}_j + \mc{B}_j ).
\end{aligned}
\end{equation}
The effect of the backaction $\mc{B}_j$ can be analyzed similarly to the one in \autoref{prop:4112FTdispersive}. The major task is to enumerate and analyze the effect of different jump operators $\{J_j\}$ occuring at different stages of the gates. We want to show that,
\begin{enumerate}
\item If $\Pi_{0} \rho_0 \Pi_{0} = \rho_0$, then the resulting state $\rho_1$ with a single jump $J_j$ satisfies $\Pi_{(1)} \rho_1 \Pi_{(1)} = \rho_1$.
\item If $\Pi_{0} \rho_0 \Pi_{0} = \rho_0$, then $\Pi_0 \rho_1 \Pi_0 = \mc{R}_{Z_L}(\varphi) (\Pi_0 \rho_0 \Pi_0)$. 
\end{enumerate}

Now, we enumerate the effect of different jump operators.
\begin{enumerate}
\item The dephasing error $J_{\phi} = \ket{e}\bra{e}$ on the data qubit $D_0$ and $D_2$ occured during the whole procedure will not propagate. The final state is $\mc{J}_{\phi}^{(i)}\mc{R}_{Z_L}(\varphi)(\rho_0)$. It is easy to prove that the final state satisfy the FT condition.

\item Consider the ancillary dephasing error $J_{\phi}^{(A)} = \ket{e}\bra{e} + \Delta_f \ket{f}\bra{f}$ during the first group of CZ gates in \autoref{fig:RZZgate_ancilla}. Here $\Delta_f$ is a constant charactering the dephasing property of the $f$-level. To analyze its effect, we propagate the noise to the end of the circuit before ancillary measurement. For simplicity, we focus on the dephasing effect of $f$-level and set $J_{\phi}^{(A)}= \ket{f}\bra{f}$. This is legal because the effect of $\ket{e}\bra{e}$ will not propagate. We can then write $\ket{f}\bra{f}:=\frac{1}{2}(I_{gf} - Z_{gf})$ and check the effect of $Z_{gf}$ instead.
\begin{equation}
\begin{aligned}
 & e^{ i \frac{\pi}{4} Y_{gf} } CZZ_{D_0, D_2; A} e^{i \varphi X_{gf}} Z_{gf} \\
=& e^{ i \frac{\pi}{4} Y_{gf} } CZZ_{D_0, D_2; A} (\cos(2\varphi) Z_{gf} + \sin(2\varphi) Y_{gf} ) e^{i \varphi X_{gf}} \\
=& e^{ i \frac{\pi}{4} Y_{gf} } (\cos(2\varphi) Z_{gf} + \sin(2\varphi) Z_0 Z_2 Y_{gf} ) CZZ_{D_0, D_2; A} e^{i \varphi X_{gf}} \\
=& (-\cos(2\varphi) X_{gf} + \sin(2\varphi) Z_0 Z_2 Y_{gf} ) e^{ i \frac{\pi}{4} Y_{gf} } CZZ_{D_0, D_2; A} e^{i \varphi X_{gf}} \\
\end{aligned}
\end{equation}
Since both $X_{gf}$ and $Y_{gf}$ error will flip the ancillary qubit to $f$-level, which will lead to a wrong detection result $o_{zz}=f$ and will be post-selected, the unnormalized final state $\rho_1\propto 0$. Therefore, it na\"ively satisfies the FT conditions.

\item The ancillary dephasing error during the idling gate will flip the ancillary qubit to $f$-level, leading to a wrong detection result and will be post-selected.

\item The ancillary dephasing error during the second group of CZ gates can be analyzed in a similar manner as the one during the first group of CZ gates. It will leads to a wrong detection result.

\item Consider an ancillary dephasing error occurs during the last $Y$-rotation gate in \autoref{fig:RZZgate_ancilla}. Suppose it occurs at the time $t_1$ during the whole gate period $t= \pi/(4\Omega)$ where $\Omega$ is the driving strength. We propagate it to the end of the $Y$-rotation gate,
\begin{equation} \label{eq:fdephasing_prop}
\begin{aligned}
 & e^{ -i \Omega(t-t_1) Y_{gf} } \ket{f}\bra{f} e^{i \Omega(t-t_1) Y_{gf} } \\
=& \frac{I}{2} + \left(\frac{1}{2} - \cos(2\Omega(t-t_1))\right)Z_{gf} - \sin(2\delta) X_{gf}.
\end{aligned}
\end{equation}
Recall that in the ideal case, the ancillary qubit should be on the level $g$. As a result, the first two terms have no effect on $\ket{g}$. The last Pauli $X_{gf}$ term will flip the ancilla to level $f$ and will be post-selected. The remained state will be proportional to the ideal state.

\item Consider an ancillary dephasing error occurs during the first $Y$-rotation gate \autoref{fig:RZZgate_ancilla}. Suppose it occurs at the time $t_1$ during the whole gate period $t= \pi/(4\Omega)$. Following \autoref{eq:fdephasing_prop}, we propagate it to the end of the first $Y$-rotation gate. Recall that the ideal state after the $Y$-rotation gate is $(\ket{g}+\ket{f})$, we can write the effect of the error on the ancillary state as
\begin{equation} \label{eq:fdephasing_prop2}
\begin{aligned}
& \left(\frac{I}{2} + \left(\frac{1}{2} - \cos(2\Omega(t-t_1))\right)Z_{gf} + \sin(2\delta) X_{gf} \right)  (\ket{g}+\ket{f}) \\
= & \left( (\frac{1}{2} +\sin(2\delta)) I + \left(\frac{1}{2} - \cos(2\Omega(t-t_1))\right)Z_{gf} \right) (\ket{g}+\ket{f}),
\end{aligned}
\end{equation}
where we use the property that $\ket{g}+\ket{f}$ is stabilized by $X_{gf}$. The Pauli $Z_{gf}$ term is a dephasing error whose effect have been discussed above. The remained state will be proportional to the ideal state.

\item The first-order relaxation error on the $g$-$f$ qubit $A$ $J_{f\to e}^{(A)}=\ket{e}\bra{f}$ will transfer its population to the $e$ level, which leads to a wrong detection result $o_{zz}=e$ and will be post-selected. Therefore, it na\"ively satisfies the FT conditions. On the other hand, the first-order relaxation error $J_{e\to g}^{(A)}$ will not occur on the ancillary state which only have population on $g$ and $f$ levels.

\item Consider the relaxation error on the data qubit $D_0$ ($D_2$) during the implementation of the second group of CZ gates in \autoref{fig:RZZgate_ancilla}. Suppose it occurs at the time $t_1$ during the whole CZ gate period $t = \pi/\chi$. We first write it as $\ket{g}_i\bra{e} = \frac{1}{2}(X + iY)$ and check the effect of Pauli $X$. The effect of Pauli $Y$ can be analyzed similarly.
First, we propagate the $X$ error after the remaining dispersive coupling dynamics of $H_{ef}^{(0)} + H_{ef}^{(2)}$,
\begin{equation}
\begin{aligned}
& e^{-i (t-t_1) H_{ef}^{(0)} } X_0 e^{i (t-t_1) H_{ef}^{(0)} } \\
=& X_0 + (\cos\delta -1) X_0\otimes \ket{f}\bra{f} - \sin\delta Y_0\otimes \ket{f}\bra{f}.
\end{aligned}
\end{equation}
Note that, the second and the third term in the second line indicates a dephasing error on the ancillary qubit. Following the ancilla dephasing analysis above, these two terms will filp the ancillary qubit and lead to a wrong ancillary detection result $o_{zz}=f$. As a result, only the $X_0$ term will survive after the ancillary measurement, which is a single-qubit error detectable by the $[[4,1,1,2]]$ code.

\item Consider the relaxation error $J_{e\to g}^{(i)} = \ket{g}_{i}\bra{e}$ on the data qubit $D_0$ ($D_2$) during the idling gate between the first and the second group of CZ gates in \autoref{fig:RZZgate_ancilla}. The relaxation error will introduce an ancillary dephasing error $\ket{f}\bra{f}$, which will be captured by the ancillary measurement.

\item The relaxation error $J_{e\to g}^{(i)} = \ket{g}_{i}\bra{e}$ on the data qubit $D_0$ ($D_2$) during the implementation of the first group of CZ gates can be analyzed in a similar manner as the relaxation error on the data qubit $D_0$ ($D_2$) during the implementation of the second group of CZ gates above. It will be captured by the ancillary measurement.
\end{enumerate}

\end{proof}

\subsection{Proof of \autoref{prop:ProjectionFT}} \label{ssc:AppProofProjectionFT}

\begin{proof}

Focusing on the multi-rotation protocol with $m=2$, i.e., $d=2k$, we first assume that Pauli errors occur during the rotation gate implementation. Importantly, all the weight-$1$ Pauli errors including $Z$-type error (e.g. $ZIII$ error on qubits $D_0$ to $D_3$ in \autoref{fig:ProjectionScheme}) are detectable by the surface code syndrome measurement. This is because the multi-rotation gate $\prod_{i=1}^k R_{ZZ,i}(\theta)$ in \autoref{eq:kR_ZZ_expand} will only transform the $\ket{+}_L$ state to the syndrome subspace that is at least $m=2$ away from the code space. As a result, only some specific weight-2 $ZZ$ error may lead to the undetectable error. For example, if the Pauli $Z_0 Z_1$ error occurs on the data qubit $D_0$ and $D_2$, then the state we prepare will change from \autoref{eq:psi_id_m2} to
\begin{equation} \label{eq:psi_ud1_m2}
\begin{aligned}
& \ket{\psi_{\mr{ud(2)}}} \\
=& \sqrt{P_{\mr{pass|ud(2)}}} \ket{\psi_{\mr{ud(2),pass}}} + \sqrt{P_{\mr{fail|ud(2)}}} \ket{\psi_{\mr{ud(2),fail}}},
\end{aligned}
\end{equation}
where
\begin{equation} \label{eq:Ppassud1_psiud1}
\begin{aligned}
P_{\mr{pass|ud(2)}} &:= \sin^2\theta \cos^2\theta (\sin^{2k-4}\theta + \cos^{2k-4}\theta), \\
\ket{\psi_{\mr{ud(2),pass}}} & \propto R_{Z_L}(\varphi_1) \ket{+}_L = \ket{r_{\varphi_1}}, \\
\varphi_1 = -\sin^{-1}&\left(\frac{1}{\sqrt{P_{\mr{pass|ud(2)}}}} \sin^{k-1}\theta\cos\theta\right) \approx -\theta^{k-2}.
\end{aligned}
\end{equation}
Here, the subscript $\mr{ud(2)}$ is used to denote an undetectable weight-$2$ error on one group of qubits (e.g., Pauli $ZZII$ or $IIZZ$ error on qubits $D_0$ to $D_3$ in \autoref{fig:ProjectionScheme}). Indeed, some weight-$4$ error on $\mr{Supp}(Z_L)$ may also lead to undetectable error which pass the post-selection (e.g., Pauli $ZZZZ$ error in \autoref{fig:ProjectionScheme}), but the probability will be ignorable with $O(p^4 \theta^4)$. The approximation of $\varphi_1 \approx -\theta^{k-2}$ is valid when $|\theta| \ll 1$.

We now analyze the error in the post-selected state $\rho_{\mr{pass}}$ after passing the syndrome measurement,
\begin{equation} \label{eq:rho_pass_m2}
\rho_{\mr{pass}} \approx P_{\mr{id|pass}}\, \rho_{\mr{id,pass}} + P_{\mr{ud(2)|pass}}\, \rho_{\mr{ud(2),pass}} + O(|\varphi|^2 p^4),
\end{equation}
where the approximation is due to higher-order undetectable error with the order of $O(|\varphi|^2 p^4)$ whose form will be analyzed later. $\rho_{\mr{id,pass}} = \ket{r_\varphi}\bra{r_\varphi}$ and $\rho_{\mr{ud(2),pass}} = \ket{r_{\varphi_1}}\bra{r_{\varphi_1}}$. $P_{\mr{ud(2)|pass}}$ is the Bayesian probability of undetectable error given by
\begin{equation} \label{eq:Pud1_pass_m2}
\begin{aligned}
&\quad P_{\mr{ud(2)|pass}} = \frac{ P_{\mr{ud(2),pass}} }{ P_{\mr{pass}} } \\
&= \frac{ P_{\mr{ud(2)}} P_{\mr{pass|ud(2)}} }{ P_{\mr{id}} P_{\mr{pass|id}} + P_{\mr{ud(2)}} P_{\mr{pass|ud(2)}} } + O(\theta^6 p^6) \\
&= \frac{ P_{\mr{ud(2)}} P_{\mr{pass|ud(2)}} }{P_{\mr{id}} P_{\mr{pass|id}}} + O(\theta^4 p^4). 
\end{aligned}
\end{equation}
Here in the first approximation, we omit the contribution from higher-order undetectable error in $P_{\mr{pass}}$. In the second approximation, we only keep the leading-order term in the denominator. We use the fact that in the small angle limit, $P_{\mr{id,pass}}=O(1), P_{\mr{ud(2),pass}}=O(\theta^2 p^2)$ and $P_{\mr{ud(4),pass}}=O(\theta^4 p^4)$.

We can then calculate the trace distance between $\rho_{\mr{pass}}$ and $\rho_{\mr{id,pass}}=\ket{r_\varphi}\bra{r_\varphi}$,
\begin{equation} \label{eq:Dtr_rho_pass_m2}
\begin{aligned}
&\quad D_{\mr{tr}}(\rho_{\mr{pass}}, \rho_{\mr{id,pass}}) = \frac{1}{2}\left| \rho_{\mr{pass}} - \rho_{\mr{id,pass}} \right| \\
&= \frac{1}{2}P_{\mr{ud(2)|pass}} \left|\, \ket{r_\varphi}\bra{r_\varphi} - \ket{r_{\varphi_1}}\bra{r_{\varphi_1}} \,\right| + O(\theta^4 p^4) \\
&= P_{\mr{ud(2)}} \frac{ P_{\mr{pass|ud(2)}} }{ P_{\mr{id}} } \sin\Delta_\varphi + O(\theta^4 p^4).
\end{aligned}
\end{equation}
In the approximation in the second and third line, we only keep the leading order terms. Here, $\Delta_\varphi:= |\varphi - \varphi_1|\approx \theta^k + \theta^{k-2}$. In the small angle limit of $|\theta|\to 0$, we have 
\begin{equation} \label{eq:Dtrpass_approx}
D_{\mr{tr}}(\rho_{\mr{pass}}, \rho_{\mr{id,pass}}) \approx P_{\mr{ud(2)}}\theta^2\cdot \theta^{k-2} \approx P_{\mr{ud(2)}} |\varphi|.
\end{equation}
Following the same approach, we can estimate the higher-order term in \autoref{eq:rho_pass_m2} and shown that $P_{\mr{ud(4)|pass}}\rho_{\mr{ud(4),pass}} = O(|\varphi|^2 p^4)$. 

We remark that, \autoref{eq:Dtrpass_approx} resembles the results in the single-rotation-scheme in \autoref{eq:Dtr_rho_pass}. The major difference lies in the value of the undetectable error $P_{\mr{ud(2)}}$. The exact value of $P_{\mr{ud}(2)}$ is given by
\begin{equation} \label{eq:PrUd2_true}
P_{\mr{ud(2)}} = \sum_{i=1}^k \Pr(Z_{2i} Z_{2i+1}, I_{\backslash{2i,2i+1}}),
\end{equation}
where $\Pr(Z_{2i} Z_{2i+1}, I_{\backslash{2i,2i+1}})$ indicates the probability where Pauli $Z$ error occurs on the qubit $2i$ and $2i+1$ and no detectable errors happen on the other qubits in the post-selection region defined in \autoref{fig:ProjectionScheme} during the whole procedure including the $\ket{+}_L$ state preparation, the implementation of the $R_{ZZ}(\theta)$ gates and the post-selection based on the syndrome measurement. For the convenience of analysis, we can estimate $P_{\mr{ud(2)}}$ by
\begin{equation} \label{eq:PrUd2}
P_{\mr{ud(2)}} \leq \sum_{i=1}^k \Pr(Z_{2i} Z_{2i+1}).
\end{equation}
Here, $\Pr(Z_{2i} Z_{2i+1})$ indicates the probability of Pauli $Z_{2i} Z_{2i+1}$ error. When $k$ becomes larger, the value of $\Pr(Z_{2i} Z_{2i+1}, I_{\backslash{2i,2i+1}})$ will be remarkably smaller than $\Pr(Z_{2i} Z_{2i+1})$. Therefore, \autoref{eq:PrUd2} serves as an upper bound for $P_{\mr{ud(2)}}$.

We can check $\Pr(Z_{2i} Z_{2i+1})$ step-by-step:
\begin{enumerate}
\item The $\ket{+}_L$ state preparation is achieved by first prepare the tensor state $\ket{+}^{\otimes n}$ and then perform $d$ rounds of syndrome measurement to fix the value of $Z$-type syndrome. By carefully assigning the order of the CNOT gates for the ancilla-based syndrome measurement, the $Z$-type hook error will be vertical to the logical $Z_L$ operation. As a result, the $ZZ$-type Pauli error on $\mr{Supp}(Z_L)$ will occur with the probability of $O(p^2)$.
\item During the implementation of the $R_{ZZ}(\theta)$ gates, by leveraging the 1-FT $R_{ZZ}(\theta)$ gate design in \autoref{sec:RZZgate4112}, the $ZZ$-type Pauli error will occur with the robability of $O(p^2)$.
\item During the post-selection based on $3$ rounds of syndrome measurement, similar to the $\ket{+}_L$ state preparation, we can avoid the hook error to propagate along $\mr{Supp}(Z_L)$. The $ZZ$-type Pauli error will occur with the robability of $O(p^2)$.
\end{enumerate}
We remark that, two independent first-order errors during the state preparation, the implementation of $R_{ZZ}(\theta)$ gates and the syndrome measurements can also lead to undetectable $Z_{2i}Z_{2i+1}$-type Pauli errors. This occurs with the probability of $O(p^2)$. However, if the two independent errors do not occur on a pair of locations $2i$ and $2i+1$ with $i=0,1,...,k-1$, this weight-$2$ error will be detectable by the syndrome measurement. 
Taking the number of rotation gates $k$ into account, we have $\Pr(Z_{2i} Z_{2i+1}) = O(p^2)$ and $P_{\mr{ud(2)}}= O(k p^2)$. From \autoref{eq:Dtrpass_approx}, we have
\begin{equation}
D_{\mr{tr}}(\rho_{\mr{pass}}, \rho_{\mr{id,pass}}) \approx O(k|\varphi| p^2).
\end{equation}
Here, $k$ is the number of implemented physical rotation gates $R_{ZZ}(\theta)$. When $k$ is a constant, the trace distance between the prepared noisy state $\rho_{\mr{pass}}$ and $\ket{r_\varphi}_L$ is $O(|\varphi|p^2)$.

\end{proof}

\section{Ancilla-based logical-rotation gate by the dual-rail qubit} \label{sec:App_dual_rail}

In \autoref{fig:RZZ_dual_rail}, we discuss an alternative approach to realize the ancillary-based logical rotation gate $R_{ZZ}(\varphi)$ based on the dual-rail encoding on two qubits~\cite{kubica2023erasure,teohDualrailEncodingSuperconducting2023b}. Consider two transmons each with the energy levels of $\ket{g}, \ket{e}$, $\ket{f}$ and so on. The dual-rail subspace is defined by
\begin{equation}
\ket{0}_L = \ket{ge},\quad \ket{1}_L = \ket{eg}.
\end{equation}
As shown in \autoref{fig:RZZ_dual_rail}(a), when the local qubit relaxation error or excitation error happens, the state of the system will leak out from the code space, which is detectable based on the parity number measurement. 

In our $R_{ZZ}(\varphi)$ gate design in \autoref{fig:RZZ_dual_rail}(b), since the dual-rail qubit acts as an destructive ancilla, we do not need the global parity number measurement. Instead, the parity measurement is done with the logical measurement together at the end of the circuit by two independent qubit measurements.

\begin{figure}[htbp]
    \centering
    \includegraphics[width=0.5\textwidth]{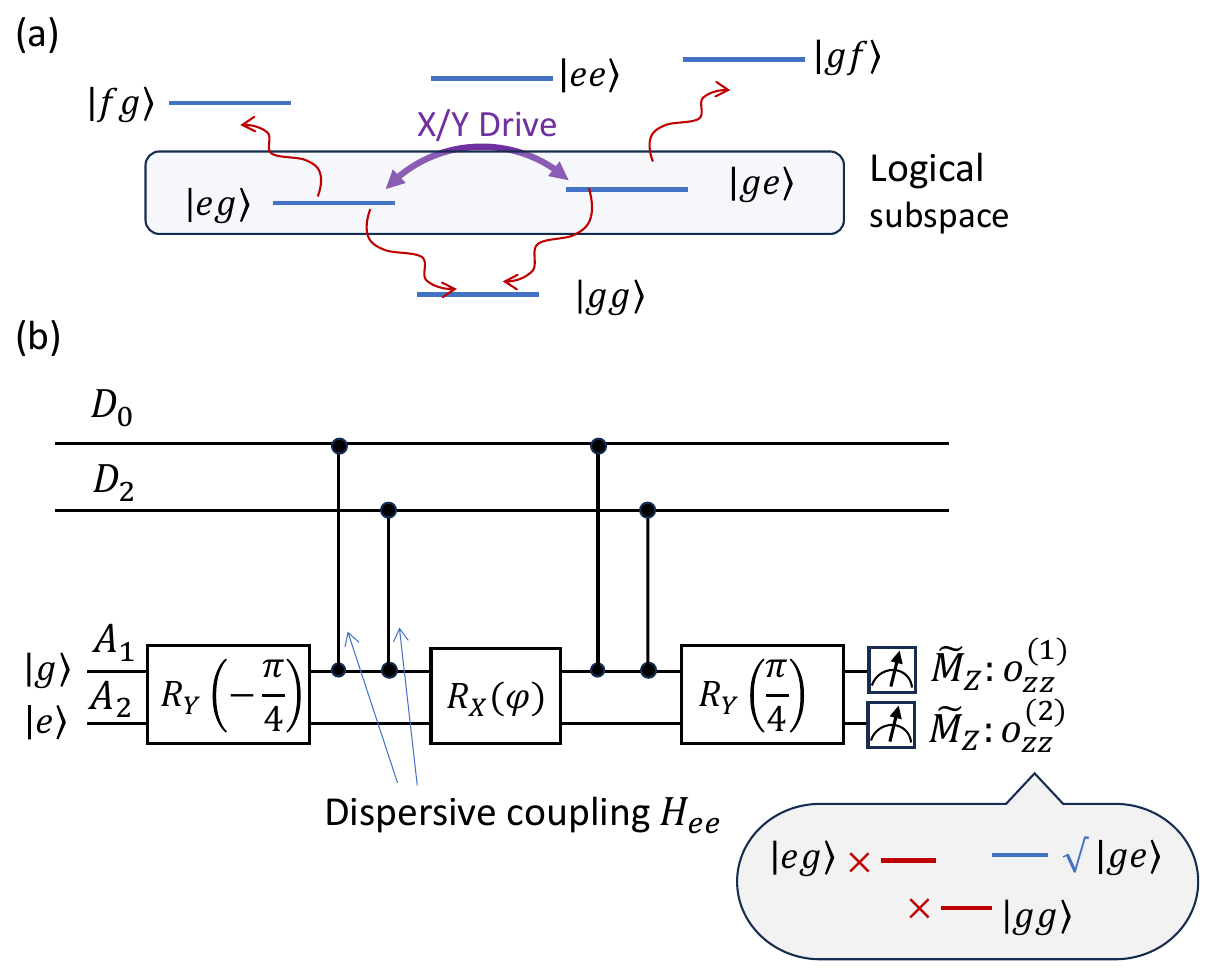}
    \caption{(a) The dual-rail subspace of two qubits. We need the $X$ and $Y$ drive in the dual-rail subspace. (b) 1-fault-tolerant rotation gate $R_{ZZ}(\varphi)$ by the dual-rail ancilla.}
    \label{fig:RZZ_dual_rail}
\end{figure}

Similar to the proof in Appendix~\ref{ssc:AppPropZZancilla}, one can show that when one dephasing, relaxation or excitation error happens during the $R_{ZZ}(\varphi)$ gate in \autoref{fig:RZZ_dual_rail}(b), it will either introduce a detectable error on the ancillary qubits $A_1$ and $A_2$ or a detectable error on the data qubits $D_0$ and $D_2$ which will be captured by the following $[[4,1,1,2]]$ parity checks. Notably, the first-order dephasing error on the dual-rail qubit will also be detected by the ancillary measurements. This makes the 

A good property in this design is that we only need regular dispersive coupling among two qubits $H_{ee} = \chi \ket{e}\bra{e} \otimes \ket{e}\bra{e}$. Meanwhile, our design is robust to the first order of the dephasing error on the dual-rail qubit: as a result, the usual bottleneck of the large dephasing noise on the dual-rail transmon~\cite{kubica2023erasure} will not lead to a severe problem in our design. We will leave the detailed error analysis of the rotation gate built from the dual-rail qubits and the corresponding resource estimation to the future study.

\section{Projection scheme with extended rotation gates} \label{sec:App_extend_rotated}

\autoref{fig:ExtendZZZ}(a) shows how to extend a $R_{ZZ}(\varphi)$ gate to a $R_{ZZZ}(\varphi)$ gate non-transversally by two CNOT gates. Suppose a depolarization noise model for two CNOT gates and a dispersive $R_{ZZ}(\varphi)$ gate as studied in \autoref{sec:RZZgate4112},  it is then easy to show that the Pauli-$ZZZ$ error occurs with a probability of $O(p^2)$ after the whole $R_{ZZZ}(\varphi)$ gate in \autoref{fig:ExtendZZZ}(a). Therefore, following the analysis in Appendix~\ref{ssc:AppProofProjectionFT}, the projection scheme is still 1-FT. 

\begin{figure}[htbp]
    \centering
    \includegraphics[width=0.35\textwidth]{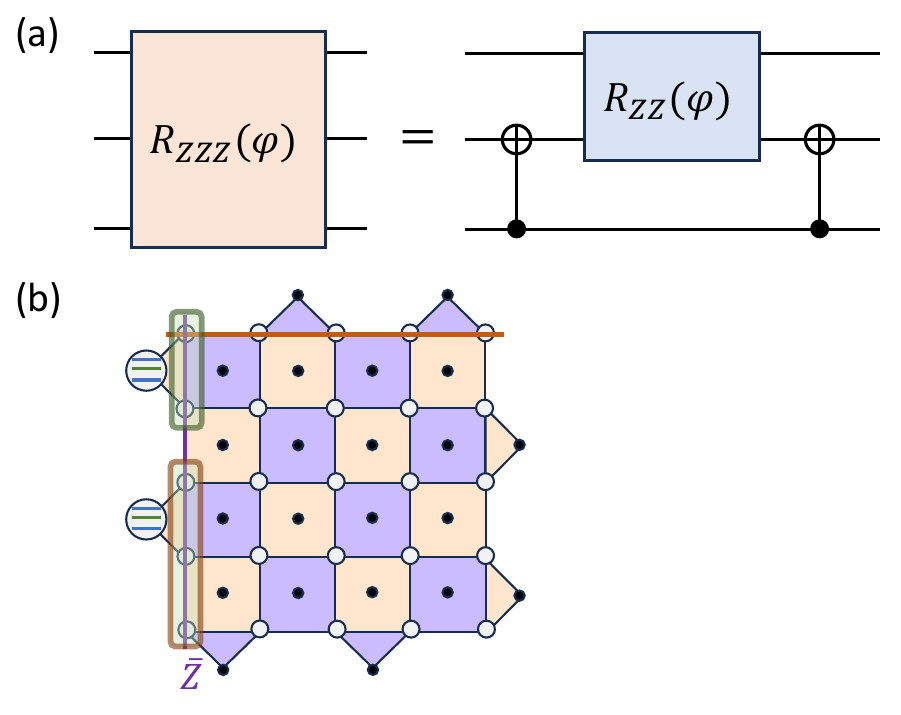}
    \caption{(a) The extension of $R_{ZZ}(\varphi)$ gate to a $R_{ZZZ}(\varphi)$ gate. (b) Projection scheme on an odd-distance rotated surface code.}
    \label{fig:ExtendZZZ}
\end{figure}

With the extended $R_{ZZZ}(\varphi)$ gate construction, we can then build the projection scheme with an odd-distance rotation surface code shown in \autoref{fig:ExtendZZZ}(b). We can also use the $R_{ZZZ}(\varphi)$ gates as the basic unit in the projection scheme for a large-distance surface codes such as the ones with $d=12$ or $18$.

\section{Numerical simulation details} \label{sec:AppNum}

\subsection{[[4,1,1,2]] code simulation} \label{ssc:App4112Num}

In \autoref{fig:Num4112coherent} in \autoref{sec:RZZgate4112}, we have performed a fully-coherent circuit simulation of the preparation of the $\ket{r_\varphi}_L$ state on the $[[4,1,1,2]]$ code. The simulation is performed based on the circuit in \autoref{fig:4112SimCircuit} with the Qiskit simulator~\cite{javadiabhari2024qiskit}. We first prepare the logical $\ket{+}_L$ state 1-fault-tolerantly based on the circuit in \autoref{fig:4112QED}(a), then apply the $R_{ZZ}(\varphi)$ gate based on the ancilla-free dispersive Hamiltonian between the data qubit $D_0$ and $D_2$ or the ancilla-based approach in \autoref{fig:RZZgate_ancilla}. After that, we perform ZX-QED based on the circuit in \autoref{fig:4112QED}(b). The final 4-qubit state $\rho$, up to an ideal projector $\Pi_0$, should be close to the ideal ancillary state $\ket{r_\varphi}$. We then calculate the trace distance by
\begin{equation}
D_{\mr{tr}}(\rho_{\Pi_0}, \ket{r_\varphi}_L),
\end{equation}
where $\rho_{\Pi_0} = \Pi_0 \rho \Pi_0/ \tr(\rho \Pi_0)$.

\begin{figure}[htbp]
    \centering
    \includegraphics[width=0.6\textwidth]{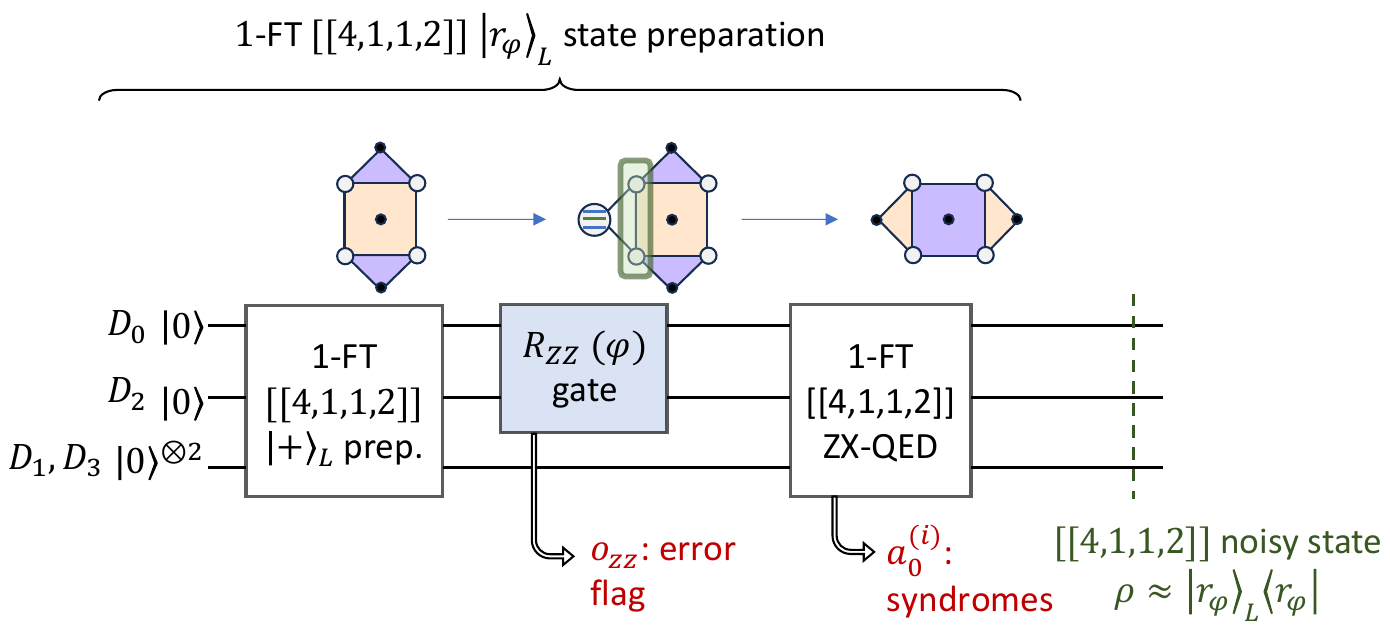}
    \caption{(a) The extension of $R_{ZZ}(\varphi)$ gate to a $R_{ZZZ}(\varphi)$ gate. (b) Projection scheme on an odd-distance rotated surface code.}
    \label{fig:4112SimCircuit}
\end{figure}

During the Qiskit simulation, after all the state preparation, idling gates, single-qubit gates and before all the single-qubit measurements, we append the following single-qubit depolarizing channel with the physical error rate $p$,
\begin{equation}
\mc{E}_{1}(\rho) = (1-p) \rho + \frac{p}{3} (X\rho X + Y\rho Y + Z\rho Z ).
\end{equation}
After all the $\mr{CNOT}$ gates, we append the following two-qubit depolarizing channel with the physical error rate $p$,
\begin{equation}
\mc{E}_{2}(\rho) = (1-p) \rho + \frac{p}{15} \sum_{E\in\{I,X,Y,Z\}^{\otimes 2}\backslash I } E \rho E.
\end{equation}
Finally, for the $R_{ZZ}(\varphi)$ gate, we first perform a Lindbladian-based simulation of the gate on QuTip~\cite{johansson2012qutip} and then export the two-qubit Kraus operators from QuTip to the Qiskit circuit simulation. 

For the ancilla-free dispersive coupling Hamiltonian gate on the data qubits $D_0$ and $D_2$, the Lindbladian $\mc{L}_{ZZ}(\bullet)$ is
\begin{equation}
\mc{L}_{ZZ}(\bullet) = -i[ H_{ZZ}, \bullet ] + \left( \sum_{j=0,2} D[\sqrt{\gamma_{e\to g}} J_{e\to g}^{(j)} ] + D[\sqrt{\gamma_{\phi}} J_{\phi}^{(j)} ] \right) \bullet,
\end{equation}
where
\begin{equation} \label{eq:HZZ_Jeg_Jphi}
\begin{aligned}
H_{ZZ} &= -\chi Z_0 Z_2, \\
J_{e\to g}^{(j)} &= \ket{g}_j\bra{e}, \\
J_{\phi}^{(j)} &= \ket{e}_j\bra{e}.
\end{aligned}
\end{equation}

For the ancilla-based approach in \autoref{fig:RZZgate_ancilla}, we perform a 5-stage Lindbladian simulation: the first ancilla $R_Y(-\pi/4)$-drive, the first group of dispersive $ZZ$ gates, ancilla $R_X(\varphi)$-drive, the second group of dispersive $ZZ$ gates and the second ancilla $R_Y(-\pi/4)$-drive. During the dispersive $ZZ$ gates, we assume the following dispersive-coupling Hamiltonian,
\begin{equation}
H = H_{fe}^{(0)} + H_{fe}^{(2)} = \chi \ket{f}_A \bra{f} \otimes \ket{e}_{D_0}\bra{e} + \chi \ket{f}_A \bra{f} \otimes \ket{e}_{D_2}\bra{e}.
\end{equation}
During the $X$ drive on the ancillary $g$-$f$ qubit, we assume the following driving Hamiltonian 
\begin{equation}
H_{gf} = \Omega_{gf} (\ket{g}\bra{f} + \ket{f}\bra{g}).
\end{equation}
We assume that the two data qubits will experience the dissipative error of $J_{e\to g}$ and $J_{\phi}$ in \autoref{eq:HZZ_Jeg_Jphi} during the dispersive $ZZ$ gates, and the ancilla $R_X(\varphi)$-drive. The $g$-$f$ qubits will experience $J_{e\to g}$, $J_{\phi} = \ket{e}\bra{e} + 2\ket{f}\bra{f}$ and also the $f\to e$ relaxation $J_{f\to e} = \ket{e}_j\bra{f}$ with the dissipation rate $\gamma_{f\to e}$ during the whole 5-stage Lindbladian simulation. 

In the QuTip simulation, we have assume that the dispersive coupling $\chi= 2\pi\cdot 5$ MHz and the $g$-$f$ drive strength $\Omega = 2\pi\cdot 20$ MHz. This is based on the parameters for the $g$-$f$ qubits in Kubica~et~al.~\cite{kubica2023erasure}. We remark that, when we set $\eta = 2\pi\cdot 250\,\mr{MHz}, \Delta = 2\pi\cdot 500\,\mr{MHz}, g_c = 2\pi\cdot 30\,\mr{MHz}, \epsilon_d = 2\pi \cdot 30\,\mr{MHz}$ in Appendix~J in Ref.~\cite{kubica2023erasure}, we have that the $f$-$f$ coupling $\chi= 2\pi\cdot 4.8$ MHz and the drive $\Omega_{gf} = 2\pi\cdot 20.36$ MHz.

During the QuTip simulation, we set the dissipation rate $\gamma_{f\to e} = \gamma_{e\to g} = \gamma_{\phi}$. To align the noise during the QuTip simulation with the depolarization noise during the Qiskit simulation, we first consider the parameters in Ref.~\cite{reinholdErrorcorrectedGatesEncoded2020a} with $\gamma=2\cdot 10^4\, \mr{s}^{-1}$ and the depolarization noise rate $p=1\cdot 10^{-3}$ as the current achievable noise parameters, then fix their ratio $\gamma/p=2\cdot 10^{7}\, \mr{s^{-1}}$ to homogeneously adjust all the error parameters.

\subsection{Expansion scheme simulation} \label{ssc:AppExpandNum}

In \autoref{fig:expansion_num} in \autoref{sec:expansion}, we have performed a circuit-level simulation of the expansion scheme to estimate the average infidelity of the quantum channel to prepare $\ket{r_\varphi}$ in \autoref{fig:expansion}. The simulation is performed based on the Stim package~\cite{gidney2021stim}. 

Unlike the Qiskit-based simulation, Stim is a Clifford circuit simulation platform where only Pauli-type noise channel is allowed. In this case, we set the rotation angle $\varphi=0$ or $\pi/4$ to ensure the Clifford property of the whole circuit. Moreover, for the two-qubit noise channel introduced by the $R_{ZZ}(\varphi)$, which is still obtained by QuTip similar to the ones introduced in Appendix~\label{ssc:App4112Num}, we only keep the Pauli part of the noise and load the two-qubit Pauli noise channel into the Stim simulation. This is reasonable since the physical coherent noise will be suppressed by the later syndrome prjection: in Ref.~\cite{beale2018qec}, the authors have shown that as long as the physical noise is local, the off-diagonal physical coherent noise will be exponentially suppressed by the distance $d$ of the QEC code. Consequently, the off-diagonal terms will not affect the overall performance of the expansion scheme when we finally expand the code to the one with a large distance.

Recall that we only estimate the noise introduced by the QED + post-selection stage in \autoref{fig:expansion}. During the whole circuit, we first prepare the $\ket{r_\varphi}$ state on the $[[4,1,1,2]]$ code and then expand it to a $d=3$ rotated surface code and perform three rounds of surface code syndrome measurement for the further post-selection. We then estimate the logical error rate of the prepared surface code state.

During the simulation of the expansion circuit, to accurately estimate the average channel infidelity, we evaluate the logical error rate for input states $\ket{0}_L$, $\ket{+}_L$, and $\ket{+i}_L$, respectively. For $\varphi = 0$, we measure the final surface code state in the logical $Z$, $X$, and $Y$ bases; for $\varphi = \pi/4$, we measure in the logical $Z$, $Y$, and $-X$ bases. This approach enables us to estimate the Pauli infidelity of the noise introduced by the whole expansion circuit corresponding to the $Z$, $X$, and $Y$ bases, respectively. The logical $Z$ and $X$ measurement can be performed fault-tolerantly by a direct transversal physical $Z$ and $X$ measurement, respectively. The logical $Y$ basis, however, cannot be fault-tolerantly measured. For the error estimation in the $Y$ basis, we add an extra round of noiseless $d=3$ surface code QED, then measure the $Y$-basis logical error rate non-fault-tolerantly. 

Here, we focus on the estimation of the average channel infidelity. For a fair comparison, we may want to estimate the trace distance of the final prepared state $\rho$ to the ideal state $\ket{r_\varphi}_L$. Based on the results in Ref.~\cite{beale2018qec}, the logical noise channel will be close to a Pauli channel where the off-diagonal coherent terms are suppressed by repetitive syndrome measurements. For an $n$-qubit Pauli channel $\mc{N}_n$, its diamond norm distance $\epsilon_\diamond(\mc{N}_n)$ to an identity channel and its average infidelity $r(\mc{N}_n)$ is related by~\cite{beale2018qec}
\begin{equation}
\epsilon_\diamond(\mc{E}_n) = (1 + \frac{1}{2^n}) r(\mc{E}_n).
\end{equation} 
For a single-qubit channel $\mc{N}$, we can then estimate the trace distance of $D_{\mr{tr}}(\rho, \ket{r_\varphi}_L)$ by $1.5\cdot r(\mc{N})$.

\subsection{Projection scheme simulation} \label{ssc:AppProjNum}

In \autoref{fig:projection_num}, \autoref{fig:projection_num} and \autoref{fig:SuccProb} in \autoref{sec:projection}, we perform the circuit-level simulation of the projection scheme on a $d\times (d+1)$ surface code where $d=m\cdot k$ using the Stim package~\cite{gidney2021stim}. Here, $k$ is the number of physical rotation gates, while $m$ is the number of qubits involved during the implementation of the rotation gates. We have $m=2$ for the ordinary $R_{ZZ}(\varphi)$ gate, and $m=3$ for the $R_{ZZZ}(\varphi)$ gate with the expansion introduced in Appendix~\ref{sec:expansion}.

We focus on the successful probability and error estimation of the preparation of $\ket{r_\varphi}_L$ in the projection scheme. During the scheme, we need to perform many rotation gates $R_{ZZ}(\theta)$ with the angle $\theta\approx \varphi^{1/k}$. Since the rotation gates are non-Clifford gates, we cannot simulate the whole circuit efficiently. However, following the approach in Ref.~\cite{toshio2024practical}, we can still estimate the successful probability and the trace distance efficiently based on the Monte-Carlo sampling of a set of Clifford circuits. 

Recall from Appendix~\ref{ssc:AppProofProjectionFT} that the post-selected noisy state in the projection scheme can be written as
\begin{equation}
\rho_{\mathrm{pass}} \approx P_{\mathrm{id|pass}}, \rho_{\mathrm{id,pass}} + P_{\mathrm{ud(2)|pass}}, \rho_{\mathrm{ud(2),pass}} + O(|\varphi|^2 p^4),
\end{equation}
where $\rho_{\mathrm{id,pass}} = \ket{r_\varphi}\bra{r_\varphi}$ is the ideal rotation state and $\rho_{\mathrm{ud(2),pass}} = \ket{r_{\varphi_1}}\bra{r_{\varphi_1}}$ is the leading-order wrong rotation state. $P_{\mathrm{id|pass}}$ and $P_{\mathrm{ud(2)|pass}}$ are the Bayesian probability of no error happens and undetectable error happens conditioning on the passing of the QED stage. Since the value of $\varphi$ and $\varphi_1$ are know in advance, the major idea of our numerical estimation is to estimate the Bayesian probabilities $P_{\mathrm{id|pass}}$ and $P_{\mathrm{ud(2)|pass}}$ and then calculate the successful probability and the trace distance based on the analytical formulas.

Recall that we can decompose the effect of multi-rotation gate as
\begin{equation}
\prod_{i=1}^k R_{ZZ,i}(\theta) = \sum_{b=0}^{2^k} u_{|b|} (Z_{(2)})^b
\end{equation}
where $b$ is a $k$-bit string, $|b|$ denotes the weight of $b$, $u_w:= i^w \sin^w \theta \cos^{k-w}\theta$ and
\begin{equation}
(Z_{(2)})^b := \prod_{i: b_i=1} Z_{2i} Z_{2i+1}.
\end{equation}
Applying the multi-rotation gate on $\ket{+}_L$, we have
\begin{equation} \label{eq:MultiRot_Plus}
\begin{aligned}
\prod_{i=1}^k R_{ZZ,i}(\theta) \ket{+}_L &= \sum_{b=0}^{2^k} u_{|b|} (Z_{(2)})^b \ket{+}_L \\
&= \sum_{b=0}^{2^{k-1}} \left( u_{|b|} (Z_{(2)})^b + u_{|\bar{b}|} (Z_{(2)})^{\bar{b}} \right) \ket{+}_L \\
&=: \sum_{b=0}^{2^{k-1}} \ket{\psi_b}.
\end{aligned}
\end{equation}
Here, $\bar{b}$ is the bitwise complement of $b$. In the second line of \autoref{eq:MultiRot_Plus}, we pair the terms with complement $b$ values. We remark that, since $(Z_{(2)})^{\bar{b}}(Z_{(2)})^b = Z_L$, the summand in \autoref{eq:MultiRot_Plus} specified by $b$ is in the same syndrome subspace with the one specified by $\bar{b}$. Specifically, $\ket{\psi_{00...0}}\propto \ket{r_\varphi}$ is the ideal rotation state. In the noiseless quantum circuit, when performing the syndrome projection, we will obtain the state $\ket{\psi_b}$ with the probability,
\begin{equation} \label{eq:Prb}
\Pr(b) = |u_{|\mr{b}|}|^2 + |u_{|\bar{\mr{b}}|}|^2.
\end{equation}

Now, suppose a Pauli error $E$ occurs during the state preparation, $R_{ZZ}(\theta)$ gate implementation or the syndrome measurement procedure. The erroneous state becomes
\begin{equation}
E \prod_{i=1}^k R_{ZZ,i}(\theta) \ket{+}_L = \sum_{b=0}^{2^{k-1}} E \ket{\psi_b}.
\end{equation}
The new state $E\ket{\psi_b}$ will then belongs to a different syndrome subspace of the surface code as long as the weight of $E$ is less than $d$. We remark that, the set of quantum state $\{E\ket{\psi_b}\}$ preserves the orthogonality of $\{\ket{\psi_b}\}$ and the sampling probability $\Pr(b) = |u_{|\mr{b}|}|^2 + |u_{|\bar{\mr{b}}|}|^2$. The only effect is that the syndrome value have changed from the ones related to the string of $(Z_{(2)})^b$ to the ones related to $E (Z_{(2)})^b$. As a result, the error sampling of $E$ and the sampling of bit-string $b$ by the syndrome measurement can be regarded as independent procedure.

Thanks to the independence of the sampling of $E$ and the bit-string $b$, we can estimate the Bayesian probabilities of $P_{\mathrm{id|pass}}$ and $P_{\mathrm{ud(2)|pass}}$ by the following procedures,
\begin{enumerate}
\item Prepare the $\ket{+}_L$ state by a single round of surface code syndrome measurement on the state $\ket{+}^{n}$;
\item Instead of implementing $k$ rotation gates $R_{ZZ}(\theta)$, sample a single bit-string $b\in \{0,1\}^{2^{d-1}}$ with probability $\Pr(b)$ in \autoref{eq:Prb} then apply $(Z_{(2)})^b$ on the qubits;
\item Apply the error channels introduced by the $k$ rotation gates $R_{ZZ}(\theta)$, which is obtained from the QuTip simulation of the lindbladian procedure;
\item Perform three rounds of surface code syndrome measurement. Determine whether the circuit passes the post-selection or not.
\end{enumerate}
Here, instead of implementing $k$ rotation gates $R_{ZZ}(\theta)$ directly, we sample the bit string $b$ induced by these gates based on \autoref{eq:MultiRot_Plus}. This is reasonable because the Pauli errors $E$ generated during the whole circuit and the string $(Z_{(2)})^b$ have independent effect onto the syndrome measurement results.

Suppose we have run $N$ samples and $N_{\mr{pass}}$ rounds pass the post-selection. We can then estimate $P_{\mr{pass}}$, $P_{\mathrm{id|pass}}$ and $P_{\mathrm{ud(2)|pass}}$ by
\begin{equation}
\begin{aligned}
P_{\mr{pass}} &\approx \frac{ N_{\mr{pass}} }{N}, \\ 
P_{\mr{id|pass}} &\approx \frac{ N_{\mr{wt(b)=0, pass}} }{ N_{\mr{pass}} },\\
P_{\mr{ud(2)|pass}} &\approx \frac{ N_{\mr{wt(b)=1, pass}} }{ N_{\mr{pass}} },\\
\end{aligned}
\end{equation}
where $N_{\mr{wt(b)=0, pass}}$ and $N_{\mr{wt(b)=1, pass}}$ indicates the number of samples with $\mr{wt}(b)=0$ and $\mr{wt}(b)=1$ among $N_{\mr{pass}}$ rounds which passes the post-selection. $\mr{wt}(b)$ is the weight of the $k$-bit string $b$. The trace distance can then be estimated by \autoref{eq:Dtr_rho_pass_m2} with the Bayesian probabilities.

\section{Details for the spacetime cost estimation of different magic-state preparation methods} \label{sec:AppCostEst}

We introduce the details of the spacetime cost estimation for the non-Clifford gates used in the second-order Trotter-based Hamiltonian simulation of the Heisenberg model in \autoref{eq:Heisenberg}.

\subsection{Count the number of T gates in the Hamiltonian simulation} \label{ssc:AppCountT}

For the $T$-state based methods, including the magic state distillation and cultivation, we first estimate the minimum number of Trotter segments $\nu$ required to implement the second-order Trotter formula. Based on the results of Proposition~M.1 and Fig.~3 in Childs~et~al.~\cite{childs2021theory}, when the commutator error analysis is taken into account, the minimum number of Trotter segment $\nu$ when the accuracy requirement is $1\cdot 10^{-3}$ can be numerically fitted by
\begin{equation}
\nu = e^{1.85} \cdot n^{0.27} t^{1.25},
\end{equation}
Now, we decompose the Heisenberg Hamiltonian to the terms with even and odd indicies,
\begin{equation}
\begin{aligned}
H &= A + B, \\
A &= \sum_{j:\mr{odd}} (X_j X_{j+1} + Y_j Y_{j+1} + Z_j Z_{j+1}) + \sum_{j:\mr{odd}} h_j Z_j, \\
B &= \sum_{j:\mr{even}} (X_j X_{j+1} + Y_j Y_{j+1} + Z_j Z_{j+1}) + \sum_{j:\mr{even}} h_j Z_j, \\
\end{aligned}
\end{equation}
Based on the $4$th-order Trotter formula, we need to implement $6$ rotation gates for each Pauli terms in $A$ and $5$ rotation gates for each Pauli terms in $B$. The total number of rotation gates in a Trotter segment of the Heisenberg Hamiltonian can be estimated by 
\begin{equation}
N_{R_Z(\varphi)} = 5.5\cdot 4 \cdot N = 22N.
\end{equation}
Here, $N$ is the number of spins. 

For each rotation gate $R_{Z_L}(\varphi)$, we need to compile it to a sequence of $H$, $S$ and $T$ gate. Following currently the most resource-efficient $T$-gate compiling approach in Ref.~\cite{bocharov2015efficient}, we need
$$ c_T = 1.149 \log_2(1/\epsilon_c) + 9.2 $$
$T$ gates to achieve a compiling accuracy of $\epsilon_c$. To ensure the accuracy of the final simulation results, we now set $\epsilon_c = 1\times 10^{-11}$. As a result, $c_T \simeq 51.186$.

To summarize, consider a digital Hamiltonian simulation of the Heisenberg model with $N$ spins,  evolution time $T$ and accuracy requirement of $\epsilon = 1\cdot 10^{-3}$, the number of $T$ gates can be estimated by
\begin{equation}
N_{T} = \nu\cdot N_{R_Z(\varphi)}\cdot c_T.
\end{equation}

\subsection{Count the spacetime cost of magic state distillation and cultivation}


For the magic state distillation protocol, we consider the usage of the concatenated $(15\text{-to-}1)^4_{13,5,5}\times (20\text{-to-}4)_{27,13,15}$ protocol designed in Ref.~\cite{Litinski2019magic} which can generate a $T\ket{+}_L$ state with an error of $2.6 \times 10^{-11}$ when the gate error is $1\cdot 10^{-3}$. Here, the subscript $\{13,5,5\}$ indicates the $X$-distance $d_X$, $Z$-distance $d_Z$ of the level-1 surface code and the measurement distance $d_m$ during the distillation, respectively. The superscript indicates the number of level-$1$ distillation blocks used in the protocol.

Based on the estimation summarized in Table~1 in Ref.~\cite{Litinski2019magic}, the spacetime cost characterized by the qubitcycles of the $(15\text{-to-}1)^6_{17,7,7}\times (20\text{-to-}4)_{23,11,13}$ is $1840000$.

For the magic state cultivation protocol~\cite{gidney2024cultivation}, we consider the estimation in Fig.~1 of \cite{gidney2024cultivation}, where we can prepare a $T$-state with an end-to-end error rate of $2\cdot 10^{-9}$ and a discard rate of $99\%$. The spacetime cost of the cultivation protocol is about $60000$ from Fig.~1. 
To make a fair comparison between the capacity of the magic state cultivation and our protocol, we assume that the error of $2\cdot 10^{-9}$ on the $T$-state will become a logical Pauli noise on the $T$ gate after injection and hence can be mitigated by the Pauli error cancellation techniques~\cite{temme2017error,endo2018practical}. As a result, we can accurately get the outcome of the entire circuit when the number of $T$ gate is below $1/(2\cdot 10^{-9}) = 5\cdot 10^8$.

\subsection{Count the spacetime cost of the projection scheme}

We estimate the spacetime cost of the projection scheme based on the $(18,3,6)$-scheme, i.e., on the surface code with a distance of $18$, we implement 6 multi-rotation gates with weight-$3$ $ZZZ$-rotation gate $R_{ZZZ}(\theta)$ realized by the extension approach in Appendix~\ref{sec:App_extend_rotated}. Recall from \autoref{fig:ProjectionScheme}(b) that during the post-selection stage of the scheme, we perform 4 rounds of syndrome measurements and post-selection based on the rightmost three columns of syndromes in \autoref{fig:ProjectionScheme}(a). If the post-selection stage passes, we then perform $18-4=14$ rounds of extra syndrome measurements to corrent the remaining error. 

We remark that, after the first round of syndrome measurement, we will apply the $R_{ZZ}(\theta)$ gates on the support of the logical $Z$ operator. If implemented by the direct dispersive coupling, this gate is usually quick: for a typical $\theta$ value of $0.1$ rad and dispersive coupling strength of $\chi = 2\pi\cdot 5$ MHz, a $R_{ZZ}(\theta)$ gate can be done within $31.8$ ns. As a result, we will ignore the time cost of the $R_{ZZ}(\theta)$ gates.

Suppose we perform a repeat-until-success preparation of the resource state: if it pass the QED stage, we keep perform the remaining syndrome measurements to complete the state preparation; otherwise we restart the state preparation procedure. When the successful probability of the QED stage is $p_{\mr{suc}}$, the spacetime cost, i.e., the qubitcycle $Q$ can be estimated by the following formula,
\begin{equation} \label{eq:Q}
\begin{aligned}
&Q= p_{\mr{suc}} (18\cdot 18\cdot 18) + (1- p_{\mr{suc}}) (18\cdot 18\cdot 4 + Q), \\
\Rightarrow\; & Q = 4536 + \frac{1296}{p_{\mr{suc}}}
\end{aligned}
\end{equation}
When the physical error rate is $1\cdot 10^{-3}$, from \autoref{fig:SuccProb} we have that the successful probability of $(18,3,6)$-scheme is $5.4\%$ for $\varphi = 1\times 10^{-3}$. In this case, we have $Q = 28536$.

Recall that the rotation-state-injection is probabilistic: when the rotation injection fails, we need to perform the injection of $R_Z(2\varphi)$ gate, whose resource state preparation owns a smaller successful probability and larger spacetime cost $Q$. 
In practice, when the rotation-state-injection fails many times and the target angle is amplified to a value which is larger than $0.2$ rad, we can always introduce a wrapping procedure to reduce its value to the region of $[-\frac{\pi}{16}, \frac{\pi}{16})$ by applying a $T$ gate, which can be implemented by the magic state cultivation and injection.

Now, let's estimate the overall spacetime cost of the projection scheme when the target angle $\varphi$ is $1\cdot 10^{-3}$, when the RUS injection is also taken into account. The overall spacetime cost is given by
\begin{equation} \label{eq:Qtot}
Q_{\mr{tot}}(\varphi) = \frac{1}{2} Q_1(\varphi) + \frac{1}{4} Q_2(\varphi) + \frac{1}{8} Q_3(\varphi) +...,
\end{equation}
where $Q_k(\varphi)$ is the accumulated spacetime cost up to the $k$th trial with $Q_1(\varphi) = Q(\varphi)$ and
\begin{equation} \label{eq:Qkp1}
Q_{k+1}(\varphi) = Q_{k}(\varphi) + Q(\Lambda(2^{k-1}\varphi)).
\end{equation}
Here, $\Lambda(\varphi')$ is the wrapping function to wrap the angle value $\phi'$ into the region of $[-\frac{\pi}{16}, \frac{\pi}{16})$. From \autoref{eq:Q} we know that the single-round value of $Q(\varphi')$ heavily depends on the angle $\varphi'$. We remark that, we also need to take the spacetime cost of the magic-state-cultivation into account once the wrapping procedure is required when $|\varphi_k|>0.2$ rad.

\begin{figure}[htbp]
    \centering
    \includegraphics[width=0.5\textwidth]{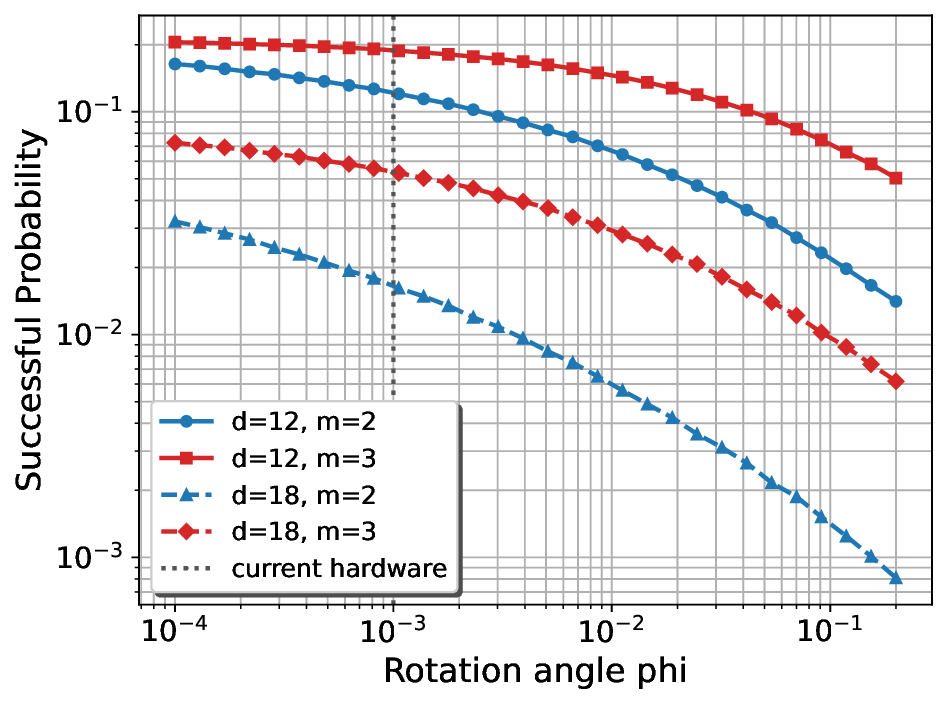}
    \caption{Successful probability of the $(m,k,d)$-projection scheme with different target angles $\varphi$.}
    \label{fig:SuccProbVSangle}
\end{figure}

For an accurate estimation of the spacetime cost, we now estimate the successful probability of the $(18,3,6)$-scheme with $p=1\cdot 10^{-3}$ and different target angle values $\varphi$, shown in \autoref{fig:SuccProbVSangle}. We can see that, when the rotation angle $\varphi$ increases from $1\cdot 10^{-4}$ to $2\cdot 10^{-1}$, the successful probability of the $(18,3,6)$-scheme decreases from $7.24\%$ to $0.612\%$. Based on the values of $p_{\mr{suc}}(\varphi)$ in \autoref{fig:SuccProbVSangle} and \autoref{eq:Qtot} and \autoref{eq:Qkp1}, we can numerically estimate $Q_{\mr{tot}}(\varphi)$ with $\varphi = 1\cdot 10^{-3}$. 
$$Q_{\mr{tot}}(\varphi=1\cdot 10^{-3})  = 70415. $$

We remark that, here we do not take the cost of coherent error cancellation introduced in \autoref{ssc:Prob_Coher_EC} into account, since the probability to implement the cancellation $P_L$ is very small compared to the other cost.

\end{appendix}

\end{document}